\documentclass[graybox]{svmult}

\usepackage{type1cm}        % activate if the above 3 fonts are
\usepackage{makeidx}         % allows index generation
\usepackage{graphicx}        % standard LaTeX graphics tool
\usepackage{multicol}        % used for the two-column index
\usepackage[bottom]{footmisc}% places footnotes at page bottom

\usepackage{newtxtext}       % 
\usepackage{newtxmath}       % selects Times Roman as basic font

\makeindex             % used for the subject index
\usepackage{undertilde}

\usepackage{color}

\usepackage[textheight=25 cm, textwidth=16cm, headheight=10pt, headsep=20pt, footskip=10pt, left=1.5cm, right=1.5cm]{geometry}
\usepackage{hyperref}
\hypersetup{pdfborder=0 0 0}

\def\bC {\mathbf{C}}

\def\bN {\mathbf{N}}

\def\bR {\mathbf{R}}

\def\bT {\mathbf{T}}

\def\bZ {\mathbf{Z}}

\def\cD {\mathcal{D}}

\def\cH {\mathcal{H}}

\def\cP {\mathcal{P}}

\def\eps {{\epsilon}}

\def\L {{\Lambda}}

\def\la {\langle}
\def\ra {\rangle}

\newcommand{\Tr}{\operatorname{trace}}

\newcommand{\ba}{\begin{aligned}}
\newcommand{\ea}{\end{aligned}}

\newcommand{\be}{\begin{equation}}
\newcommand{\ee}{\end{equation}}

\newcommand{\vp}{\varphi}

\newcommand{\ec}[1]{\varphi^{#1}}

\newcommand{\C}{\mathbb{C}}
\newcommand{\R}{\mathbb{R}}
\newcommand{\Z}{\mathbb{Z}}
\renewcommand{\P}{\mathbb{P}}
\newcommand{\boM}{\mathcal{M}}

\newcommand{\inv}{^{-1}}

\newcommand{\pbz}{\psi^\alpha_z}

\newcommand{\freq}{\omega}

\def\bea{\begin{eqnarray}}
\def\eea{\end{eqnarray}}

\newcommand{\adir}[2]{|#1\rangle\langle #2| }
\newcommand{\com}[1]{{#1^c}}
\newcommand{\ccomp}[1]{\mathcal C(#1)}
\begin{document}

\title*{\Large Species of spaces}

%\author[T. Paul]{\Large Thierry Paul}
%\address[T.P.]{CNRS \& LJLL Sorbonne Université 4 place Jussieu 75005 Paris, France France}
%\email{thierry.paul@upmc.fr}

\author{Thierry Paul}
\institute{LYSM – Laboratoire Ypatia de Sciences Mathématiques \& LJLL – Laboratoire Jacques-Louis Lions,  
%Sorbonne Université 4 place Jussieu 75005 Paris France,
\email{thierry.paul@sorbonne-universite.fr}}

%\begin{abstract}
%\Large
%
%\end{abstract}
%\LARGE
\maketitle
\large
\hfill\textit{à André, Berthe, Georgette et Raoul, pour ce que je n'ai pas su}
%\begin{refguide}
\large
\vskip 3cm
 \newenvironment{changemargin}[2]{\begin{list}{}{%
\setlength{\topsep}{0pt}%
\setlength{\leftmargin}{0pt}%
\setlength{\rightmargin}{0pt}%
\setlength{\listparindent}{\parindent}%
\setlength{\itemindent}{\parindent}%
\setlength{\parsep}{0pt plus 1pt}%
\addtolength{\leftmargin}{#1}%
\addtolength{\rightmargin}{#2}%
}\item }{\end{list}}

\begin{changemargin}{7cm}{0cm}
``L’espace de notre vie n’est ni continu, ni infini, ni homogène, ni isotrope. Mais sait-on précisément o\`u il se brise, o\`u il se courbe, o\`u il se déconnecte et o\`u il se rassemble ? On sent confusément des fissures, des hiatus, des points de friction, on a parfois la vague impression que \c ca se coince quelque part, ou que \c ca \'eclate, ou que \c ca cogne. Nous cherchons rarement \`a en savoir davantage et le plus souvent nous passons d'un endroit  \`a l'autre, d'un espace  \`a l'autre sans songer  \`a mesurer,  \`a prendre en charge,  \`a prendre en compte ces laps d’espace. Le problème n'est pas d'inventer l'espace, encore moins de le r\'einventer (trop de gens bien intentionn\'es sont l\`a aujourd’hui pour penser notre environnement…), mais de l'interroger, ou, plus simplement encore, de le lire ; car ce que nous appelons quotidienneté n'est pas \'evidence, mais opacit\'e : une forme de c\'ecit\'e, une manière d’anesth\'esie.
C’est  \`a partir de ces constatations \'el\'ementaires que s'est d\'evelopp\'e ce 
%livre, 
[...]
journal d'un usager de l'espace. "

\hfill Georges Perec

\hfill\textit{Especes d'espaces}
%\hskip Seuil

\vskip 1cm
\end{changemargin}
\begin{changemargin}{5cm}{0cm}
%\LARGE
%{\bf {\it \noindent Alex's view of   quantum mechanics was an entanglement of geometry, analysis and algebra. From nested Hilbert spaces to von Neumann algebras of type II, he has constantly been interested in quantizing classical underlying phase-spaces.
%\noindent The lines that follow make the inverse journey and study traces of noncommutativity which remains at the classical limit. 
% They all testify to Alex's deep influence on me.}}
\end{changemargin}
%\LARGE

%\abstract{
{\bf {\it \noindent Alex's view of   quantum mechanics was an entanglement of geometry, analysis and algebra. From nested Hilbert spaces to von Neumann algebras of type II, he has constantly been interested in quantizing classical underlying phase-spaces.
\noindent The lines that follow make the inverse journey and study traces of noncommutativity which remains at the classical limit. 
 They all testify to Alex's deep influence on me.}
}
%}

\vskip 1cm
\abstract{
Classical limits of quantum systems are shown to lead to different conceptions of spaces different from the classical one underlying the process of quantization of such systems. The accent is put in situations where traces of noncommutativity, witness of an emblematic feature of quantum mechanise remains when the Planck constant vanishes, in the framework of noncommutative geometry. Complex canonical transformations, spin-statistics, topological quantum fields theory, long time semiclassical approximation and underlying chaotic dynamics are considered, together with a comparison/fusion of classical unpredictability with quantum indeterminism.
}
%\tableofcontents
\newpage
%\begin{refguide}
%\hfill\textit{à André, Berthe, Georgette et Raoul, pour ce que je n'ai pas su}
\section{Introduction: a space journey through the wonderful landscape of quantum and classical mechanics }\label{intro}
\vskip 1cm

{
%\centerline{\textit{version préliminaire}}
{\it \hfill ``Il ne faut pas peindre  ce que l'on voit,

\hfill puisqu'on le voit.

\hfill Il ne faut pas peindre ce que l'on ne voit pas,

\hfill puisqu'on ne le voit  pas.

\hfill {\textbf {Il  faut  peindre  
%que l'
qu'on ne voit  pas}}."}

\hfill Claude Monet (attr. \cite{godard})\footnote{A rough translation could be: `One shouldn't paint what one sees, since one sees it. One shouldn't paint what one don't see, since one doesn't see it. One should paint that one doesn't see."}.
\vskip 1cm

It is striking to anybody who worked with Alex that he had a strong geometrical view of the problems he was looking at. More precisely, he had an automatic and unavoidable a priori way of not only looking but intentionally \textit{watching geometrical}. This attitude happened to be very fruitful in several cases where the underlying geometrical structure was not evident to detect, but was crucial for the early development of future theories. To quote only two among many, let us first mention 
the excavation of the group theoretical structure (orthogonality relations for non-unimodular groups) \cite{gm} out of the hectic, say, computations of Jean Morlet \cite{m75}: the geometry of group theory underlying the not already truly born continuous wavelets was definitive for their primitive development. The second one, much less known, is the discovery, in \cite{ag73} that a non trivial operator algebra appears in regular quantum mechanics for solid state physics in presence of a magnetic field. The reduction of the dynamics by the periodic symmetry, 
%(equivalent to the Bloch decomposition without magnetic field leading to a reduced space  equal to $T^*\bT^d$ in $d$ dimensions)
 much more sophisticated than the (cotangent bundle of the) torus obtained by Bloch decomposition when no magnetic field is present: the reduced dynamics lives in a von Neumann algebra of type II, a premonition of noncommutative geometry to be born more than 10 years later in \cite{c81}.

%By some respects, the birth of quantum mechanics in 1925 in the extraordinary article by Heisenberg in \cite{wh25} recalls the birth of the (mathematical)   theory of wavelets: 
By some respects, the birth of the (mathematical)   theory of wavelets recalls  the birth of quantum mechanics, in 1925, in the extraordinary article by Heisenberg in \cite{wh25}: 
working together with Max Born on a possible quantization `` \`a la Bohr-Sommerfeld" of the perturbation series arising in the classical Helium atom, Heisenberg felt somehow necessary to enlarge the classical paradigm from functions subject to commutative multiplication (actually Fourier series endowed with commutative convolution) to matrices, endowed with their noncommutative multiplication law\footnote{the similitude between iteratde convolutions involving sum of terms of the form $\sum\dots a_{i-j}b_{j-k}c_{k-l}\dots$ with iterated multiplication on matrices involving this times sum of the form $\sum\dots a_{i,j}b_{j,k}c_{k,l}\dots$ is striking $\dots$ at least when one knows what a matrix is, which wasn't the case of Heisenberg.}. At this precise moment, a new space was born: the one in which a true quantum \textit{dynamics} could live, that is a natural kinematics for it. Every mathematician understands the gap between a function and a matrix and the impossibility to reduce the second to the first. Nevertheless, as the (commutative algebra of) continuous functions defines in a unique way the manifold on which they are defined, the noncommutative algebra of Heisenberg matrices was going to define a new space, a noncommutative one, the space of quantum mechanics. By no mean reducible to the space of classical mechanics. %By this discovery, 

In 1925, Heisenberg, 
%like Monet, 
``\textit{tel Monet, a peint qu'on ne voyait pas}".\footnote{``like Monet, has painted that one didn't see."}

\subsection{From Planck to Heisenberg}\label{plaheis}

The quite long gestation of quantum mechanics, from the seminal work of Planck to the first true paper by Heisenberg involving truly quantum dynamics  twenty five years later, can maybe be explained by the fact that the totally new quantum paradigm was born deeply inside the classical one. A classical paradigm strongly influenced, from centuries, by the notion of a geometrical space (flat or curved, in three or four dimensions, what ever) which was considered as the unavoidable environment inside which any dynamics should take place.

One of the difficulties met by quantum mechanics to be born was precisely the incompatibility of its foundations with the classical space, while at the same time the perception of quantum effect  live precisely in this classical space.
%%%%%%%%%%%%%%%%%%%%%%%%%%%%%%%%%%%%pouralex8%%%%%%%%%%%%%%%%%%%
%Christ en majeeste,perspective

Very early, soon after Heisenberg's paper and even before Schrödinger delivered the equation in space of the waves conjectured by de Brogglie, Dirac understood a fundamental link between classical and quantum paradigm : the commutators of matrices is the quantization of the Poisson bracket\footnote{Let us remind that the commutator of two matrices $A$ and $B$ is defined through $[A,B]:=AB-BA$ and the Poisson bracket of two functions $f(q,p),g(q,p)$ on (say) $\bR^2$ is the function defined by $\{f,g\}(q,p):=\partial_qf(q,p)\partial_pg(q,p)-\partial_pfq,p)\partial_q g(q,p)$.} of  functions on classical spaces: 
$$
\mbox{classical $\{\cdot,\cdot\}\ \longleftrightarrow \ \tfrac i{\hbar}[\cdot,\cdot]$ quantum}
$$
In fact, the general feeling at the beginning of quantum mechanics was that only the right oriented arrow was significant: one quantizes the Poison bracket to get the commutator:
$$
\mbox{classical $\{\cdot,\cdot\}\ \overset{0\ \to\ \hbar}\longrightarrow \ \tfrac i{\hbar}[\cdot,\cdot]$ quantum}
$$
 but soon, thanks to the asymptotic studies of the Schrödinger equation for small values of the Planck constant, the left oriented arrow became more important, and somehow satisfying: the world is quantum but when the Planck constant is small, the classical paradigm reappears
 $$
\mbox{classical $\{\cdot,\cdot\}\ \underset{\hbar\ \to\ 0}{\longleftarrow} \ \tfrac i{\hbar}[\cdot,\cdot]$ quantum}
  $$
  But, if this ``correspondence" has the advantage of showing the (our) classical word at the border ($\hbar=0$) of the (true) quantum word, both directions of the arrow just mentioned  have their own, unsolvable problems.
  
  The direct quantization arrow $\overset{0\to\hbar}\longrightarrow$, introduced by Weyl very few years after the works by Heisenberg and Schrödinger, links two paradigms having two fundamentally different groups of \textit{symmetries}. The natural geometrical setting of classical mechanics is the one of symplectic spaces, endowed  with the group of symplectomorphism. The natural symmetries of the target quantum space are conjugations by unitary operators. These two groups of symmetries merge only for the restriction to linear symplectomorphisms, mapped into the so-called metaplectic operators (we'll get back to this with precise definitions in Section \ref{quantcomptrans} below). In other words, the only classical  changes of variables truly  ``quantizable" are the ones which have already the very quantum flavour of linearity, a very small selection of possible classical symmetries. On the contrary, quantum statistics, bosons versus fermions, have no direct equivalents in classical ``naive" situations. And these are just few examples.
  
  The classical limit direction $\underset{\hbar\to 0}{\longleftarrow}$ is by many ways analogue to the passage from physical to geometrical optics: it reintroduces trajectories    out of waves. Moreover it re-establishes classical  functions with point wise (commutative)  multiplication out of noncommutative quantum multiplication of matrices by the (formal) argument:
  $$
  \mbox{$[\cdot,\cdot]\sim\tfrac\hbar i\{\cdot,\cdot\}\to 0\mbox{ as }\hbar\to 0.$}
  $$
  Beside its aesthetic, securing beauty, this argument is by many ways too simple, and leave unsolved many problems.
  
  First of all it doesn't explain the very fundamental difference between classical and quantum dynamics for large times: the quantum one, linear, being quasi-periodic and the classical one, in the chaotic situation, showing high dependence on initial conditions. Second, the culturally admitted, till the last decades, equivalence between the two dichotomies
  $$
  quantum\ /\ classical\sim 
  \left\{
  \begin{array}{rcl} 
  micro&/&macro\\
  short\ times&/& long\ times
  \end{array}
  \right.
  $$
  is not experimentally suitable any more: very excited atoms of Rubidium which have  the size of a small bacteria, are tractable nowadays in experimental quantum information for time scales of the order of several seconds.
  
  The reality is that what is contained in the ``limit" $\hbar\to 0$ is much more complex than taking a simple limit in the values of one parameter, as one does, for example, by letting the speed of light diverging. In fact it is, at an epistemological point of view, closer to the different large number $N$ of particle limits in classical mechanics: mean-field, grand canonical, Boltzmann Grad etc. Very different situations are reachable under the simple aphorism $N\to\infty$.
  
  In the different sections of this article, we will present several quantum  situations where the classical limit $\hbar\to 0$ involves geometrical situations which cannot be handled by the classical structures we just described, with a particular focus on the situations where the noncommutative aspects, flavour of quantum mechanics, don't disappear at the limit $\hbar\to 0$.
  
  \subsection{From symbols to operators: a quick journey in quantizland}\label{quantizland}
  In 1925, Heisenberg invented quantum mechanics as a change of paradigm from (classical) functions  to (quantum) matrices. He founded the new mechanics on the well known identity
$$
\tfrac 1{i\hbar}[Q,P]=1
$$
that, a few months later, Dirac recognized as the quantization of the Poisson bracket
$$
\{q,p\}=1.
$$
Again a few years later, Weyl stated the first general quantization formula by associating to any function $f(q,p)$ the operator
$$
F(Q,P)=\int \tilde f(\xi,x)e^{ i\frac{ xP-\xi Q}\hbar}d\xi dx
$$ 
where $\tilde f$ is the symplectic Fourier transform defined analogously through 
$$
f(q,p)=\int \tilde f(\xi,x))e^{ i\frac{ xp-\xi q}\hbar}d\xi dx.
$$

Many years after was born the pseudodifferential calculus first establish by Calderon and Zygmund, and then formalized by H\"ormander through the formula giving the integral kernel $\rho_F$ for the quantization $F$ of a symbol $f$ in $d$ dimensions as
$$
\rho_F(x,y)=
%(2\pi\hbar)^{-d}
\int f(x,\xi)e^{i\frac{\xi(x-y)}\hbar}
%d\xi.
\frac{d\xi}{(2\pi\hbar)^{d}}.
%d\xi/(2\pi\hbar)^{d}
$$
A bit earlier had appeared, both in quantum field theory  and in optics (Wick quantization) the (positive preserving) Toeplitz quantization of a symbol $f$\footnote{The Dirac notation will be used through the whole article:  $\vert\#\rangle$ will be ment as an element of a Hilbert space $\cH$, $\langle\#1\vert\#2\rangle$ as the scalar product 
%between 
$(\vert\#1\rangle,\Vert\#2\rangle)$, and $\vert\#1\rangle\langle\#2\vert$ as the operator on $\cH$ defined by $\vert\#1\rangle\langle\#2\vert\varphi:=(\vert\#2\rangle,\varphi)\vert\#1\rangle.$,$\varphi\in\cH$.}
$$
\mbox{Op}^T[f]=\int f(q,p)|q,p\rangle\langle q,p|dqdq
$$ 
where $|q,p\rangle$ are the famous (suitably normalized) coherent states defined as $\psi^{\alpha=i}_{z=(q,p)}$ in Section \ref{off} below.

As we see, quantization is not unique. But all the different symbolic calculi obtained above, after inversion of the quantization formul\ae\ written above, share the same two first asymptotic features:
\begin{itemize}
\item the symbol of a product is, modulo $\hbar$, the product of the symbols
\item the symbol of the commutator divided by $i\hbar$ is, modulo $\hbar$ again, the Poison bracket of the symbols.
\end{itemize}
In other words, they all define a classical underlying space (an algebra of functions) endowed with a Poisson (of more generally symplectic) structure.
\subsection{From operators to symbols: tell me which operator you are, I'll tell you what is your symbol}\label{tellme}

But it is very easy to show that this nice quantum/classical picture has its limits. And one can easily construct quantum operators whose classical limit will not follow the two items expressed above.

Consider for example the well known creation and annihilation operators $a^+=Q+iP,a^-=Q-iP$. They act of the eigenvectors $h_j$ of the harmonic oscillator by
$$
 a^+h_j=\sqrt{(j+\tfrac12)\hbar}h_{j+1},\   a^-h_j=\sqrt{(j-\tfrac12)\hbar}h_{j-1}.
 $$
 Consider now the matrices
 
\[
M_1^+=\begin{pmatrix}
0&0&0&0&\dots&\dots\\
1&0&0&0&\dots&\dots\\
&&\dots&&&\\
&&\dots&&&\\
&&\dots&&&\\
0&\dots&0&1&0&\dots\\
\dots&\dots&\dots&\dots&\dots&\dots
\end{pmatrix}
\] 
and its adjoint
\[
M_1^-=
\begin{pmatrix}
0&1&0&0&\dots&\dots\\
0&0&1&0&\dots&\dots\\
&&\dots&&&\\
&&\dots&&&\\
&&\dots&&&\\
0&\dots&0&0&0&\dots\\
\dots&\dots&\dots&\dots&\dots&\dots
\end{pmatrix},
\]
and the operator $\mathcal M_1^\pm$ they define on the basis $\{\varphi_n^N,n=o,\dots,N-1\}$ of $\mathcal H_N$ defined by \eqref{basis}.

An elementary computation shows that 
$$
\mathcal M_1^+=a^+(P^2+Q^2)^{-1/2},\ \mathcal M_1^-=(P^2+Q^2)^{-1/2}a^-. 
$$
therefore, their (naively) expected leading symbols are $f^+(q,p)=\sqrt{\frac{q+ip}{q-ip}}$ and $f^-(q,p)=\sqrt{\frac{q-ip}{q+ip}}$ or, in polar coordinates $q+ip=\rho e^{i\theta} $,
%$f^+=\sqrt\frac{e^{i\theta}}{e^{-i\theta}}, f^-=\sqrt\frac{e^{-i\theta}}{e^{i\theta}}$.
$f^\pm=e^{\pm i\theta}$.

 %are Toeplitz operators, in the sense of Section \ref{canonical}, of exact symbol the functions $f_1^\pm(z)=\sqrt{\frac z{\bar z}}=e^{i\theta}$.

If symbolic calculus would work 
%for these Toeplitz operators, 
the leading symbol of $\mathcal M_1^+
\mathcal M_1^-$ should be equal to $f^+f^-=1$ 
%and a Toeplitz operator of leading symbol $f_1^+(z)\times f_1^-(z)=1$, 
and $\mathcal M_1^+
\mathcal M_1^-$ should be therefore  close to the identity $I$ as $\hbar\to 0$.

 But
\[
M_1^+M_1^-=
\begin{pmatrix}
0&0&0&0&\dots&\dots\\
0&1&0&0&\dots&\dots\\
&&\dots&&&\dots\\
&&\dots&&&\dots\\
&&\dots&&&\dots\\
0&\dots&0&0&1&\dots\\
\dots&\dots&\dots&\dots&\dots&\dots
\end{pmatrix}=I-|h_0\rangle\langle h_0|\nsim I+o(1)\mbox{ as }\hbar\to 0,
\]
The reason for this defect comes from the fact that the functions $e^{i\pm\theta}=\sqrt{\frac{q\pm ip}{ q\mp ip}}$ are not  smooth functions on the plane. In fact they are not even continuous at the origin: $f^\pm$ can tend to any value in $\{e^{i\theta},\theta\in\R\}$ when $q,p\to 0$ tends to zero.

Note finally that the commutator 
\be\label{comm}
[\mathcal M_1^+,\mathcal M_1^-]
=
%\scriptsize
\begin{pmatrix}
-1&0&0&0&\dots&\dots\\
0&0&0&0&\dots&\dots\\
&&\dots&&&\dots\\
&&\dots&&&\dots\\
&&\dots&&&\dots\\
0&\dots&0&0&0&\dots\\
\dots&\dots&\dots&\dots&\dots&\dots
\end{pmatrix}=-|h_0\rangle\langle h_0|\neq O(\hbar),
\small
\ee
so that its symbol doesn't vanish at leading order, as expected by standard symbolic asymptotics.
\subsection{From symbols to  classical: spaces as possible symbolic calculi}\label{claposymb}\ 

Therefore the following question arises at the ``classical" limit $\hbar\to0$: \textbf{what is the corresponding classical paradigm underlying this limit?}
%\vfill

The natural answer to this question seems to us: \textit{it is the locus where lives the limit at $\hbar\to 0$ of the symbol of the quantum object, after the later has been expressed in a through a symbolic calculus uniform in the Planck constant.}

For all the cases treated in Section \ref{quantizland}, Wigner, Weyl, Töplitz etc,  this locus was the same: the classical phase-space out of which the diverse quantization procedures were issued. But we just saw in the preceding Section \ref{tellme} an example where this ``rule" failed.
The subject of the present article is to build other situations, interesting both from a physical and mathematical points of view, where the classical underlying space is not the standard one, mostly endowed with a noncommutative structure inherited from the overlying, true, quantum paradigm.
 
 It happens that the different symbolic constructions that we will meet  will be generalizations of the T\" oplitz one:
 $$
\mbox{Op}^T[f]=\int f(q,p)|q,p\rangle\langle q,p|dqdq
$$ 
which will remain somehow the core of our constructions.
Let us remind that, on the contrary of the Weyl or pseudodifferential quantization situation, the T\" oplitz quantization doesn't offer a recipe which enables to recover the (full) symbol out of the operator per se. We believe that this fact, often considered as a pain in the neck, is in fact responsible for the rich possibilities of enlarging the aforementioned formul\ae\ to more general ones.
\vskip 0.5cm
To put it in a  nutshell, \textbf{a space, quantum, classical or anything else, is nothing but a symbolic calculus endowing an operator algebra}.

\section*{Summarize}

This article id devoted to the study of the space hosting the  dynamics underlying the quantum one at the limit where the Planck constant vanishes. In many situations this underlying space is not the expected one, namely the classical space on which the quantum theory is built through different quantization processes. In other words,  the true classical \textit{underlying} space of quantum mechanics is not the space of which the quantum \textit{overlying} dynamics is the quantum one. Joking speaking, the diagram doesn't commute.

In particular, we will show that a strong taste on noncommutativity, keystone of the quantum paradigm, remains sensitive, at the classical limit, in many situations. This underlying space will take different aspects, endowing therefore an existentialist flavour totally transverse, if not orthogonal, to the essential vision of an absolute space hosting the dynamics.

This unusual conception of what a space should be, unusual in our classical culture, has to be put in front of the famous sentence by Einstein: `` It is the theory who decides what is to be observed."\footnote{comments to W. Heisenberg after a talk he gave at the period where he was still looking for quantum trajectories, F. Balibar, private communication.}. Different dynamics, although they all live in the same quantum framework of Hilbert space, give rise to different spaces when the Planck constant becomes small.

In fact, in several situations, a naive guess for the underlying space leads to drastic singularities, which  forbid the existence or uniqueness of the classical dynamics.But  these singularities disappear after a new suitable noncommutative re-dressing. This is perfectly in accordance with the philosophy of noncommutative geometry of Alain Connes, and more generally, in fact, with the birth (and efficiency) of the quantum paradigm.

If one believes that nature is quantum, after all, these new ``classical"  structures, albeit arising at a large scale compared to the one of the Planck constant but also keeping track of the source of stability that constitutes the quantum noncommutativity, might explain the amazingly stable feeling provided by the macroscopic world.

%\red{AFAIRE A LA RECHERCHE DE L'ESPACE PERDU, OU COMMENT LA BONNE QUESTION (FREQUENCES A LA PLACE DES vp) DONNE LA BONNE REPONSE. \textit{c'EST LA THEORIE QUI DIT CE QU'IL FAUT MESURER. ET C'EST BEN LES FREQUENCES QUE L'ON VOIT}
%
%MENTIONNER AUSSI QUE LE CLASSIQUE C'EST LE DIGAONAL ET QUE LE OFF DIAGONAL EST STRICTEMENT QUANTIQUE, EXEMPLE DU MOINS QUI DEVIENT VIRGULE CHEZ POINCARÉ
%
%DIRE PEUT-ETRE ICI 9(USSI) QUE L'ON REPRESENTE LÍDENTITE EN MQ  SOUS DIVERSES REPRESENTATIONS ASSOCIEES A  DIVERSES GEOMETRIES SOUS-JACENTES , EXISTENTIELLES CAR DEPENDANTES DE LA DYNAMIQUE  POURTANT INVISIBLE SOUS LÍDENTITE
%
%ARLER DE TORES HOMOCLINES (INFINITESIMAUX)
%\vskip 1cm
%INSISTER SUR REDONDANCE, COMME ONDELETTES
%\vskip 1cm{\bf
%TRES MPORTANT: ;Amc COMPORTE DES TROUS D' INCOMPLETUDE QUE LA mq VIENT COMBLER A LA LIMITE HBAR PETIT: LA LIMITE CLASSIQUE NON SEULEMENT  REDONE LE PARADIGME CCLASSIQUE MAIS DONNE UN PARADIGME CLASSIQUE COMPLET
%}
%}
\section*{Overview of the article}

This article contains  six sections, each of them devoted to a different quantum situation whose limit at $\hbar\to 0$ does not correspond to the expected classical commutative paradigm.

%\red{AFAIRE
%METTRE EN VUE LE FAIT QUE C'E STTOUJOURS TOPLITZ
%}
The keystone through the whole paper will be the (positive preserving) Toeplitz quantization of a symbol $f$ already defined in  Section \ref{intro}
$$
\mbox{Op}^T[f]=\int f(q,p)|q,p\rangle\langle q,p|dqdq
$$ 
where $|q,p\rangle$ are the (suitably normalized) coherent states, which will go through several modifications all along the paper.

One of the most striking modification of the preceding formula will consist in changing the \textit{diagonal} diadic element $|q,p\rangle\langle q,p|$ by \textit{off diagonal} ones of the form $|q',p'\rangle\langle q,p|$ where $(q',p')$ is related to $(q,p)$: $(q',p')=f(q,p)$ for a certain function $f$.

\vskip 1cm
In Section \ref{quantcomptrans}, we investigate a new Toeplitz-like quantization suitable fro quantum observables propagated by quantum flows generated by complex quadratic Hamiltonians. More precisely, we look at the quantization of complex linear symplectic flows, a subject studied by Alex in a beautiful paper \cite{gross}. In particular we show that a very simple and ``natural" off-diagonal extension of the Toeplitz paradigm $\mbox{Op}^T[f]$ introduced in the preceding section happens to be very useful. By off-diagonal, we mean a construction of the form
 $$
\mbox{Op}^T_{complex}[f]=\int f(q,p)|q'(q,p),p'(q,p)\rangle\langle q,p|dqdq
$$ 
for some mapping $(q,p)\to(q',p')$. Associated to this formula is a noncommutative symbol belonging to the $C^*$-algebra of a groupoid linked to the complex symplectic group. Moreover, this formulation extends naturally to the non-canonical case.
\vskip 0.5cm
Section \ref{boseinferm} investigates quantum statistics in front of 
%the situation where the Planck constant vanishes
a possible underlying phase space. The general feeling about this question is that the notion itself of quantum statistics disappears at the classical limit, since the wave function disappear by oscillation when $\hbar\to 0$.
\noindent In fact we show that exchanging two quantum particles correspond at the classical level to letting a complex linear mapping  act on phase space, like in Section \ref{quantcomptrans}. But this time the mapping is anticanonical: it send the symplectic form  to its opposite
$$
dp\wedge dq\rightarrow -dp\wedge dq.
$$
In other words it is not a symmetry of the classical paradigm. But this construction infers the possibility of introducing a noncommutative ``classical" phase space suitable for carrying quantum spin-statistics at ``a" classical level.

\vskip 0.5cm
In Section \ref{tqft}, we turn to  a problem in low dimension topology which in principle has nothing to do with quantum mechanics: the asymptotics on large values of colors of curve operators in topological quantum field theory. In \cite{mp} J. March\'e and myself showed that ``generically" on the way this asymptotics is taken, the curve operators are Toeplitz operators associated to phase spaces which are moduli spaces isomorphic to the  two dimensional sphere.
\noindent Recently, in \cite{tp2}, I investigated   the general case of this asymptotics, and shown that the corresponding ``semiclassical" structure is the one of another extension of Toeplitz quantization- This extension correspond to changing the Hilbert structure of the Hilbert space appearing in the (geometrical) quantization of the symplectic 2D sphere and consider the associated orthogonal projector noted $|\cdot\rangle_a\langle\cdot|$:
$$
\mbox{Op}^T_{TQFT}=\int |q,p\rangle_a\langle q,p|dqdp.
$$
%In the present paper, w
We investigate the noncommutative structure one has to put on the sphere in order that it corresponds to the  moduli (phase) space underlying this singular asymptotics.

\vskip 0.5cm
Section \ref{longsemevo} is concerned with the long time quantum evolution of quantum observables,  pseudodifferential operators, uniformly when $\hbar\to 0$. It is well known that such evolution remains in the pseudodifferential category for scales of times not larger that logarithmic in $1/\hbar$. the reason for this is the fact that the non-locality of the quantum paradigm induces a sensibility to nearby trajectories whose spreading in general is (shown to be less than) exponential in time. Willing to be ``microlocal" implies that this spreading in phase space  should not exceed the size of the Planck constant fixes the logarithmic scale aforementioned.
\noindent In \cite{tp} was shown than passing this logarithmic barrier was possible thanks to, again, an extension of the Toeplitz procedure. In the present article we reformulate this extension in the following ``groupoid" form
$$
\mbox{Op}^T_{groupoidT}=\int_{(q',p')\sim^{u/s}(q,p)} |q',p'\rangle\langle q,p|dqdpdq'dp'.
$$
where the equivalence relation $\sim^{u/s}$ means ``belongs to the same unstable manifold" for a long time forward propagation  and `belongs to the same stable manifold" for a long time backward one.

\vskip 0.5cm
In Section \ref{homo}, we describes how these Toeplitz extension can provide semiclassical  spectral informations. Since spectral data are related to invariance by (quantum) evolution, it is normal to try to let merge the two for- and backward settings of the preceding section, namely to take the intersection of stable and unstable manifold. In other words to consider homoclinic orbits. That is to say that one has to handle ansatz of the form
$$
\mbox{Op}^T_{homoclin}=\int_{(q',p')\sim^{homo}(q,p)} |q',p'\rangle\langle q,p|dqdpdq'dp'.
$$
where $(q',p')\sim^{homo}(q,p)$ means ``$(q',p')$ belongs to any  orbit homoclinic to (the trajectory passing by) $(q,p)$".
\noindent We construct a noncommutative algebra associated to this ``homoclinic" groupoid which might determines (part of) the set of frequencies (differences of eigenvalues) at the classical limit. This construction is astonishly close conceptually to the explicit case of the non-resonant harmonic oscillator: the set of homoclinic orbits, intersections of two given stable and unstable manifolds of a dynamical  system playing the role of an invariant torus in the integrable case.

Finally, Section \ref{imvrsin} concerns the link between quantum indeterminism, a key stone of the quantum paradigm, and  unpredictability in classical chaotic systems, forced by sensitivity to initial conditions at the limit where the time of evolution diverges. We show that these two concepts - indeterminism and  unpredictabilty - merge at the classical limit
when is performed a change of the concept of classical space.

\vskip 0.5cm
We conclude this article with a short epistemological essay on the geometrical structure of the quantum-classical entanglement: the classical dynamics appears at the border of the quantum one but contains also, at its own border, the quantum one.....
 \vskip 1cm 
 
 Let us finish this overview by remarking that the key stone of most of the mathematical constructions in this paper have the same core: the Toeplitz quantization.
 
 %\vskip 3cm
 \section{Quantizing complex canonical transformations}\label{quantcomptrans}
 
% \red{AFAIRE
% GROUPOIDE ASSOCIE A L 'EQUATION DE H VN: LOI DE GROUPOIDE QUAND ON COMPOSE....
% }
 This section contains the first of our attends of considering classical phase spaces as noncommutative spaces underlying quantum ones, namely the push-forward of the usual classical phase space  by complex canonical transforms.
 We shall consider  linear canonical (and later on non canonical) transforms, so, as mentioned earlier, there is no need of taking the limit $\hbar\to 0$ to ``recover" a classical space. This is the case for real linear canonical transforms, but we will see that the situation is much more rich in the complex case.
 
 \vskip 1cm
 Complex quantum Hamiltonians have (re)gained a lot of interest these last years, see e.g. the book \cite{Ler} and all the references quoted there. It seems that, on the contrary, quantization of complex symplectic flows didn't get the interest it deserves as it did at the period of the birth of Fourier integral operators, see \cite{Sjo} for example. Inside this category, representation of \textit{linear} complex symplectic flows, namely the complex symplectic group and its corresponding complex metaplectic representation has also lost most of interest since the golden years of ``group theoretical methods in physics", see \cite{gross}. This section relies on \cite{tp3} and we will handle complex canonical transform not by direct metaplectic operators, but rather by conjugation of observables by them.

Let us remark that, since complex numbers are  (according to us) necessary to the formulation of quantum mechanics, considering the Schrödinger equation on a Hilbert space $\cH$  with complex (e.g. bounded) Hamiltonian $\cH$ does not create any intrinsic a priori difficulty, a property that classical mechanics (on a real symplectic phase space $\cP$) doesn't share with its quantum counterpart. The quantum flow is still given by $e^{-itH/\hbar}$ which exists (of course it is not any more unitary) for all time for example when $H$ is bounded.  $e^{-itH/\hbar}$ defines for any $t\in\R$ a bounded operator on $\cH$, On the contrary of the Lie exponential $\L^{-th}$ associated to the symbol $h$ of $H$ which doesn't apply on $\cP$ because the Hamiltonian vector field associated to $h$ is not tangent to $\cP$ any more.

\subsection{Off-diagonal Toeplitz representation}\label{off}

 \vskip 1cm
We will study conjugation of quantum observables given by the  Toeplitz (anti-Wick) construction by complex metaplectic operators. In order to avoid heavy notation, we will be, in most of the paper, concern with the one degree of freedom case: a two dimensional phase space that will be, for simplicity, $\R^2\sim\C$. The general higher dimension case is extensively treated in \cite{tp3} which contains most of the rest of this section.

 More precisely we consider operators of the form 
\[
H=\int_\C h_\alpha(z)\vert\pbz\rangle\langle\pbz\vert\frac{dzd\bar z}{2\pi\hbar}.
\]
where the family of coherent states $\pbz$ is defined, 
%for $\Im\beta>0$, by
%\[
%\pbz==
%\frac{e^{-\frac{\beta(x-z\re)^2}{2\hbar}}e^{i\frac{z\im x}\hbar}}{(\pi\hbar/\beta)^\frac14}
%%e^{-\frac{\beta(x-z\re)^2}{2\hbar}}e^{i\frac{z\im x}\hbar}.
%\]\\
%Let, 
for $\alpha\in\C,\ \Im\alpha>0$ and $z=(q,p)\in\R^2$ 
%the coherent state $\psi^\alpha_z$ defined 
by
\[
\psi^\alpha_z(x)
:=
%\frac1{\sqrt{\pi\Re\beta}}
\left(\frac{\Im\alpha}{\pi\hbar|\alpha|^2}\right)^{\frac14}e^{-\frac i{2\hbar\alpha}(x-q)^2}e^{i\frac{px}\hbar}e^{-i\frac{pq}{2\hbar}}.
\]
Note that the standard Toeplitz quantization correspond to $\alpha=i$, but $\alpha$ will show to be a true dynamical variable, so we need to consider it as a degree of freedom. The link between different $h_\alpha(\cdot)$ leading to the same operator $H$ (in particular the standard case) is 
%given below, in Appendix \ref{w}, Lemma \ref{hushus}. One has
the following
\begin{eqnarray}
\int_\C h_\alpha(z)\vert\pbz\rangle\langle\pbz\vert\frac{dzd\bar z}{2\pi\hbar}
&=&
\int_\C h_{\alpha,\alpha'}(z)\vert\psi^{\alpha'}_z\rangle\langle\psi^{\alpha'}_z\vert\frac{dzd\bar z}{2\pi\hbar}\label{alphaalphaprime}\\
&\Updownarrow&\nonumber\\
h_{\alpha,\alpha'}&=&e^{-i\frac\hbar4(\alpha-\alpha')\Delta_\xi+i\frac\hbar4(\frac1\alpha-\frac1{\alpha'})\Delta_x}h_{\alpha'}\label{betabetaprime}
\end{eqnarray}
for $ \Im{(\alpha-\alpha')}>0.\ \Im{(\frac1\alpha-\frac1{\alpha'})}<0$.

An easy computation shows that, when $h_\alpha(z):=1$, $H=\mbox{Id}$ for any value of $\alpha$.
\vskip 1cm
We consider the operator $U(S)\inv HU(S)$, conjugated of $H$ by the operator $U(S)$ where $S$ is a real $2\times 2$ matrix $\begin{pmatrix}
a&b\\c&d
\end{pmatrix} $ of determinant one and $U$ is the metaplectic representation. More precisely $U(S)$ is the operator of integral kernel given by
\bea\label{udexy}
U(x,y)&=&\frac1{\sqrt{b2\pi\hbar}} e^{-\frac i{2b\hbar}\left(dx^2-2xy+ay^2\right)}\ \ \ \ \ \ \ \ b\neq 0\nonumber\\
&&\\
&=&{\sqrt d}\delta\left(dx  -y\right)e^{i\frac c d\frac{x^2}2}\ \ \ \ \  \ \ \ \ \ \ \ \ \ \ \   b=0\nonumber
\eea
In order to show of \eqref{udexy} can be derived, let us recall that one way to define the metaplectic representation is through the formula
\be\label{meta}
S\binom{x}{-i\hbar\frac d{dx}}=\binom{U(S)^{-1}xU(S)}{U(S)^{-1}{(-i\hbar\frac d{dx})}U(S)}.
\ee
Writing that $U(S)$ is unitary by $U(S)^{-1}=U(S)^*$, that is $U^{-1}(x,y)=\bar{U(y,x)}$ we get for $S=\begin{pmatrix}
a&b\\c&d\end{pmatrix}$, an equation whose solution is \eqref{udexy}, of course modulo a global phase.

  It is well known and easy to derive after \eqref{meta} that, when $S$ is real the Weyl symbol of $U(S)HU(S)^{-1}$ is the push-forward of the Weyl symbol of $H$ by $S$. namely
  \be\label{pfw}
  \sigma_{U(S)^{-1}HU(S)}^{Weyl}\binom{q}{p}=\sigma_{H}^{Weyl}\left(S^{-1}\binom{q}{p}\right).
  \ee
  
  %and $\beta=1$
  An easy computation
  %, see Lemma \ref{wickmeta} below, 
  shows that this implies the following result.
  \[
U(S)^{-1}HU(S)=\int h_{\alpha_S(\alpha),\alpha}(Sz)\vert\pbz\rangle\langle\pbz\vert\frac{dzd\bar z}{2\pi\hbar}%\mbox{ with }\beta_S(\beta)=iS\cdot(-i\beta),
\]
$\mbox{ with }\alpha_S(\alpha)=iS\cdot(\alpha)$, where  $S\cdot z=\frac{az+b}{cz+d}$, and $h_{\alpha_S(\alpha),\alpha}$ is given by \eqref{betabetaprime}.
 
 Of course this formula doesn't make any sense for general (non analytic) symbols $h$ when $S$ is not real any more. But we will see that there is a general ``off-diagonal" Toeplitz representation.
%  of the form
%\[
%U(S)HU(S)^{-1}=C_S\int h_{\beta_S(\beta)}(T^{-1}_S(z))\vert\psi^{\beta_{S^{-1}\bar S}}_{z_S}\rangle
%\langle\psi^{\beta'_S}_{z}\vert\frac{dzd\bar z}{2\pi\hbar}.
%\]

%We associate to $U(S\inv )HU(S)$ a (real) ``symbol" of the form $\rho(z,z'),\ z,z'\in T^*\R\sim \R^2\sim\C$ with
%\[
%\rho(z,z')=h\circ\Phi_S(z)e^{\frac{\mathbb Q_S(z}\hbar)}\delta(z'-\Psi_S(z)),
%\]
%in the sense that
%\[
%U(S)\inv HU(S)=\int \rho(z,z')\vert\pbz\rl\pbz\vert dzdz'.
%\]
The main result of the present short note is the following theorem.

\textbf{Notation}: 
%through the whole paper 
we will denote 
%by 
%%$z\re\in\mathbb R$ and $z\im\in\mathbb R$ the real and imaginary parts of $z\in\mathbb C$, i.e.  $z=z\re+iz\im$ and 
%$\com{S}$ the matrix complex conjugate to $S$ and 
for $S= \begin{pmatrix}
a&b\\c&d
\end{pmatrix}\in SL(2,\C)$, $z=(q,p)$ and $\alpha\in\C$, 
$$
\com{S}= \begin{pmatrix}
\bar  a&\bar b\\\bar c&\bar d
\end{pmatrix},\ S (z)=\begin{pmatrix}
a&b\\c&d
\end{pmatrix}\begin{pmatrix}
q\\p
\end{pmatrix}\mbox{ and }
S\cdot\alpha=\tfrac{a\alpha+b}{c\alpha+d}.
$$
Moreover, $\wedge$ will denote the symplectic form on $T^*\R$, $z\wedge z'=pq'-qp'$.
%\be\left\{
%\begin{array}{l}
%z^S=(p^S,q^S),\ \binom{p^S}{q^S}=S\binom{p}{q},\\
%\alpha_S(\alpha)=S\cdot\alpha,\ \begin{pmatrix}
%a&b\\c&d
%\end{pmatrix}\cdot\theta=\tfrac{a\theta+b}{c\theta+d},\ \theta\in\C\\
%z_S^\alpha=: {_\alpha T_S(z)}=(p_S^\alpha,q_S^\alpha)\in\R^2\mbox{ defined by }q_S^\alpha+\alpha_S(\alpha)p_S^\alpha=q^S+\alpha_S(\alpha)p^S.
%\end{array}
%\right.
%\ee
%Then, for any $S,\alpha$ such that $|\alpha_S(\alpha)|<\infty,\ \Im\alpha_S(\alpha)>0$, 
%%(when it exists, that is when $\Re(i\frac ab+\beta)>0$ or $b=0$)
%\[
%U(S)\psi^\alpha_{z}=e^{i\frac {z^S\wedge z_S^\alpha}{2\hbar}}\psi^{\alpha_S(\alpha)}_{z_S^\alpha}.
%\]
%\section{Off-diagonal Toeplitz representation}
\begin{theorem}[Off-diagonal Toeplitz representation]\label{thm1}\ 

Let, for $h$ compactly supported,
%%for $z=z\re+iz\im\in\mathbb C,\ ,\ 
%$\psi^\beta_{z}=
%\frac{e^{-\frac{\beta(x-z\re)^2}{2\hbar}}e^{i\frac{z\im x}\hbar}}{(\pi\hbar/\beta)^\frac14}$, $\beta\re>0$ 
%and
\[
H=\int h(z)\vert\psi^\alpha_z\rangle\langle\psi^\alpha_z\vert\frac{dzd\bar z}{2\pi\hbar}.
\]
Define, for $S
%=\left(\begin{matrix}a&b\\c&d
%\end{matrix}\right)
\in SL(2,\Z)$ and $\Im\alpha>0$, the real $2\times 2$ matrix ${_\alpha T_S}$ by 
%$$(\bar a d-\bar c b)\re-\beta(\bar d b+\bar ca))\im>0,\ \ \ \ \ \ \ \  (\star)$$
$$
%z_S^\alpha=: 
{_\alpha T_S(z)}=(q_S^\alpha,p_S^\alpha)\in\R^2\mbox{ defined by }q_S^\alpha+\alpha_S(\alpha)p_S^\alpha=q^S+\alpha_S(\alpha)p^S,\ 
\begin{pmatrix}
q^S\\p^S
\end{pmatrix}:=S(z),\ z=(q,p).
$$
Then, for any $S,\alpha$ such that $|V\cdot \alpha|<\infty,\ \Im(V\cdot \alpha)>0,\ V=S^{-1},\com{S}^{-1}, S^{-1}\com{S}$, 

%\[
%U(S)^{-1}HU(S)=\int 
%h_{\com{S}\cdot \alpha,\alpha}(_\alpha T_{\com{S}}z)
%%%e^{\frac{\beta_{S^{-1}}(T_{S^{-1}}T^{-1}_{\bar S}z)\re^2-2\beta(T^{-1}_{\bar S}z)\re^2+\beta_{\bar S}z\re^2}{2\hbar}}
%%\frac
%%%{e^{\frac{(T_{S^{-1}}T^{-1}_{\bar S}z)\re^2\beta_{S^{-1}}-2(T^{-1}_{\bar S}z)\re^2\beta+z\re^2\beta_{\bar S}-\delta}{2\hbar}}}
%%{e^{\frac{\mathbb Q_S(z)}\hbar)}}
%%{D^\beta_{\bar S}D^\beta_{S^{-1}}}
%|\det{_\alpha T_\com{S}}|
%{_\alpha\fract_S}
%e^{i\frac{
%%S^{-1}\com{S}(z)\wedge {_\alpha T_{S^{-1}\com{S}}}(z)
%S^{-1}({_\alpha T_{\com{S}}}(z))\wedge({_\alpha T_{S^{-1}\com{S}}(z)-z
%})}\hbar}
%\vert\psi^{S^{-1}\com{S}\cdot\alpha}
%_{{_\alpha T_{S^{-1}\com{S}}}(z)}
%%(\beta)}_
%%%{T_{S^{-1}}T^{-1}_{\bar S}z}
%%{_\alpha T_{S^{-1}\com{{S}}}(z)}
%\rangle\langle\psi^{\alpha}_z\vert\frac{dzd\bar z}{2\pi\hbar},
%\]
\[
U(S)^{-1}HU(S)=\int 
{_\alpha T_{\com{S}}}\#h_{\com{S}\cdot \alpha,\alpha}(z)
\tfrac{\vert\psi^{S^{-1}\com{S}\cdot\alpha}
_{{_\alpha T_{S^{-1}\com{S}}}(z)}
\rangle\langle\psi^{\alpha}_z\vert}
{\langle\psi^{S^{-1}\com{S}\cdot\alpha}
_{{_\alpha T_{S^{-1}\com{S}}}(z)}
\vert\psi^{\alpha}_z\rangle}
\frac{dzd\bar z}{2\pi\hbar},
\]
where $h_{\com{S}\cdot \alpha,\alpha}$ is defined by \eqref{alphaalphaprime}-\eqref{betabetaprime} and $T\# h$ is the push-forward of $h$ by $T$.
\end{theorem}
The proof is given in \cite{tp3}. the key ingredient for the effective computation of the result consists in the following trick: one easily compute that
$$
U(S)^{-1}\vert\psi^\alpha_z\rangle\langle\psi^\alpha_z\vert U(S)=L
\vert\psi^{S^{-1}\cdot\alpha}_{_\alpha T_{S^{-1}}z}\rangle\langle
\psi^{{\com{S}}^{-1}\cdot\alpha}_{_\alpha T_{{\com{S}}^{-1}}z}\vert \mbox{ for } L\in\C.
$$
But $U(S)^{-1}\vert\psi^\alpha_z\rangle\langle\psi^\alpha_z\vert U(S)$ is a projector, therefore 
$$
L^2\langle \psi^{S^{-1}\cdot\alpha}_{_\alpha T_{S^{-1}}z}\vert
\psi^{{\com{S}}^{-1}\cdot\alpha}_{_\alpha T_{{\com{S}}^{-1}}z}\rangle=L
\Longrightarrow L=\tfrac1{\langle \psi^{S^{-1}\cdot\alpha}_{_\alpha T_{S^{-1}}z}\vert
\psi^{{\com{S}}^{-1}\cdot\alpha}_{_\alpha T_{{\com{S}}^{-1}}z}\rangle}.
$$
%\vskip 1cm
%Note that, in particular
%\be\label{tbars}
%T_{\bar S}^{-1}=
%\frac1{(\beta_{\bar S})\re\beta\vert b\beta-id\vert^2}
%%\frac1{(\beta_{\bar S})\re\vert b+id\vert^2}
%\left(\begin{matrix}
%((\overline{d-ib\beta})\beta_{\bar S})\re&-(\overline{d-ib\beta})\im\\
%-((d-ib\beta)\overline{\beta}_{\bar S}\beta)\im &((d-ib\beta)\beta)\re
%\end{matrix}\right),
%\ee
%\be\label{tsmoins1}
%T_{S^{-1}}T_{\bar S}^{-1}=
%\frac1{(\beta_{\bar S})\re(\beta_{S^{-1}})\re\beta\vert b\beta-id\vert^4
%%\beta
%}
%%\left(\begin{matrix}
%%q_S\\p_S
%%\end{matrix}\right):=T_S
%\left(\begin{matrix}
%((d^2+b^2\beta^2)\overline\beta_{\bar S}\beta)\re&((d^2+b^2\beta^2)\beta)\im\\
%((d^2+b^2\beta^2)\beta_{S^{-1}}\overline\beta_{\bar S}\beta)\im&((d^2+b^2\beta^2)\beta_{S^{-1}}\beta)\re
%\end{matrix}\right).
%\ee
%%5and
%%, for $S=\left(\begin{matrix}a&b\\c&d
%%\end{matrix}\right)$,
%%\[
%%\beta_{S}=\frac{-c+id\beta}{b\beta+ia}.
%%\]

\subsection{Link with Weyl}\label{linkw}
It is well known, see the beautiful paper by Alex et al \cite{gs}, that the Weyl symbol of a Toeplitz operator is the function obtained by letting the heat flow at time $\hbar/4$ act on the Toeplitz symbol. therefore the Weyl symbol of a Toeplitz operator is an entire function on which a complex canonical transform acts naturally. 
%The following results gives also the classical limit of such a Weyl symbol.
As a corollary of Theorem \ref{thm1}, we get the following (expected) result.
\begin{theorem}\ 

Let $H$ with Toeplitz symbol $h$. Let us denote by $\sigma^\hbar_{H}$ the Weyl symbol of $H$ (note that $\sigma^\hbar_{H}=e^{\frac\hbar4\Delta}h$ is an entire function).

 Then
\[
\sigma^\hbar_{U(S)\inv HU(S)}=\sigma^\hbar\circ S.
\]
%AFAIRE
%and
%\[
%\sigma^0_{U(S)\inv HU(S)}(q,p)=\int \widehat{h}(\xi,x)e^{-i(q^S\xi+p^Sx)}dxd\xi
%\]
%where
%$$
%\left(\begin{matrix}
%q^S\\p^S
%\end{matrix}\right)={\LARGE S}\left(\begin{matrix}
%q\\p
%\end{matrix}\right).
%$$
\end{theorem}

\subsection{Flows on extended phase-space}\label{extphspa}
Consider on the extended phase space $\C\times \C=\{(z,\alpha)\}$ the mapping
\[
\Phi_S: (\alpha,z)\mapsto (S.\alpha,{_\alpha T_S}(z)).
\].
\begin{theorem}
\[
\Phi_{SS'}=\Phi_S\Phi_{S'}.
\]
\end{theorem}
\subsection{Examples}\label{annb}%AFAIRE
%Let us fix $\beta=1$.

Several examples are presented in the table below.

Let
$$
D_\alpha(q,p)=
%e^{i\frac{S^{-1}({_\alpha T_{\com{S}}}(z))\wedge({_\alpha T_{S^{-1}\com{S}}(z)-z
%})}\hbar},
{\langle\psi^{S^{-1}\com{S}\cdot\alpha}
_{{_\alpha T_{S^{-1}\com{S}}}(z)}
\vert\psi^{\alpha}_z\rangle}^{-1},\ z=(q,p).
$$
\begin{table}
%\caption{\large Examples}\label{table1}
\centering
%\caption{\large Examples}\label{table1}
\begin{tabular}
%{|p{2cm}|p{5.5cm}|p{2.5cm}|p{2.7cm} |p{2.7cm} |}
{|c|c|c|c|c|c|c|}
\hline  
\LARGE \begin{color}{white}$\begin{pmatrix}\\\end{pmatrix}$\end{color}\large $U$\LARGE \begin{color}{white}$\begin{pmatrix}\\\end{pmatrix}$\end{color}&\large $S$ &\large $S^{-1}S^c$ &\large $S^{-1}S^c\cdot i$ &\large $_iT_{S^{-1}S^c}$ &\large 
$_iT_{S^c}$ &\large $D_i(q,p)$\\  \hline

\LARGE \begin{color}{white}$\begin{pmatrix}.\\.\end{pmatrix}$\end{color}  \large complex time  free evolution
\LARGE \begin{color}{white}$\begin{pmatrix}.\\.\end{pmatrix}$\end{color} 
&\ \large $\ \ \begin{pmatrix}1&-it\\0&1\end{pmatrix}\ \ $   &\ \large $\ \ \begin{pmatrix}1&2it\\0&1\end{pmatrix}\ \ $   &\large $ i(1+2t)$ &\ \large $\ \ \begin{pmatrix}1&0\\0&\tfrac{1+4t}{1+2t}\end{pmatrix}\ \ $    &\ \large $\ \ \begin{pmatrix}1&0\\0&\tfrac{1+2t}{1+t}\end{pmatrix}\ \ $   &\large $e^{t\frac{(1+3t)^2}{(1+2t)^2}\frac{p^2}\hbar}$\\   \hline

\LARGE \begin{color}{white}$\begin{pmatrix}.\\.\end{pmatrix}$\end{color}\large  multiplication by  $  e^{-tx²} $\LARGE\begin{color}{white}$\begin{pmatrix}.\\.\end{pmatrix}$\end{color} &\ \large $\ \ \begin{pmatrix}1&0\\-it&1\end{pmatrix}\ \ $   &\ \large $\ \ \begin{pmatrix}1&0\\2it&1\end{pmatrix}\ \ $   &\large $\frac i{1-2t}$ &\ \large $\ \ \begin{pmatrix}\tfrac{1-4t}{1-2t}&0\\0&1\end{pmatrix}\ \ $    &\ \large $\ \ \begin{pmatrix}\tfrac{1-2t}{1-t}&0\\0&1\end{pmatrix}\ \ $   &\large $e^{t\frac{(1+3t)^2}{(1+2t)^2}\frac{q^2}\hbar}$\\   \hline

\LARGE\begin{color}{white}$\begin{pmatrix}.\\.\end{pmatrix}$\end{color}\large  complex dilation \LARGE\begin{color}{white}$\begin{pmatrix}.\\.\end{pmatrix}$\end{color} &\ \large $\ \ \begin{pmatrix}e^{it}&0\\0&e^{-it}\end{pmatrix}\ \ $   &\ \large $\ \ \begin{pmatrix}e^{-2it}&0\\0&e^{2it}\end{pmatrix}\ \ $   &\large $e^{-4it}i$ &\ \large $\ \ \begin{pmatrix}\tfrac{1-4t}{1-2t}&0\\0&1\end{pmatrix}\ \ $    &\ \large $\ \ \begin{pmatrix}\tfrac{1-2t}{1-t}&0\\0&1\end{pmatrix}\ \ $   &\large $e^{2i\cos t\frac{ qp}\hbar}$\\   \hline

\LARGE\begin{color}{white}$\begin{pmatrix}.\\.\end{pmatrix}$\end{color}\large complex  harmonic oscillator\LARGE\begin{color}{white}$\begin{pmatrix}.\\.\end{pmatrix}$\end{color} &\ \small $\ \ \begin{pmatrix}\cosh t&i\sinh t\\-i\sinh t&\cosh t\end{pmatrix}\ \ $   &\ \small $\ \ \begin{pmatrix}\cosh 2t&i\sinh 2t\\-i\sinh 2t&\cosh 2t\end{pmatrix}\ \ $   &\large $i$ &\ \large $\ \ \begin{pmatrix}e^{-2t}&0\\0&e^{2t}\end{pmatrix}\ \ $    &\ \large $\ \ \begin{pmatrix}e^{-t}&0\\0&e^{t}\end{pmatrix}\ \ $    &\small $e^{\frac{\sinh 2t(q^2+p^2)}\hbar}$\\   \hline

\LARGE\begin{color}{white}$\begin{pmatrix}.\\.\end{pmatrix}$\end{color}\large $\dots$\LARGE\begin{color}{white}$\begin{pmatrix}.\\.\end{pmatrix}$\end{color} &\ \large $\ \ \begin{pmatrix}0&i\\i&0\end{pmatrix}\ \ $   &\ \large $\ \ \begin{pmatrix}-1&0\\0&-1\end{pmatrix}\ \ $   &\large $i$ &\ \large $\ \ \begin{pmatrix}-1&0\\0&-1\end{pmatrix}\ \ $    &\ \large $\ \ \begin{pmatrix}-1&0\\0&1\end{pmatrix}\ \ $   &\large $1$\\   \hline

\end{tabular}

\label{Tab:interactions}

\end{table}

We finish this section by the same computations for a non-canonical $S$, used in Sections \ref{nomcatrans} and \ref{linkmeta} below.
\begin{table}
%\caption{\large Examples}\label{table1}
\centering
%\caption{\large Examples}\label{table1}
\begin{tabular}
%{|p{2cm}|p{5.5cm}|p{2.5cm}|p{2.7cm} |p{2.7cm} |}
{|c|c|c|c|c|c|c|}
\hline  
\LARGE \begin{color}{white}$\begin{pmatrix}\\\end{pmatrix}$\end{color}\large \LARGE \begin{color}{white}$\begin{pmatrix}\\\end{pmatrix}$\end{color}&\large $S$ &\large $S^{-1}S^c$ &\large $S^{-1}S^c\cdot i$ &\large $_iT_{S^{-1}S^c}$ &\large 
$_iT_{S^c}$ 
%&\large $D_i(q,p)$
\\  \hline

%\LARGE\begin{color}{white}$\begin{pmatrix}.\\.\end{pmatrix}$\end{color}\large anticanonical example\LARGE\begin{color}{white}$\begin{pmatrix}.\\.\end{pmatrix}$\end{color} &\ \large $\ \ \begin{pmatrix}0&-i\\i&0\end{pmatrix}\ \ $   &\ \large $\ \ \begin{pmatrix}-1&0\\0&-1\end{pmatrix}\ \ $   &\large $i$ &\ \large $\ \ \begin{pmatrix}-1&0\\0&-1\end{pmatrix}\ \ $    &\ \large $\ \ \begin{pmatrix}1&0\\0&1\end{pmatrix}\ \ $   
%%&\large\ $e^{-\frac{q^2+p^2}{2\hbar}}$
%\\   \hline

\LARGE\begin{color}{white}$\begin{pmatrix}.\\.\end{pmatrix}$\end{color}\large its opposite\LARGE\begin{color}{white}$\begin{pmatrix}.\\.\end{pmatrix}$\end{color} &\ \large $\ \ \begin{pmatrix}0&i\\-i&0\end{pmatrix}\ \ $   &\ \large $\ \ \begin{pmatrix}1&0\\0&1\end{pmatrix}\ \ $   &\large $i$ &\ \large $\ \ \begin{pmatrix}1&0\\0&1\end{pmatrix}\ \ $    &\ \large $\ \ \begin{pmatrix}-1&0\\0&-1\end{pmatrix}\ \ $   
%&\large\ $e^{-\frac{q^2+p^2}{2\hbar}}$
\\   \hline

\LARGE\begin{color}{white}$\begin{pmatrix}.\\.\end{pmatrix}$\end{color}\large anticanonical example\LARGE\begin{color}{white}$\begin{pmatrix}.\\.\end{pmatrix}$\end{color} &\ \large $\ \ \begin{pmatrix}0&-i\\i&0\end{pmatrix}\ \ $   &\ \large $\ \ \begin{pmatrix}-1&0\\0&-1\end{pmatrix}\ \ $   &\large $i$ &\ \large $\ \ \begin{pmatrix}-1&0\\0&-1\end{pmatrix}\ \ $    &\ \large $\ \ \begin{pmatrix}1&0\\0&1\end{pmatrix}\ \ $   
%&\large\ $e^{-\frac{q^2+p^2}{2\hbar}}$
\\   \hline

\end{tabular}

\label{Tab:interactions}

\end{table}
\subsection{Real phase space associated to complex  linear symplectomorphisms}\label{realcomp}
The following result is immediate.
\begin{proposition}\label{prophas}
$_\alpha T_{\bar S}\inv$ and $_\alpha T_{S\inv}$ determine $S$ and $\alpha$.
\end{proposition}
%\begin{proof}
%$_\alpha T_{\bar S}\inv$ determine $b\alpha+id$ and etc, and $_\alpha T_{S\inv}$ determine $b\alpha-id$ and etc. Moreover one has $ad-bc=1$. This gives  five equalities for five unknowns.
%\end{proof}
%\red{AFAIRE
%REPRENDRE LA DICSUSSION SUR ESPECE DE PHSE
%}
As a corollary of Proposition \ref{prophas}, we see that $\R^{\times 2}\times \C=\{((p,q),\beta)\}=T^*\R\times \C$, on which $_\beta T^{-1}_S$ acts, appears as the real phase-space encoding  the action of complex linear canonical transforms.

In the next section, we will give another, noncommutative, interpretation of this phase-space.
\subsection{Noncommutative geometry interpretation }
\label{noncoi}In this section we give a noncommutative interpretation of the off diagonal Toeplitz representation in Theorem \ref{thm1} ([see \cite{tp3} for further details).
\subsubsection{The canonical groupoid}\label{cangroup}\ 

We consider on $\P=T^*\R\times\C^+$ the action of the group $SP(2n,\C)$ 
%with $\G_0=\P$ and 
%$$
%(z,\alpha;z',\alpha')\in\G \Longleftrightarrow\exists S\in SL(2,\C)\mbox{ such that }
%\left\{\begin{array}{l}
%z'={_\alpha T_S)}(z)\\
%\alpha'=S\cdot\alpha
%\end{array}
%\right.
%$$
defined, for any $S\in SP(2n,\C)$ by
$$
(z,\alpha)\to({_\alpha T_S}(z),S\cdot\alpha).
$$
Let us define the groupoid $G$ defined as the semi-direct product $\P\rtimes  SL(2n,\C)$ of $\P$ by $SL(2n,\C)$ \cite[Definition1 p. 104-105 and Section 7]{ac} as  $G=\P\times SL(2n,\C)$, $G^{(0)}=\P\times\{\bf 1\}$ and the functors range and source given by
$$
r((z,\alpha),S)=(z,\alpha),\ \ \ s((z,\alpha),S)=((S(z),S\cdot\alpha)\ \ \forall ((z,\alpha),S)\in\P\times SL(2n,\R).
$$
The $C^*$ algebra associated to the groupoid $G$ is the crossed product $C_0(\P)\rtimes_\Phi SL(2n,\C)$ of the algebra of continuous functions on $\P$ by the action of $SL(2n,\C)$ defined by $\Phi$.
\subsubsection{Symbols}\label{symb}\ 

Let $H$ be a ${_\alpha \mbox{T}}$\"oplitz operator of symbol $\sigma^T_H$ as given by \eqref{alphaalphaprime}. By Theorem \ref{thm1}, we associate to $U(S)^{-1}HU(S)$ the 
%symbol
couple
$$
%\sigma_{U(S)^{-1}HU(S)}^{off}=
((\sigma^T_H)_{S\cdot\alpha,\alpha}\circ {\alpha T}_{\com{S}},\Phi_{S^{-1}\com{S}})
$$
where $(\sigma^T_H)_{S\cdot\alpha,\alpha}$ is given by \eqref{betabetaprime}.

This can be seen as an  element of the algebra associated to the canonical groupoid defined in Section \ref{cangroup} by the following construction:
 we associate to $((\sigma^T_H)_{S\cdot\alpha,\alpha}\circ {\alpha T}_{\com{S}},\Phi_{S^{-1}\com{S}})$ the function $\sigma^{off}[U(S)^{-1}HU(S)]$  on the canonical groupoid identified with $\P\times\P$ defined by
\begin{eqnarray}
&&\sigma^{off}[U(S)^{-1}HU(S)]((z,\alpha),(z',\alpha'))
%\nonumber\\
%&:=&
:=
{_\alpha T}_{\com{S}}\#(\sigma^T_H)_{S\cdot\alpha,\alpha}(z))\delta((z',\alpha')-\Phi_{S^{-1}\com{S}}(z,\alpha))
%{_\alpha\fract_S}
%e^{i\frac{
%%S^{-1}\com{S}(z)\wedge {_\alpha T_{S^{-1}\com{S}}}(z)
%S^{-1}({_\alpha T_{\com{S}}}(z))\wedge({_\alpha T_{S^{-1}\com{S}}(z)+z
%})}\hbar}
,\nonumber
\end{eqnarray}
where ${_\alpha T}_{\com{S}}\#(\sigma^T_H)_{S\cdot\alpha,\alpha}$ designate the push-forward of $(\sigma^T_H)_{S\cdot\alpha,\alpha}$ by ${_\alpha T}_{\com{S}}$.

Conversely,  we ``quantize" the symbol $\sigma_{U(S)^{-1}HU(S)}^{off}$ by the following off-diagonal Toeplitz type quantization formula
\be\label{formquant}
T^{off}[\sigma
%_{U(S)^{-1}HU(S)}
^{off}]
:=
\int_{\P\times\P}\sigma^{off}
%[U(S)^{-1}HU(S)]
((z,\alpha),(z',\alpha'))
\tfrac{|\psi^{\alpha'}_{z'}\rangle\langle\psi^\alpha_z|}
{\langle\psi^{\alpha'}_{z'}|\psi^\alpha_z\rangle}
\tfrac{dzd\bar zdz'd\bar{z'}}{2\pi\hbar}.
\ee
\begin{proposition}
$$
T^{off}[\sigma^{off}[U(S)^{-1}HU(S)]]=U(S)^{-1}HU(S).
$$
\end{proposition}
\subsubsection{On the (formal) composition of symbols}\label{compsymb}\ 

Conjugating an observable by $U(S)$ correspond to a (complex or real) change of variable. Therefore, multiplication of functions should be defined on the same system of coordinates, computationally. This leads to  associate to $U(S)^{-1}HU(S)$ the operator of ``multiplication"  acting on $H'$ given by
\begin{eqnarray}
U(S)^{-1}HU(S)\cdot_S H'&:=&
U(S)^{-1}HH'U(S).\nonumber
%\\
%&=&U(S)^{-1}HU(S)U(S)^{-1}H'U(S).
\end{eqnarray}

This gives rise to the following multiplication of symbols: when $H,H'$ are Toeplitz operators, so is (asymptotically) $HH'$ and its symbol is at leading order the products of the symbols of $H$ and $H'$. therefore the symbol of $U(S)^{-1}HU(S)\cdot_S H'$ is the groupoid composition of the one of $U(S)^{-1}HU(S)$ by th (trivial) one of $H'$. 

In the case where $H':=U(S')^{-1}H^{in}U(S')$

\begin{eqnarray}
U(S)^{-1}HU(S)\cdot_S H'&:=&
U(S)^{-1}HH'U(S).\nonumber
\\
&=&U(S)^{-1}HU(S')^{-1}H'U(S')U(S).\nonumber
\end{eqnarray}
using the result of Theorem \ref{thm1}
\begin{eqnarray}
&&U(S')^{-1}HU(S')\nonumber\\
&=&\int 
{_\alpha T_\com{S'}}\#h_{\com{S'}\cdot \alpha,\alpha}(_\alpha T_{\com{S'}}z)
%|\det{_\alpha T_\com{S'}}|
%e^{i\frac{
%{S'}^{-1}({_\alpha T_{\com{{S'}}}}(z))\wedge({_\alpha T_{{S'}^{-1}\com{{S'}}}(z)-z
%})}\hbar}
\tfrac{\vert\psi^{{S'}^{-1}\com{{S'}}\cdot\alpha}
_{{_\alpha T_{{S'}^{-1}\com{{S'}}}}(z)}
\rangle\langle\psi^{\alpha}_z\vert}
{\langle\psi^{{S'}^{-1}\com{{S'}}\cdot\alpha}
_{{_\alpha T_{{S'}^{-1}\com{{S'}}}}(z)}
\vert\psi^{\alpha}_z\rangle}
\frac{dzd\bar z}{2\pi\hbar},
\nonumber
\end{eqnarray}
and  (formally)
$$
H\vert\psi^{{S'}^{-1}\com{{S'}}\cdot\alpha}
_{{_\alpha T_{{S'}^{-1}\com{{S'}}}}(z)}
\rangle=h({{_\alpha T_{{S'}^{-1}\com{{S'}}}}(z)})
\vert\psi^{{S'}^{-1}\com{{S'}}\cdot\alpha}
_{{_\alpha T_{{S'}^{-1}\com{{S'}}}}(z)}
\rangle +O(\hbar)
$$
we get formally the usual groupoid composition of symbols (see \cite{tp3} for explicit expressions).
 
\subsection{Non canonical transforms}
\label{nomcatrans}
%it exists
%
%exterior action : no intertwining operator
%
%our formula resists

It is striking to notice that the definition of the metaplectic representation as defined by \eqref{udexy}, namely, for a matrix $S=\begin{pmatrix}
a&b\\c&d
\end{pmatrix} $ of determinant one,
 %and $U$ is the metaplectic representation. More precisely $U(S)$ is 
 the operator of integral kernel given by \eqref{udexy},
%\bea\label{udexy}
%U(x,y)&=&\frac1{\sqrt{b2\pi\hbar}} e^{-\frac i{2b\hbar}\left(dx^2-2xy+ay^2\right)}\ \ \ \ \ \ \ \ b\neq 0\nonumber
%%&&\\
%%&=&{\sqrt d}\delta\left(dx  -y\right)e^{i\frac c d\frac{x^2}2}\ \ \ \ \  \ \ \ \ \ \ \ \ \ \ \   b=0\nonumber
%\eea
depends only on the numbers $a,b,d$. The absence of $c$ is hidden by the fact that, thanks to $\det S=1$, $c=\tfrac{ad-1}b$.

On the contrary, the main formula in Theorem \ref{thm1} is expressed directly on the matrix $S$ and therefore admits an extension to the case $\det S\neq 1$. Note that this extension is highly non-trivial also in the real case $S\in M(2,\R)$.

In the present paper, we will limit ourself to the case $\det S=\pm 1$.  
We set
$$
M^\pm(2,\C):=\{S\in SL(2,\C),\ \det S=\pm 1\}.
$$

\begin{definition}\label{defnonca}
Let
%%for $z=z\re+iz\im\in\mathbb C,\ ,\ 
%$\psi^\beta_{z}=
%\frac{e^{-\frac{\beta(x-z\re)^2}{2\hbar}}e^{i\frac{z\im x}\hbar}}{(\pi\hbar/\beta)^\frac14}$, $\beta\re>0$ 
%and
\[
H=\int h(z)\vert\psi^\beta_z\rangle\langle\psi^\beta_z\vert\frac{dzd\bar z}{2\pi\hbar}.
\]
Define, for $S
%=\left(\begin{matrix}a&b\\c&d
%\end{matrix}\right)
\in M^\pm(2,\C),\ \det S\neq0,$ and $\Im\alpha>0$, the real $2\times 2$ matrix ${_\alpha T_S}$ by 
%$$(\bar a d-\bar c b)\re-\beta(\bar d b+\bar ca))\im>0,\ \ \ \ \ \ \ \  (\star)$$
$$
%z_S^\alpha=: 
{_\alpha T_S(z)}=(q_S^\alpha,p_S^\alpha)\in\R^2\mbox{ defined by }q_S^\alpha+\alpha_S(\alpha)p_S^\alpha=q^S+\alpha_S(\alpha)p^S,\ 
\begin{pmatrix}
q^S\\p^S
\end{pmatrix}:=S(z),\ z=(q,p).
$$
Then, for any $S,\alpha$ such that $|V\cdot \alpha|<\infty,\ \Im(V\cdot \alpha)>0,\ V=S^{-1},\com{S}^{-1}, S^{-1}\com{S}$, 
we define the composition operator $\ccomp{S}$ acting on $H$ by
\be
\ccomp{S}H=\int 
 {_\alpha T_{\com{S}}}\# h_{\com{S}\cdot \alpha,\alpha}((-1)^{\frac{1-\det{S}}2}z)
%%%e^{\frac{\beta_{S^{-1}}(T_{S^{-1}}T^{-1}_{\bar S}z)\re^2-2\beta(T^{-1}_{\bar S}z)\re^2+\beta_{\bar S}z\re^2}{2\hbar}}
%%\frac
%%%{e^{\frac{(T_{S^{-1}}T^{-1}_{\bar S}z)\re^2\beta_{S^{-1}}-2(T^{-1}_{\bar S}z)\re^2\beta+z\re^2\beta_{\bar S}-\delta}{2\hbar}}}
%%{e^{\frac{\mathbb Q_S(z)}\hbar)}}
%%{D^\beta_{\bar S}D^\beta_{S^{-1}}}
%%|\det{_\alpha T_\com{S}}|
%{_\alpha\fract_S}
%e^{i\frac{
%%S^{-1}\com{S}(z)\wedge {_\alpha T_{S^{-1}\com{S}}}(z)
%S^{-1}({_\alpha T_{\com{S}}}(z))\wedge({_\alpha T_{S^{-1}\com{S}}(z)-z
%})}\hbar}
\tfrac{\vert\psi^{S^{-1}\com{S}\cdot\alpha}
_{{_\alpha T_{S^{-1}\com{S}}}(z)}
\rangle\langle\psi^{\alpha}_z\vert}
{\langle\psi^{S^{-1}\com{S}\cdot\alpha}
_{{_\alpha T_{S^{-1}\com{S}}}(z)}
\vert I^{\frac{1-\det{S}}2}\psi^{\alpha}_z\rangle}
\frac{dzd\bar z}{2\pi\hbar},\label{defpm}
\ee
where $I$ is the parity operator defined on $L^2(\R)$ by $I\psi(x)=\psi(-x)$.
%where $h_{\com{S}\cdot \alpha,\alpha}$ is defined by \eqref{alphaalphaprime}-\eqref{betabetaprime} and 
%$$
%{_\alpha\fract_S}:=
%\left(\frac{a+b\alpha^{-1}}{a+b\bar\alpha^{-1}}\right)^{-\frac12}\left(-\tfrac{\Im(S\cdot\alpha)^{-1}}{\Im(S\cdot\bar\alpha)^{-1}} \right)^{\frac14}.
%$$
\end{definition}
When $\det{S}=1$, \eqref{defpm} is the same as the result of Theorem \ref{thm1} so that, in this case, $\ccomp{S}\cdot=U(S)^{-1}\cdot U(S)$.
When $\det{S}=-1$, the presence of the operator $I$ in the normalization constant $\frac 1
{\langle\psi^{S^{-1}\com{S}\cdot\alpha}
_{{_\alpha T_{S^{-1}\com{S}}}(z)}
´\vert I\psi^{\alpha}_z\rangle}$  and of the factor $(-1)^{\frac{1-\det{S}}2}$ in the argument of $ {_\alpha T_{\com{S}}}\# h_{\com{S}\cdot \alpha,\alpha}$ follows from the following intuitive arguments (see \cite{tp3} for a rigorous one): 

\noindent - as we have seen right after its statement, the key stone of the proof of Theorem \ref{thm1} was the fact that the normalization constant $L$ ensures  $U(S^{-1})|\psi^\alpha_z\rangle\langle\psi^\alpha_z|U(S)$
 to be a projector. This requirement can be also seen as following the fcat that Wigner functions of pure states composed by canonical transforms satisfy the same equality than the original one, namely
 \be\label{twistconv}
 \int S\#W(z-z')S\#W(z)e^{i\frac{z\wedge z'}\hbar} dz'=S\#W(z)\Leftrightarrow
 \int W(z-z')W(z')e^{i\frac{z\wedge z'}\hbar} dz'=W(z)
 \ee
 since $\det{S}=1\Rightarrow S(z)\wedge S(z')=z\wedge z'$. When $\det(S)=-1$, the left hand side of \eqref{twistconv} becomes
\be\label{twistconv2}
 \int S\#W(z-z')S\#W(z)e^{-i\frac{z\wedge z'}\hbar} dz'=S\#W(z)\Leftrightarrow
 \int W(z-z')W(-z')e^{i\frac{z\wedge z'}\hbar} dz'=W(z)
 \ee 
 leading to, if $R$ denotes the operator of Wigner function $W$, $\ccomp{S}RI\ccomp{S}R=\ccomp{S}R$. This shows easily that $I$ has to be introduced in $\langle\psi^{S^{-1}\com{S}\cdot\alpha}
_{{_\alpha T_{S^{-1}\com{S}}}(z)}
´\vert I\psi^{\alpha}_z\rangle$.

\noindent - in the course of the proof of Theorem \ref{thm1} is used the equality 
\be\label{eqweyl}
U(S^{-1})e^{i\frac{z\wedge Z}\hbar}U(S)=e^{i\frac{z\wedge S(Z)}\hbar}
=e^{i\frac{S^{-1}(z)\wedge Z}\hbar},\ Z=\binom{x}{-i\hbar\frac d{dx}},
\ee
due to the fact that $S$ is canonical.
When $\det{S}=-1$, \eqref{eqweyl} becomes
$$
U(S^{-1})e^{i\frac{z\wedge Z}\hbar}U(S)=e^{i\frac{z\wedge S(Z)}\hbar}
=e^{-i\frac{S^{-1}(z)\wedge Z}\hbar},
$$
responsible for the change $z\to -z$ in $\vert\psi^{S^{-1}\com{S}\cdot\alpha}
_{{_\alpha T_{S^{-1}\com{S}}}(z)}
\rangle\langle\psi^{\alpha}_z\vert$ and therefore in the argument of $ {_\alpha T_{\com{S}}}\# h_{\com{S}\cdot \alpha,\alpha}$ by change of variable in the integration in \eqref{defpm}.

Note again that, on the contrary of the symplectic case, $\ccomp{S}$ is not in general a conjugation.
Nevertheless, since $\ccomp{S}H$ has the form $
\ccomp{S}H=\int f(z)|\ec{\alpha'(\alpha)}_{z'(z)}\rangle\langle\ec{\alpha}_{z}|dz$, one can  extend $\ccomp{S}$, asin the conjugation case, to more general operator than the Toeplitz class and define $\ccomp{S}\ccomp{S'}$ by the same formula as in definition \ref{defnonca} after first replacing  $(z,\alpha)$ by $({{_\alpha T_{S^{-1}\com{S}}}(z)},{S^{-1}\com{S}\cdot\alpha})$ and then multiplying by the weight $$
%h_{\com{S}\cdot \alpha,\alpha}(_\alpha T_{\com{S}}z)
%%%e^{\frac{\beta_{S^{-1}}(T_{S^{-1}}T^{-1}_{\bar S}z)\re^2-2\beta(T^{-1}_{\bar S}z)\re^2+\beta_{\bar S}z\re^2}{2\hbar}}
%%\frac
%%%{e^{\frac{(T_{S^{-1}}T^{-1}_{\bar S}z)\re^2\beta_{S^{-1}}-2(T^{-1}_{\bar S}z)\re^2\beta+z\re^2\beta_{\bar S}-\delta}{2\hbar}}}
%%{e^{\frac{\mathbb Q_S(z)}\hbar)}}
%%{D^\beta_{\bar S}D^\beta_{S^{-1}}}
%|\det{_\alpha T_\com{S}}|
%{_\alpha\fract_S}
%e^{i\frac{
%%S^{-1}\com{S}(z)\wedge {_\alpha T_{S^{-1}\com{S}}}(z)
%S^{-1}({_\alpha T_{\com{S}}}(z))\wedge({_\alpha T_{S^{-1}\com{S}}(z)-z
%})}\hbar}
\left.\tfrac{{_\alpha T_{\com{S}}}\# h_{\com{S}\cdot \alpha,\alpha}(z)}
{{\langle\psi^{S^{-1}\com{S}\cdot\alpha}
_{{_\alpha T_{S^{-1}\com{S}}}(z)}
\vert\psi^{\alpha}_z\rangle}}
\right|_{S=S'}$$
(see \cite{tp3} for further details).

%\subsection{The result}\label{noncares}
With this definition of $\ccomp{S'}
%\ccomp{S}$ $\ccomp{S}
$ is a representation of $M(d,\C)$:
\begin{theorem}
$$
\ccomp{S'}\ccomp{S}=\ccomp{S'S}\ \mbox{ for all } S,S'\mbox{ in }M^\pm(2,\C).
$$ 
\end{theorem} 
 
As a significant example useful in the next section, let us consider the case $S=\begin{pmatrix}0&i\\-i&0\end{pmatrix}$ computed in the second table of Section \ref{annb}. We get, in the case $\alpha=i$,
\be\label{exas}
\ccomp{S}H=\int 
h(z)
%e^{
%-\frac{q^2+p^2}{2\hbar}
%\frac{|z|^2}\hbar}
\vert\psi^{i}
_{-z}
\rangle\langle\psi^{i}_z\vert\frac{dzd\bar z}{2\pi\hbar},
\ee
since $\langle\psi^i_{-z}\vert I\psi^i_z\rangle=\langle\psi^i_{-z}\vert\psi^i_{-z}\rangle=1$.

In other words,  $\ccomp{S}H$ is the quantization of the symbol (with a slight abuse of notation)
\be\label{exassymb}
\sigma^{off}[\ccomp{S}H]((z,i),(z',i))
=
%h(z)
%e^{
%%-\frac{q^2+p^2}{2\hbar}
%\frac{|z|^2}\hbar}\delta(z'+z)
%=
h(z)\delta(z'+z).
%e^{
%%-\frac{q^2+p^2}{2\hbar}
%\frac{|z|^2}{2\hbar}}e^{
%%-\frac{q^2+p^2}{2\hbar}
%\frac{|z'|^2}{2\hbar}}.
\ee
The case $S=\begin{pmatrix}0& -i\\ i&0\end{pmatrix}$ can be treated the same way and leads to, thanks to the same table,
\be\label{exas}
\ccomp{S}H=\int 
h(z)
%e^{
%-\frac{q^2+p^2}{2\hbar}
%\frac{|z|^2}\hbar}
\vert\psi^{i}
_{z}
\rangle\langle\psi^{i}_{-z}\vert\frac{dzd\bar z}{2\pi\hbar},
\ee

 \section{Bose-Einstein-Fermi at the classical level}\label{boseinferm}
 Quantum statistics is a fundamental hypothesis in quantum mechanics. It ensures in particular the stability of matter. On the contrary of many other aspects of non-relativistic quantum mechanics which have a natural ```classical" counterpart, it seems at the first glance  difficult to associate to statistics properties of quantum object a classical corresponding symmetry. Changing the sign after permutation of coordinates of different particle doesn't appeal any classical simple action. Moreover most of the quantities which ``passes" at the limit of vanishing Planck constant are quadratic and therefore looks,   wrongly, as we will see, insensible to the change of sign. Finally, typical fermionic expressions such as exchange term in the Hartree-Fock theory vanishes numerically at the limit $\hbar\to 0$\footnote{Pierre-Louis Lions, hand-written private communication (1990).}.

In this section, which relies on \cite{tp4}, we will implement this ``exchange" action on three (in fact four) different symbols associated to  quantum density matrices: the Husimi function (average of the density matrix on coherent states, therefore a probability density), Wigner functions (that is the Weyl symbol suitably renormalized by a power of the Planck constant in order to be of integral $1$ (but non positive) and the  Toeplitz symbol appearing in the so-called positive quantization procedure. The fourth one will be related to the construction elaborated in Section \ref{nomcatrans}. 

The presentation will be rather formal, rigorous statements and derivations can be found in \cite{tp4}.

\vskip 1cm
%%%%%%%%%%%%%%%%%%%%%%%%%%%%%%pouralex8%%%%%%%%%%%%%%%%%%%%
%- formulas for Husimi, Weyl and Toplitz
%
%- interpretation as  anticanonical transformtion
%
%- link with complex transform
%
%-complex versus cotangent bundle (of the phase-space)
%
%- link with Alex

\begin{definition}
 Let $\rho$ be a density matrix given by an integral kernel $\rho(X;Y),\ X=(x_1,\dots,x_n),\ Y=(y_1,\dots,y_n)$. We define, for $i,j=1,\dots,N$, the mappings 
 
$U_{i\leftrightarrow j}:\ \rho(X;Y)\to U_{i\to j}\rho(X;Y)=\rho(X;Y)|_{y_i\leftrightarrow y_j}$ 

and 

$V_{i\leftrightarrow j}:\ \rho(X;Y)\to V_{i\to j}\rho(X;Y)=\rho(X;Y)|_{x_i\leftrightarrow x_j}$.
\end{definition}
In terms of density matrices, quantum statistics will be seen as looking at density matrices which are eigenvectors of eigenvalue $1$ or $-1$ of the two mappings  $U_{i\leftrightarrow j},V_{i\leftrightarrow j}$.

The indistinguishability property of the quantum system reads as
\be\label{indis}
U_{i\leftrightarrow j}V_{i\leftrightarrow j}
=
V_{i\leftrightarrow j}U_{i\leftrightarrow j},\ \ \ \forall i,j=1,\dots,N.
\ee

\subsection{Husimi}\label{hus}
Let us recall that the Husimi function of a density matrix $\rho$ is defined as
\be\label{dejhus}
\widetilde W[{\rho}](Z,\bar Z)
=
\frac1{(2\pi\hbar)^{dN}}
\langle\varphi_Z|\rho|\varphi_Z\rangle,
\ee
where, for $Z=q+ip\in\bZ^{dN}$ and $x\in\bR^{dN}$, 
\be\label{defec}
\varphi_Z(x)
=
\frac1{(\pi\hbar)^\frac{dN}4}
e^{-\frac{(x-q)^2}{2\hbar}}e^{i\frac{p.x}\hbar}.
\ee
The most elementary properties of the Husimi transform are
\be\label{elemprophus}
\widetilde W[{\rho}]\geq 0
\mbox{ and }\int_{\bZ^{dN}}\widetilde W[{\rho}](Z)dZ=\Tr\rho=1,
\ee and we remark that $\widetilde W[{\rho}]$ is a function anlytic in $Z$ and in $\bar Z$.

Our first link between quantum statistics and the classical underlying space is the contents of the following result.
\begin{lemma}
%Let $\rho$ be a density matrix given by an integral kernel $\rho(X;Y),\ X=(x_1,\dots,x_n),\ Y=(y_1,\dots,y_n)$. Let us consider the mappings 
%
%$U_{i\leftrightarrow j}:\ \rho(X;Y)\to U_{i\to j}\rho(X;Y)=\rho(X;Y)|_{y_i\leftrightarrow y_j}$
%$V_{i\leftrightarrow j}:\ \rho(X;Y)\to V_{i\to j}\rho(X;Y)=\rho(X;Y)|_{x_i\leftrightarrow x_j}$.

Let us consider the Husimi function of $\rho$, $\widetilde W[\rho](Z,\bar Z)$ expressed on the complex variables $Z=(z_1,\dots,z_n)$, $z_l=~q_l+ip_l,\ \bar z_l=q_l-ip_l$.

Then
\[
\widetilde W[{U_{i\leftrightarrow j}\rho}](Z,\bar Z)
=
e^{-\frac{(\bar z_i-\bar z_j)(z_i-z_j)}{2\hbar}}\widetilde W[\rho](Z,\bar Z)|_{z_i\leftrightarrow z_j}
\]
\[
\widetilde W[{V_{i\leftrightarrow j}\rho}](Z,\bar Z)
=
%e^{\frac{(\bar z_i-\bar z_j)(z_i-z_j)}{2\hbar}}
e^{-\frac{|z_i-z_j|^2}{2\hbar}}
\widetilde W[\rho](Z,\bar Z)|_{\bar z_i\leftrightarrow \bar z_j}
\]
\end{lemma}
\vskip 1cm
Note that, as expected,
\[
\widetilde W[{V_{i\leftrightarrow j}U_{i\leftrightarrow j}\rho}](Z,\bar Z)
=
\widetilde W[\rho](z,\bar z)|_{z_i\leftrightarrow z_j,\ \bar z_i\leftrightarrow \bar z_j}
\]
\vskip 1cm
Note also that, with the definition
\be\label{defzpm}
z_\pm=q_\pm+ip_\pm:=\frac{z_i\pm z_j}{\sqrt 2},
\ee
\be\label{eqhus}\nonumber
\begin{pmatrix}
z_i\\z_j\\\bar z_i\\\bar z_j
\end{pmatrix}
\to
\begin{pmatrix}
z_j\\z_i\\\bar z_i\\\bar z_j
\end{pmatrix}   
\Longleftrightarrow\   
\begin{pmatrix}
z_+\\z_-\\\bar z_+\\\bar z_-
\end{pmatrix}
\to
\begin{pmatrix}
z_+\\-z_-\\\bar z_+\\\bar z_-
\end{pmatrix}
\Longleftrightarrow  
\begin{pmatrix}
q_+\\q_-\\p_+\\p_-
\end{pmatrix}
\to
\begin{pmatrix}
q_+\\-ip_-\\p_+\\ iq_-
\end{pmatrix}
\ee
so the complex metaplectic transform associated to the exchange term is the matrix $I_+\otimes S^c_H$ with
\be\label{cSHus}
S^c_H=\begin{pmatrix}
0&-i\\i&0
\end{pmatrix},\ \det{S^c_H}=
-
1.
\ee

\subsection{Wigner}\label{wig}
The Wigner function of a density matrix is nothing but its Weyl symbol, divided by $(2\pi\hbar)^{dN}$. More precisely the Wigner function of $\rho$ is defined as 
\be\label{defwig}
W[\rho](X,\Xi)=\int_{\bR^{2dN}}
\rho(X+\hbar\frac\delta2,X-\hbar\frac\delta2)e^{i\frac{X.\Xi}\hbar}d\delta
\ee
At the contrary of the Husimi function, $W[\rho]$ is not positive, but its main elementary properties are
\begin{eqnarray}\label{elempropwig}
\int_{\bR^{2dN}}W[\rho](X,\Xi)dXd\xi=\Tr\rho&=&1\\
\mbox{ and }&&\nonumber\\
\frac1{(2\pi\hbar)^{dN}}\int_{\bR^{2dN}}W[\rho](X,\Xi0W[\rho'](X,\Xi)dXd\Xi&=&\Tr{(\rho\rho')}.\nonumber
\end{eqnarray}

Let us now define the semiclassical 
%non
 symplectic Fourier transform as 
\[
f({\widehat{q,p}}^\hbar)
=
\frac1{(2\pi\hbar)^d}
\int_{\bR^d\times\bR^d} f(x,\xi)e^{i\frac{q\xi-px}\hbar}dxd\xi.
\]
Note that, at the difference of the usual Fourier transform:
\[
f(\widehat{\widehat{x,\xi}^\hbar}^\hbar)=f(x,\xi)
\]
Let $a_\mp=\frac{a_i\mp a_j}{\sqrt2}$ for $a=q,p,y,\xi$. And let omit the dependence in the variable $q_1,\dots,q_{i-1},q_{i+1},\dots,q_{j-1},q_{j+1},\dots, q_N$ and the same for $p$.

We denote
\[
W^{\frac\pi2}[\rho](x_+,\xi_+;x_-,\xi_-)=W[\rho](x_i,x_j;\xi_i,\xi_j).
\]
\begin{lemma}
\[
W^{\frac\pi2}[U_{i\leftrightarrow j}\rho](q_+,p_+;p_-,q_-)=
W^{\frac\pi2}[\rho](q_+,p_+;\widehat{q_-,p_-}^\hbar)
\]
\[
W^{\frac\pi2}[V_{i\leftrightarrow j}\rho](q_+,p_+;p_-,q_-)=
W^{\frac\pi2}[\rho](q_+,p_+;\widehat{-q_-,-p_-}^\hbar)
\]
%\[
%W[V_{i\leftrightarrow j}\rho](q_1,p_1,\dots,q_i,p_i,\dots,q_j,p_j,\dots,q_n,p_n)
%=\]
%\[
%W[\rho]
%(q_1,p_1,\dots,q_{i-1},p_{i-1},\widehat{-p_j,-q_j}^\hbar,\dots,q_{j-1},p_{j-1},
%\widehat{-q_i,-p_i}^\hbar,\dots,q_n,p_n)
%\]
\end{lemma}
Note that
\[
W[V_{i\leftrightarrow j}U_{i\leftrightarrow j}\rho](q_1,p_1,\dots,q_i,p_i,\dots,q_j,p_j,\dots,q_n,p_n)
=\]
\[
W[\rho]
(q_1,p_1,\dots,q_{i-1},p_{i-1},q_j,p_j,\dots,q_{j-1},p_{j-1},
q_i,p_i,\dots,q_n,p_n)
\]

Let us call now $W^-$ the Wigner function (done with the symplectic Fourier transform) on the two variables $q_-,p_-$, namely, 
%\large
$$\small W^-\big[W^{\frac\pi2}[\rho]\big](q_+,p_+|p_-,q_-;x_-,\xi_-)=$$
$$\small\int\overline{W^{\frac\pi2}[\rho]\big](q_+,p_+,p_-+2\delta\hbar,q_-+2\delta'\hbar)}
$$
$$
%\times
W^{\frac\pi2}[\rho]\big](q_+,p_+,p_--2\delta\hbar,q_--2\delta'\hbar)e^{i(x_-\delta-\xi_-\delta')}
d\delta d\delta'.$$
%\LARGE
\vskip 0.3cm
$ W^-\big[W^{\frac\pi2}[\rho]\big]$ 
 lives on $T^*(\bR^d_{q_+})\times T^*(\bR^{2d}_{(p_-,q_-)})$
equipped with the symplectic form
$$
dq_+\wedge qp_+
+dq_-\wedge d\xi_-+dp_-\wedge dx_-.
$$

One has
\[
W^-\big[W^{\frac\pi2}[U_{i\leftrightarrow j}\rho]\big](q_+,p_+|p_-,q_-;x_-,\xi_-)=
\]
\[
W^-\big[W^{\frac\pi2}[\rho]]
(q_+,p_+|-\xi_-,-x_-;q_-,p_-)
%|_{\substack{p_-\rightarrow -q_-\\q_-\rightarrow p_-\\x_-\rightarrow -\xi_-\\\xi_-\rightarrow x_-}}
\]
%\[
%=
%W^-\big[W^{\frac\pi2}[\rho]]
%(q_+,p_+|\xi_-,-x_-;-q_-,p_-).
%\]
That is, the action of $U_{i\leftrightarrow j}$ on $\rho$ is seen on $ W^-\big[W^{\frac\pi2}[\rho]\big]$ by the pointwise action of the following matrix:
%\[
%S=
%\begin{pmatrix}
%S_+&0\\0&S_-
%\end{pmatrix}
%=
%\begin{pmatrix}
%\begin{pmatrix}1&0&0&0\\
%0&1&0&0\\
%0&0&1&0\\
%0&0&0&1\\
%\end{pmatrix}&0\\
%0&
%\begin{pmatrix}
%0&0&0&-1\\
%0&0&-1&0\\
%0&1&0&0\\
%1&0&0&0\\
%\end{pmatrix}
%\end{pmatrix}
%\mbox { on }
%\begin{pmatrix}
%q_+\\ \xi_+\\ p_+\\ x_+\\ p_-\\ q_-\\ x_-\\ \xi_-
%\end{pmatrix}
%\]
\[
S=
\begin{pmatrix}
S_+&0\\0&S_-
\end{pmatrix}
=
\begin{pmatrix}
\begin{pmatrix}1&0&0&0\\
0&1&0&0\\
0&0&1&0\\
0&0&0&1\\
\end{pmatrix}&0\\
0&
\begin{pmatrix}
0&0&0&1\\
0&0&-1&0\\
0&1&0&0\\
-1&0&0&0\\
\end{pmatrix}
\end{pmatrix}
\mbox { on }
\begin{pmatrix}
q_+\\ \xi_+\\ p_+\\ x_+\\ q_-\\ \xi_-\\ p_-\\ x_-
\end{pmatrix}
\]

and this matrix is symplectic.

Defining now $z_\pm=p_\pm+ix_\pm,\ \theta_\pm=q_\pm+i\xi_\pm$ we find that $S$ becomes on these new variables,
$S^c=(S^c_+,S^c_-)=(I,i\scriptsize{\begin{pmatrix}
0&1\\1&0
\end{pmatrix}}
)
$
%\end{proof}
And so the complex  metaplectic transform  associated is 
\[
S^c_W=\begin{pmatrix}
0&i\\i&o
\end{pmatrix},\ \det{S^C_W}=1.
\]

\subsection{Toeplitz}\label{top}
Let $\rho$ be a  Toeplitz operator of symbol.
$
%\underset{\widetilde{\ \ }}{T}
\utilde W[\rho].
$
This means that $\rho$ can be written as
\be\label{deftop}
\rho
=
\frac1{(2\pi\hbar)^{dN}}
\int_{\bC^{dN}}
\utilde W[\rho](Z,\bar Z)|\varphi_Z\rangle\langle\varphi_Z|dZ
\ee
(here the integral as to be understood in the weak sense on $\cH$).
Elementary properties of $\utilde W[\rho]$ are
\be\label{elemproptop}
\utilde W[\rho]\geq 0\Rightarrow\rho>0,\mbox{ and }\int_{\bC^{dN}}\utilde W[\rho]dZ=\Tr\rho.
\ee
Moreover,the second property of \eqref{elempropwig} can be ``disintegrated" in the following coupling between Husimi and Toeplitz settings:
\be\label{elemhustop}
\int_{\bC^{dN}}\widetilde W[\rho](Z,\bar Z)\utilde W[\rho'](Z,\bar Z)dZ
=
\Tr{(\rho\rho')}.
\ee

\begin{lemma}\ [for $\utilde W$ entire]
%\large
\[
\utilde W[U_{i\leftrightarrow j}\rho](z_i,\bar z_i,z_j,\bar z_j)=
%e^{(|z_i|^2+|z_j|^2)/\hbar}
e^{-\frac{|z_i-z_j|^2}{2\hbar}}
\utilde W[\rho](z_j,\bar z_i,z_i,\bar z_j)
%e^{-\frac{|z_i-z_j|^2}{2\hbar}}
\]
\[
\utilde W[U_{i\leftrightarrow j}\rho](q_-,p_-;q_+,p_+)
=
%e^{\frac{q_i^2+p_i^2+q_j^2+p_j^2}\hbar}
e^{-\frac{q_-^2+p_-^2}{2\hbar}}
\utilde W[\rho](-ip_-,iq_-;q_+,p_+)
\]
%\[
%\utilde W[U_{i\leftrightarrow j}\rho](q_1,p_1,\dots,q_i,p_i,\dots,q_j,p_j,\dots,q_n,p_n)
%=\]
%\[
%\frac{e^{\frac{q_i^2+p_i^2+q_j^2+p_j^2}\hbar}}{(2\pi\hbar)^{2d}}\]
%\[\times
%\utilde W[\rho]
%(q_1,p_1,\dots,q_{i-1},p_{i-1},ip_j,-iq_j,\dots,q_{j-1},p_{j-1},
%ip_i,-iq_i,\dots,q_n,p_n)
%\]
%\[
%=
%\frac1{(2\pi\hbar)^{2d}}{e^{\frac{q_i^2+p_i^2+q_j^2+p_j^2}\hbar}}{(2\pi\hbar)^{2d}}
%\utilde W[\rho]|_{\substack{\ \\z_i\leftrightarrow z_j\\\bar{z_i}\leftrightarrow-\bar{z_j}}},\ z_i=q_i+ip_i.
%\]
\vskip 0.5cm
\[
\utilde W[V_{i\leftrightarrow j}\rho](q_1,p_1,\dots,q_i,p_i,\dots,q_j,p_j,\dots,q_n,p_n)
=
%\frac{
%e^{\frac{q_i^2+p_i^2+q_j^2+p_j^2}\hbar}
e^{-\frac{(q_i-q_j)^2+(p_i-p_j)^2}{2\hbar}}
%}{(2\pi\hbar)^{2d}}
\]
\[\times
\utilde W[\rho]
(q_1,p_1,\dots,q_{i-1},p_{i-1},-ip_j,iq_j,\dots,q_{j-1},p_{j-1},
-ip_i,iq_i,\dots,q_n,p_n)
\]
\[
=
%\frac1{(2\pi\hbar)^{2d}}
%{e^{\frac{q_i^2+p_i^2+q_j^2+p_j^2}\hbar}}
e^{-\frac{(q_i-q_j)^2+(p_i-p_j)^2}{2\hbar}}
%{(2\pi\hbar)^{2d}}
\utilde W[\rho]|_{\substack{\ \\z_i\leftrightarrow -z_j\\\bar{z_i}\leftrightarrow\bar{z_j}}},\ z_i=q_i+ip_i.
\]

%\[
%W[V_{i\leftrightarrow j}\rho](q_1,p_1,\dots,q_i,p_i,\dots,q_j,p_j,\dots,q_n,p_n)
%=\]
%\[
%W[\rho]
%(q_1,p_1,\dots,q_{i-1},p_{i-1},\widehat{-p_j,-q_j}^\hbar,\dots,q_{j-1},p_{j-1},
%\widehat{-q_i,-p_i}^\hbar,\dots,q_n,p_n)
%\]
\end{lemma} 

In other words, the exchange action on the Toeplitz symbol is the same as the one on the Husimi function.
%, modulo a different Gaussian weight.
\subsection{On Wigner again}\label{ow}
Let us denote
\[
 U_{i\leftrightarrow j}^WW[\rho]=W[U_{i\leftrightarrow j}\rho]
 \]
 Let us moreover denote by $W^2[\rho]$ the Wigner function of the Wigner function of $\rho$ (see footnote 1):
 \[
 W^2[\rho]=W[W[\rho]].
 \]
 Let us denote by $Q_i=(q_i,\xi_i)$ and $P_i=(p_i,x_i),\ i=1,\dots, N,$ the variables in $T^*(T^*\bR^d))$. We define:
 \[
 Q_i^t=(\xi_i,q_i),\ \ \ P_i^t=(x_i,p_i).
 \]
 
 \begin{lemma}
 %\large
\[
 W^2[U_{i\leftrightarrow j}\rho](Q_1,P_1,\dots,Q_i,P_i,\dots,Q_j,P_j,\dots,Q_n,P_n)
=\]
\[
 W^2[\rho]
(Q_1,P_1,\dots,Q_{i-1},P_{i-1},P^t_j,-Q^t_j,\dots,Q_{j-1},P_{j-1},
P^t_i,-Q^t_i,\dots,Q_n,P_n)
\]
\[
 W^2[U_{i\leftrightarrow j}\rho]=W[U^W_{i\leftrightarrow j}W[\rho]]=
W^2[\rho]|_{\substack{\ \\Q_i\leftrightarrow P^t_j\\P_i\leftrightarrow-Q^t_j}}.
\]
\vskip 1cm
\[
 W^2[V_{i\leftrightarrow j}\rho](Q_1,P_1,\dots,Q_i,P_i,\dots,Q_j,P_j,\dots,Q_n,P_n)
=\]
\[
 W^2[\rho]
(Q_1,P_1,\dots,Q_{i-1},P_{i-1},-P^t_j,Q^t_j,\dots,Q_{j-1},P_{j-1},
-P^t_i,Q^t_i,\dots,Q_n,P_n)
\]
\[
 W^2[V_{i\leftrightarrow j}\rho]=W[V^W_{i\leftrightarrow j}W[\rho]]=
W^2[\rho]|_{\substack{\ \\Q_i\leftrightarrow -P^t_j\\P_i\leftrightarrow Q^t_j}}.
\]
%\[
%\utilde W[V_{i\leftrightarrow j}\rho](q_1,p_1,\dots,q_i,p_i,\dots,q_j,p_j,\dots,q_n,p_n)
%=\]
%\[
%\frac{e^{\frac{q_i^2+p_i^2+q_j^2+p_j^2}\hbar}}{(2\pi\hbar)^{2d}}\]
%\[\times
%\utilde W[\rho]
%(q_1,p_1,\dots,q_{i-1},p_{i-1},-ip_j,iq_j,\dots,q_{j-1},p_{j-1},
%-ip_i,iq_i,\dots,q_n,p_n)
%\]
%\[
%=
%\frac1{(2\pi\hbar)^{2d}}{e^{\frac{q_i^2+p_i^2+q_j^2+p_j^2}\hbar}}{(2\pi\hbar)^{2d}}
%\utilde W[\rho]|_{\substack{\ \\z_i\leftrightarrow -z_j\\\bar{z_i}\leftrightarrow\bar{z_j}}},\ z_i=q_i+ip_i.
%\]
%
%%\[
%%W[V_{i\leftrightarrow j}\rho](q_1,p_1,\dots,q_i,p_i,\dots,q_j,p_j,\dots,q_n,p_n)
%%=\]
%%\[
%%W[\rho]
%%(q_1,p_1,\dots,q_{i-1},p_{i-1},\widehat{-p_j,-q_j}^\hbar,\dots,q_{j-1},p_{j-1},
%%\widehat{-q_i,-p_i}^\hbar,\dots,q_n,p_n)
%%\]
\end{lemma}
So $U^W_{i\leftrightarrow j},\ V_{i\leftrightarrow j}$ are metaplectic operators  associated to canonical transforms on $T^*(T^*(\bR^{dN}))$.
\begin{lemma}
Denoting now $z_i=q_i+\xi_i,\ \theta_i=p_i+ix_i$ we have
\[
W[U^W_{i\leftrightarrow j}W[\rho]]=W^2[U_{i\leftrightarrow j}\rho]=W^2[\rho]|_{\substack{z_i\leftrightarrow iz_j\\
\theta_i\leftrightarrow i\theta_j}}
\]
\[
W^2[V_{i\leftrightarrow j}\rho]=W^2[\rho]|_{\substack{z_i\leftrightarrow -iz_j
\\
\theta_i\leftrightarrow -i\theta_j}}
\]
\end{lemma}
%So 
So $U^W_{i\leftrightarrow j},\ V_{i\leftrightarrow j}$ are metaplectic operators  associated to complex canonical transforms on the complexification of $T^*(\bR^{dN})$.

\subsection{Off-diagonal Toeplitz representations}\label{offtop}

In this section, we take $d=1$ and $N=2$. 

A density matrix $\rho$ has an integral kernel $\rho(x_1,x_2;y_1,y_2)$ and 
$$
(U\rho)(x_1,x_2;y_1,y_2)=\rho(x_1,x_2;y_2,y_1)
$$
$$
(V\rho)(x_1,x_2;y_1,y_2)=\rho(x_2,x_1;y_1,y_2).
$$
therefore, performing a change of variables 
$$x=(x_1-x_2)/\sqrt2, x' =(x_1+x_2)/\sqrt2,$$
$$
y=(y_1-y_2)/\sqrt2,y'=(y_1-y_2)/\sqrt2,
$$
one get, with a slight abuse of notation that
$$
U\rho(x,y;x',y')=\rho(x,-y:x',y')
$$
$$
V\rho(x,y;x',y')=\rho(-x,y:x',y')
$$ 
In the rest of this section we will omit the variables $x',y'$.

Let us consider a (generalized) Toeplitz operator
%$\psi^\beta_{z}=
%\frac{e^{-\frac{\beta(x-q)^2}{2\hbar}}e^{i\frac{p x}\hbar}}{(\pi\hbar/\beta)^\frac14}$, $\beta>0,\ z=q+ip$ 
%and
\[
H=\int h(z)\vert\psi^\beta_z\rangle\langle\psi^\beta_z\vert\frac{dzd\bar z}{2\pi\hbar},
\]
where, for $\beta>0,\ z=q+ip$,
$$
\psi^\beta_{z}=
\frac{e^{-\frac{\beta(x-q)^2}{2\hbar}}e^{i\frac{p x}\hbar}}{(\pi\hbar/\beta)^\frac14}.$$
Let us define $H^l$ by its integral kernel $H^l(x,y)=H(-x,y)$ where $H(x,y)$ is the integral kernel of $H$. Let $H^r$ be defined the same way by $H^r(x,y)=h(x,-y)$.

Obviously
\newcommand{\lr}{{\substack{l\\r}}}
\[
H^\lr=\int h(z)\vert\psi^\beta_{\mp z}\rangle\langle\psi^\beta_{\pm z}\vert\frac{dzd\bar z}{2\pi\hbar}.
\]
%Let us define $S_\lr=\begin{pmatrix}
%0&\mp i\\\pm i&0
%\end{pmatrix}$. Then $S^{-1}=\begin{pmatrix}
%0&\mp i\\\pm i&0
%\end{pmatrix}$, $\bar S=\begin{pmatrix}
%0&\pm i\\\mp i&0
%\end{pmatrix}$ and $\bar S^{-1}=\begin{pmatrix}
%0&\pm i\\\mp i&0
%\end{pmatrix}$ so that $S^{-1}\bar S^{-1}=\begin{pmatrix}
%-1&0\\0&-1
%\end{pmatrix}$.
%
%Moreover, in the case $r$, $\beta_{\bar S}=\frac1\beta$ and $_\beta T_{\bar S}(q,p)=(q',p')$ with
%\[
%\frac1\beta q'+ip'=\frac1\beta ip+i(-iq).
%\]
%So $_\beta T_{\bar S_\lr}=\pm\begin{pmatrix}
%\beta&0\\0&\frac1\beta
%\end{pmatrix}$. Also $_\beta T_{S_\lr^{-1}}=\pm\begin{pmatrix}
%-\beta&0\\0&-\frac1\beta
%\end{pmatrix}$.
%
%We will take $\beta=1$.
%
%We have $D^1_{\bar S}=\frac1{(-1)^{1/2}},\ D^1_{S^{-1}}=\frac1{(i)^{1/2}}$ so 
%$D^1_{\bar S}D^1_{S^{-1}}=1$.
%
%$\delta=-2\bar z^2$, $\mathbb Q=q^2-p^2$
%
%So
%\[
%H^{S\lr}=\int h(q,p)e^{\frac{q^2-p^2}\hbar}
%\vert\psi_{\mp z}\rangle\langle\psi_{\pm z}\vert\frac{dzd\bar z}{2\pi\hbar}.
%\]
Therefore, we get the following off-diagonal expressions.
\begin{lemma}\label{offdiag}
\bea
VH&=&\int h(q,p)
\vert\psi_{-z}\rangle\langle\psi_z\vert\frac{dzd\bar z}{2\pi\hbar}\nonumber\\
UH&=&\int h(q,p)
\vert\psi_{z}\rangle\langle\psi_{-z}\vert\frac{dzd\bar z}{2\pi\hbar}\nonumber\\
UVH&=&\int h(q,p)
\vert\psi_{-z}\rangle\langle\psi_{-z}\vert\frac{dzd\bar z}{2\pi\hbar}\nonumber\\
U^2=V^2&=&1\nonumber
\eea
\end{lemma}
These expressions have to be compared to the following ones, derived form Section \ref{top}.
\begin{lemma}
\bea
VH&=&\int h(ip,-iq)e^{-\frac{q^2+p^2}{2\hbar}}
\vert\psi_{z}\rangle\langle\psi_z\vert\frac{dzd\bar z}{2\pi\hbar}\nonumber\\
UH&=&\int h(-ip,iq)e^{-\frac{q^2+p^2}{2\hbar}}
\vert\psi_{z}\rangle\langle\psi_{z}\vert\frac{dzd\bar z}{2\pi\hbar}\nonumber
\\
UVH&=&\int h(-q,-p)
\vert\psi_{z}\rangle\langle\psi_{z}\vert\frac{dzd\bar z}{2\pi\hbar}\nonumber
\eea
\end{lemma}
The Toeplitz symbol of $VH$ (resp. $UH$) is $h_V(q,p)=h(ip,-iq)e^{-\frac{q^2+p^2}{2\hbar}}$ (resp.  $h_U(q,p)=h(-ip,iq)e^{-\frac{q^2+p^2}{2\hbar}}$).
\begin{lemma}
\noindent Let $h\geq 0, \int h=1$. 

Then $H^B:=\tfrac14(H+VH+UH+UVH)$ is a bosonic state,

\  and $H^F:=\tfrac14(H-VH-UH+UVH)$ is a fernionic one.
\end{lemma}
\begin{proof}
One has $H^B=VH^B=UH^B=UVH^B$, $\mbox{Tr }H^B=1$,
 $H^F=-VH^B=-UH^B=UVH^B$, $\mbox{Tr }H^B=1$, and 
\[
H^B=\tfrac14\int h(q,p)
\vert\psi_{z}+\psi_{-z}\rangle\langle\psi_z+\psi_{-z}\vert\frac{dzd\bar z}{2\pi\hbar}\geq 0.
\]
\[
H^F=\tfrac14\int h(q,p)
\vert\psi_{z}-\psi_{-z}\rangle\langle\psi_z-\psi_{-z}\vert\frac{dzd\bar z}{2\pi\hbar}\geq 0.
\]
\end{proof}
Finally, $H^B$ is ``semiclassical".
\subsection{Link with the complex (anti)metaplectic representation}\label{linkmeta}
We have seen in the previous (sub)sections that $U$ (resp. $V$) is associated to the action of the matrix $\begin{pmatrix}
0&-i\\i&0
\end{pmatrix}$ (resp. $\begin{pmatrix}
0&i\\-i&0
\end{pmatrix}$) on the Husimi function and the Toeplitz symbol. Therefore it is natural to think that $U$ (resp. $V$) should be associated to the ``metaplectic" quantization of $\begin{pmatrix}
0&-i\\i&0
\end{pmatrix}^{-1}=\begin{pmatrix}
0&i\\-i&0
\end{pmatrix}$ (resp. $\begin{pmatrix}
0&i\\-i&0
\end{pmatrix}^{-1}=\begin{pmatrix}
0&-i\\i&0
\end{pmatrix}$),  ``metaplectic" because these matrices are not canonical. Precisely, a definition of quantization of anticanonical mappings has been provide in the preceding section that we can use in the present situation.

With the 
%notation of section \ref{off} we can make the following observations.
%
%Let us define $S_\lr=\begin{pmatrix}
%0&\mp i\\\pm i&0
%\end{pmatrix}$. Then $S^{-1}=\begin{pmatrix}
%0&\mp i\\\pm i&0
%\end{pmatrix}$, $\bar S=\begin{pmatrix}
%0&\pm i\\\mp i&0
%\end{pmatrix}$ and $\bar S^{-1}=\begin{pmatrix}
%0&\pm i\\\mp i&0
%\end{pmatrix}$ so that $S^{-1}\bar S^{-1}=\begin{pmatrix}
%-1&0\\0&-1
%\end{pmatrix}$.
%
%Moreover, in the case $r$, $\beta_{\bar S}=\frac1\beta$ and $_\beta T_{\bar S}(q,p)=(q',p')$ with
%\[
%\frac1\beta q'+ip'=\frac1\beta ip+i(-iq).
%\]
%So $_\beta T_{\bar S_\lr}=\pm\begin{pmatrix}
%\beta&0\\0&\frac1\beta
%\end{pmatrix}$. Also $_\beta T_{S_\lr^{-1}}=\pm\begin{pmatrix}
%-\beta&0\\0&-\frac1\beta
%\end{pmatrix}$.
%
%We will take $\beta=1$.
%
%We have $D^1_{\bar S}=\frac1{(-1)^{1/2}},\ D^1_{S^{-1}}=\frac1{(i)^{1/2}}$ so 
%$D^1_{\bar S}D^1_{S^{-1}}=1$.
%
%$\delta=-2\bar z^2$, $\mathbb Q=q^2-p^2$
%
%So
%\be\label{hslr}
%H^{S\lr}=\int h(q,p)e^{\frac{q^2-p^2}\hbar}
%\vert\psi_{\mp z}\rangle\langle\psi_{\pm z}\vert\frac{dzd\bar z}{2\pi\hbar}.
%\ee
%Therefore, with the metaplectic representation $U(S)$ defined in section \ref{off}, Theorem 1, we get the following identities, leading, finally, to a direct metaplectic representation for the exchange map, but associated to an anti-canonical relation.
%
%Indeed, using \eqref{hslr}, Lemma \ref{offdiag} and the definition of $U(S)$ in section \ref{off} Theorem 1, 
definition of $\ccomp{S}$  in  \cite
%[Definition 1.1, Section 9.1]
{tp3}  recalled in Section \ref{nomcatrans} above,
we get  our final result, as a direct application of \eqref{exas}.
\begin{lemma}
Let $H$ a Toeplitz operator of symbol $h(q,p)$. Then
\bea
%VH&=&
%%U\big(\small\begin{pmatrix}
%%0&i\\-i&0
%%\end{pmatrix}\big)^{-1}H'
%%U\big(\small\begin{pmatrix}
%%0&i\\-i&0
%%\end{pmatrix}\big)
%\ccomp{\scriptsize\begin{pmatrix}
%0&-i\\i&0
%\end{pmatrix}}H
%\nonumber\\
UH&=&
%U\big(\small\begin{pmatrix}
%0&-i\\i&0
%\end{pmatrix}\big)^{-1}H'
%U\big(\small\begin{pmatrix}
%0&-i\\i&0
%\end{pmatrix}\big)
\ccomp{\scriptsize\begin{pmatrix}
0&i\\-i&0
\end{pmatrix}}H.\nonumber\\
VH&=&
%U\big(\small\begin{pmatrix}
%0&i\\-i&0
%\end{pmatrix}\big)^{-1}H'
%U\big(\small\begin{pmatrix}
%0&i\\-i&0
%\end{pmatrix}\big)
\ccomp{\scriptsize\begin{pmatrix}
0&-i\\i&0
\end{pmatrix}}H
\nonumber
\eea
%where 
%$H'$ is the T\"oplitz operator of symbol $h'(q,p)=h(q,p)e^{-\frac{q^2+p^2}{\hbar}}$.

%$\ccomp{S}$ is defined in  \cite[Definition 1.1, Section 9.1]{tp3} and is recalled in Section \ref{nomcatrans} above.
\end{lemma}

But the ``true" result is the following, that we express only  for $U$, the case $V$ being straightforwardly the same).
\begin{proposition}\label{propnoncan}
$$
UH=T^{off}\left[\sigma^{off}[\ccomp{\scriptsize\begin{pmatrix}
0&i\\-i&0
\end{pmatrix}}H]
%\times \left(e^{-\frac{|\cdot|^2}{2\hbar}}\right)^{\otimes 2}
\right],
$$
where $\sigma^{off}[\ccomp{\scriptsize\begin{pmatrix}
0&i\\-i&0
\end{pmatrix}}H]$ is defined by \eqref{exassymb} and  $T^{off}$ by the off-diagonal Toeplitz quantization formula \eqref{formquant}.
% and
%$$
%\left(\sigma^{off}[\ccomp{\scriptsize\begin{pmatrix}
%0&i\\-i&0
%\end{pmatrix}}H]
%\times \left(e^{-\frac{|\cdot|^2}{2\hbar}}\right)^{\otimes 2}\right)(z,z')=
%\sigma^{off}[\ccomp{\scriptsize\begin{pmatrix}
%0&i\\-i&0
%\end{pmatrix}}H](z,z')
%e^{-\frac{|z|^2}{2\hbar}}e^{-\frac{|z'|^2}{2\hbar}}.
%$$

Namely,  $UH$ is given by the off-diagonal Toeplitz quantization of the off-diagonal Toeplitz symbol of $ccomp{\scriptsize\begin{pmatrix}
0&i\\-i&0
\end{pmatrix}}H$  {\bf without} the multiplication
 %(in the two variables $(z,z')$) 
 by the factor $e^{-\frac{|\cdot|^2}{\hbar}}$ 
 %that we already met in 
 as for the Husimi and the (diagonal) Toeplitz cases, as seen in the previous sections.

\end{proposition}
Proposition \ref{propnoncan} shows clearly first that the exchange mappings $U,V$ are clearly associted to complex non-canonical linear transformations, and second that the off-diagonal Toeplitz quantization/representation of $Mpm(2,\C)$ established in Section \ref{nomcatrans}, Definition \ref{defnonca}, is meaningful.

Note again that $\mbox{det }{\small\begin{pmatrix}
0&i\\-i&0
\end{pmatrix}}=\mbox{det }\small\begin{pmatrix}
0&-i\\i&0
\end{pmatrix}=-1$.

\subsection{A classical phase space with symmetries inherited form quantum statistics}\label{clasphasstat}
%\red{AFAIRE
%- EXTENSION AUX ANTISYMPLECTIQUES
%
%- LIEN AVEC COMPLEXIFICATION
%
%- VERS  QUANTIFICATIONS DE SMAPIINGS LINEAORES GENERAUX (NI SYMOL, NI ANTISYMPL.
%
%}
%AFAIRE
The construction of the preceding (sub)section suggests that a noncommutative  extension of the usual phase space of classical mechanics, namely the cotangent bundle of the configuration space, is possible in order to handle the trace, at the classical underlying level, of more  symmetries, coming from the quantum one, than the one usually considered: namely the unitary in a Hilbert space of the quantum propagation leading to the symplectic classical evolution associated to $SL(2,\R)$. One recovers  the presence of the fundamental (as responsible, e.g., of the stability of matter) spin-statistics symmetries at the classical level by extending the group of symmetry  $SL(2,\R)$ acting pointwise on the phase-space to $M^pm(2,\C)=\{S\in M(2,\C), \det{S}=\pm 1\}$ with its action on a noncommutative space  established in Section \ref{nomcatrans}.

 \section{Noncommutative moduli spaces underlying  topological quantum field theory in large coloring asymptotics}\label{tqft}
 \newcommand{\su}{\mathrm{SU}_2}
 
 Quantum Field Theory (QFT) is an extension of quantum mechanics to situations where particles interact with fields. It involve a quantization procedure of fields whose ``philosophy" is the same as the one of regular quantum mechanics, but whose  physical and mathematical realization is far from being as complete as for quantum mechanics.  In fact, the modern formulation of QFT consists ``only"\footnote{Let us quote Claude Itzykson when he was saying that the most difficult QFT is the one in dimension zero, that is regular quantum mechanics, because it is the one  to which one asks the most precise questions.} in computing amplitudes of transitions through  Feynman integrals of the (highly non rigorous) form
 \be\label{feyn}
I_{S}=\int e^{i\frac{S(\Phi)}\hbar}\cD[\phi]
 \ee
 where the ``integral" is over families of fields $\Phi$ and $S(\Phi)$ is a classical action, namely an integro-differential functional of $\Phi$. the ``measure" $\cD$ is supposed to be an extension of the Lebesgue measure in highly infinite dimension.
 
 The enormous success of the formal and mysterious formula \eqref{feyn} comes from the fact that it can be perturbatively computed when $S(\Phi)=S_o(\Phi)+\eps S_1(\Phi)$ where $S_0$ is a quadratic functional. This leads to an exact computation of $I_{S_0}$ by extension  of Gaussian integrals to infinite dimension, and asymptotic computation of $I_{S_o+\eps S_1}$ by some perturbative calculus. Meanwhile,  when the action has a large group of symmetries, $I_S$ reduces to an integral over the quotient of the set of fields by the group of symmetries and allows sometime non-perturbative (and physically interesting) computations.
 
 This is precisely the case for Topological Quantum Field Theory (TQFT), for which, roughly speaking,  one impose to the action $S$ to be sensitive only to topological underlying structures (and not, e.g., the geometrical ones). Then the formula \eqref{feyn} produces  topological invariants which happen to be directly linked to numerous  domains in mathematics like knot  theory, moduli spaces
in
algebraic geometry and low dimension topology, and in physics such as fractional quantum Hall effect and condensed matter physics.

TQFT \`a la Witten  \cite{witten} is a physical model for the Jones polynomial of knots. Without defining here precisely the objects involved (see \cite{mp} for some details) the invariant is then the partition function $Z_r(M,K)$ associated to any knot $K$ in a 3-manifold $M$ expressed  as a Feynman integral over all connections $A$ on some $G$-bundle over $M$, for  a compact group $G$, of  the form 
\be\label{feyntqft}
Z_r(M,K)=\int  e^{ir\textrm{CS}(A)}\Tr_V (\textrm{Hol}_K A)\mathrm{d}A
\ee 
where $\Tr_V$ is the trace associated to a representation $V$ of $G$ and Hol$_K(A)$ is the holonomy of the connection $A$ along $K$ and CS$(A)$ is the Chern-Simons functional.

The analogy between $\eqref{feyn}$ and \eqref{feyntqft} is clear: the fields are the connections $A$, the ``measure" $\cD(A)=\Tr_V (\textrm{Hol}_K A)\mathrm{d}A$, the action is $CS$ and, very important for us, $$r=\tfrac1\hbar.$$
Therefore the asymptotics $r\sim\infty$ correspond to the semiclassical regime $\hbar\sim 0$, for which the stationary phase principle implies that the integralion in \eqref{feyntqft} should concentrate on critical points of the Chern-Simons functional, that is connections which are flat on $M\setminus K$. Moreover, if $M$ has boundary $\Sigma$, then $Z_r(M,K)$ should be interpreted as an element of the geometric quantization of the moduli space $\boM(\Sigma,G)$ i.e. the space of  gauge equivalence classes of flat $G$-connections on $\Sigma$.
.
 
 To any   closed oriented surface $\Sigma$ with $n$ marked points $p_1,\ldots,p_n$, any integer $r>0$ and any coloring $c=(c_1,\ldots,c_n), c_i\in\{1,\dots,r-1\}$ of the marked points, TQFT provides, by the construction of \cite{bhmv}, a finite dimensional hermitian vector space $V_r(\Sigma,c)$ together with a basis $\{\phi_n,n=1,\dots,\dim{(V_r(\Sigma,c))}\}$ of this space (see Sections 2.1 and 2.5 in \cite{mp}).

 On the other (classical) side, to each $t\in(\pi\mathbb Q)^n$ we can associate the  moduli space:
$$\mathcal{M}(\Sigma,t)=\{\rho:\pi_1(\Sigma\setminus\{p_1,\ldots,p_n\})\to\su\text{ s.t. }\forall i,\ tr\rho(\gamma_i)=2\cos(t_i)\}/\sim$$ 
where one has $\rho\sim \rho'$ if there is $g\in$ SU$_2$ such that $\rho'=g\rho g^{-1}$ and  $\gamma_i$ is any curve going around $p_i$. 

When $\Sigma$ is either a once punctured torus or a 4-times punctured sphere, $\mathcal{M}(\Sigma,t)$ is symplectomorphic to the standard sphere $S^2=\C P^1$.

To any curve $\gamma$ on the surface $\Sigma$ (that is, avoiding the marked points) we can associate two objects: a quantum one, the curve operator $T_\gamma$ acting on $V_r(\Sigma,c)$, and a classical  one,  the function $f_\gamma$ on the symplectic manifold $\mathcal{M}(\Sigma,t)$.
\begin{itemize}
\item $T_\gamma$ is obtained by a combinatorial topological construction recalled in Sections 2.3 and 2.4 in \cite{mp}. By the identification of the finite dimensional space $V_r(\Sigma,c)$ with the Hilbert space of the quantization of the sphere $\mathcal H_N$ defined in \cite[Section 2.1]{tp2} with $N:=\dim{(V_r(\Sigma,c))}$, through $\{\phi_n,n=1,\dots,\dim{(V_r(\Sigma,c))}\}\leftrightarrow \{\psi^N_n,n=1,\dots,N\}$, $T_\gamma$ can be seen as a matrix on $\mathcal H_N$. One of the main results of  \cite{mp} was to prove than this matrix is a trigonometric one in the sense of \cite[Section 3.2]{tp2}.
\item $f_\gamma:\ \mathcal{M}(\Sigma,t)\to [0,\pi]$ is defined by 
\be\label{fgamma}
\rho\mapsto f_\gamma(\rho):= -tr \rho(\gamma).
\ee
\end{itemize}
The asymptotics considered in \cite{mp} consists in letting $r\to\infty$ and  considering a sequence of colourings $c_r$ such that $\pi\frac{c_r}r$ converges to $t$ and the dimension of $V_r(\Sigma,c_r):=N$, grows linearly with $r$. One sees immediately that, by the identification  $V_r(\Sigma,c)\leftrightarrow \mathcal H_N$, this correspondents to the semiclassical asymptotics $N\to\infty$.

%The asymptotism considered in \cite{mp} consists in letting $r\to\infty$ and  considering a sequence of colorings $c_r$ such that $\pi\frac{c_r}r$ converges to $t$ and the dimension of $V_r(\Sigma,c_r):=N$, grows linearly with $r$. One sees immediately that, by the identification  $V_r(\Sigma,c)\leftrightarrow \mathcal H_N$, this correspondents to the semiclassical asymptotism $N\to\infty$.  

The main result of \cite{mp} states that, {\bf for generic values of $t$}, 
$$
T_\gamma\mbox{ is a 
%(standard) 
Toeplitz operator of leading symbol }f_\gamma.
$$
the generic values of $t$ are the one for which $f_\gamma$ considered as a function on $S^2$ by the symplectic isomorphism mentioned earlier, belongs to $C^\infty(S^2)$.
\vskip .5cm
The present section relies on \cite{tp2} and concerns the asymptotics $\pi\frac{c_r}r\to t$ for \textbf{any value} of $t\in[0,\pi]$.
\vskip 0.5cm
In Section \ref{canonical} we present the geometric quantization of the sphere, together with the definition of Toeplitz operators we just mentioned. Then we  state in Section \ref{at} the extension of the Toeplitz calculus necessary to get rid of the generiticity condition  for the asymptotic $t\to\infty$. Our main result is contained in Section \ref{mrtqft} and some geometrical considerations on the singular (in fact noncommutative) underlying classical limit space are given in Section \ref{classical}.

\subsection{The standard geometric quantization of the sphere}\label{canonical}
In this section we will consider the quantization of the sphere in a very down-to-earth way. See \cite{folland,mp} for more details.

Given an integer $N$, we define the space $\mathcal H_N$  of polynomials in the complex variable $z$ of order strictly less than
$N$ and set 
\be\label{basis}
\la P,Q\ra=\frac{i}{2\pi}\int_\C
\frac{P(z) \overline{Q(z)}}{(1+|z|^2)^{N+1}}
dzd\bar z\quad \text{and}\quad\vp^N_n (z)=\sqrt{\frac{N!}{n!(N-1-n)!}}z^n
\ee
The vectors $(\vp^N_n)_{n=0\dots N-1}$ form an orthonormal basis of $\mathcal H_N$.

By the stereographic projection $$S^2\ni (\tau,\theta)\in [0,1]\times S^1\to z=\sqrt{\frac \tau{1-\tau}}e^{i\theta}\in \C\cup\{\infty\},$$ 
The space $\mathcal H_N$ can be seen as a space of functions on the sphere (with a specific behaviour at the north pole).
Write 
\be\label{dmun}
d\mu_N=\frac{i}{2\pi}\frac{dzd\bar z}{(1+|z|^2)^{N+1}}.
\ee
 As a space of analytic functions in $L^2(\C,d\mu_N)$,  the space $\mathcal  H_N$ is closed.

For $z_0\in \C$, we define the coherent state 
\be\label{repro}\rho_{z_0}(z)=N(1+\bar z_0 z)^{N-1}.
\ee
 These vectors satisfy $\la f,\rho_{z_0}\ra=f(z_0)$ for any $f\in \mathcal H_N$ and the orthogonal projector $\pi_N:L^2(\C,d\mu_N)\to \mathcal H_N$ satisfies $(\pi_N\psi)(z)=\la \psi,\rho_z\ra$.
\newcommand{\ci}{C^{\infty}}

For $f\in \ci (S^2,\R)$ we define the (standard) Toeplitz quantization of $f$ as the  operator 
\begin{eqnarray}
T^N[f]:\mathcal H_N&\to&\mathcal H_N\nonumber\\
T^N[f]&:=&\int_{\C}f(z)|\rho_z\rangle\langle\rho_z|d\mu_N(z)\nonumber\\
\mbox{ i.e. }T^N[f]\psi&:=&\int_\C f(z)\langle\rho_z,\psi\rangle_{\mathcal H_N} \rho_zd\mu_N(z)=\pi_N(f\psi)\mbox{ for } \psi\in\mathcal H_N.\label{deftop}
\end{eqnarray}
 A Toeplitz operator on $S^2$ is a sequence of operators $(T_N)\in \mathrm{End} (\mathcal H_N)$ such that 
there exists a sequence $f_k\in\ci(S^2,\R)$ such that for any integer $M$ the operator $R^M_N$ defined by the equation $$T_N=\sum_{k=0}^M N^{-k}T_{f_k}+R_N^M$$ is a bounded operator whose norm satisfies $||R_M||=O(N^{-M-1})$.

An easy use of the stationary phase Lemma shows that the (anti-)Wick symbol (also called Husimi function)  of $T_f$, namely $\frac {\la T_f\rho_z,\rho_z\ra}{\la\rho_z,\rho_z\ra}$ satisfies
\be\label{lap}
\frac {\la T_f\rho_z,\rho_z\ra}{\la\rho_z,\rho_z\ra}=f+ \frac 1  N \Delta_S f +O(N^{-2}),
\ee
where $\Delta_S=(1+\vert z\vert^2)^2\partial_z\partial_{\bar z}$ is the Laplacian on the sphere.

\newcommand{\cof}{C}
\subsection{a-Toeplitz quantization}\label{at}

Let $a\in\mathcal S(\R)$, $||a||_{L^2(\R)}=1$ and $z\in\C$. We define
\be\label{def}
\psi^a_z=\int_\R a(t)e^{i\frac{\tau(z)t}\hbar}\rho_{e^{it}z}\frac{dt}{\sqrt{2\pi}}
\ee
where $\tau(z)=\frac{|z|^2}{1+|z|^2}$ and $\hbar=\frac\pi N$.

%Although we won't need it in this paper, let us note that, when $z$ is far away form the origin and the point at infinity, $\psi^a_z$ is a lagrangian semiclassical distribution (WKB state) (by a similar construction as in \cite{pu}).

Since, by \eqref{repro}, $\rho_z=\sum\limits_{n=1}^{N-1}\sqrt{\frac{N!}{n!(N-1-n)!}}\bar z^n\vp^N_n$, we get that
\be\label{decomp}
\psi^a_z=\sum_{n=0}^{N-1}\widetilde a\left(\frac{\tau(z)-n\hbar}\hbar\right)\sqrt{\frac{N!}{n!(N-1-n)!}}\bar z^n\vp^N_n=\sum_{n=0}^{N-1}\widetilde a\left(\frac{\tau(z)-n\hbar}\hbar\right)\overline{\vp^N_n(z)}\vp^N_n.
\ee
where $\widetilde a$ is the Fourier transform of $a$
\[
\widetilde a(y):=\frac1{\sqrt{2\pi}}\int_\R e^{ixy}a(x)dx.
\]
\begin{remark}\label{suptildea}Note that, by \eqref{decomp}, $\psi^a_z$ depends only on the values of $\tilde a$ on $[0,N]$. Therefore one can always restrict the choice of $a$ to the functions whose Fourier transform is supported on $[0,N]$. In the sequel of this article we will always do so.
\end{remark}
\newcommand{\cons}{{C}}
\newcommand{\conss}{{D}}
\newcommand{\psiN}{\psi^N}
\newcommand{\unun}{M}
\newcommand{\ns}{{s}}
\newcommand{\at}{{$a$-Toeplitz}\ }
By \cite[Lemma 2]{mp} we have that
\be\label{decompa}
%\frac1N
\int_\C|\psi^a_z\rangle
%_a
\langle\psi^a_z|
%\frac i{2\pi}\frac{dzd\bar z}{(1+|z|^2)^2}
d\mu_N(z)
=\sum_{n=0}^{N-1} \cons^N_n|\vp^N_n\rangle\langle\vp^N_n|
\ee
%where the ``bra" $_a\langle\psi^a_z|$ is the adjoint of $\psi^a_z$ with respect to an $a$-dependent scalar product $\langle\cdot,\cdot\rangle_a$ to be defined later, namely $|\psi^a_z\rangle_a\langle\psi^a_z|\psi=
%\langle\psi^a_z,\psi\rangle_a\psi^a_z
%$, and
%where 
with 
\be\label{cNn}\cons^N_n=\frac{(N-1)!}{n!(N-1-n)!}\int_0^1\left|\widetilde a\left(\frac{\tau-n\hbar}\hbar\right)\right|^2\left(\frac\tau{1-\tau}\right)^n
{(1-\tau)^{N-1}}\frac{d\tau}\hbar
%(\times (\tau(1-\tau))^\frac12(N-1)^\frac32?)
.
\ee
Moreover, as $N-1=\frac1\hbar\to\infty$
 %and $0<n\hbar<1$,
\be\label{casymp}
\cons^N_n
%\to ||a||_{L^2(\R)}
=1+O(\frac1N) ,\ \ \ 0<n\hbar<1
%\frac 1{N-1}
.
\ee
\be\label{casymp0}
\cons^N_n\sim \frac1{n!}\int_0^\infty |\widetilde a(\lambda-n)|^2\lambda^ne^{-\lambda}\sqrt{2\pi\lambda}d\lambda ,\ \ \ 0\sim n\hbar
%\frac 1{N-1}
.
\ee

\be\label{casymp1}
\cons^N_n\sim C^N_{N-1-n},\ \ \  n\hbar\sim1
%\frac 1{N-1}
.
\ee
%\end{lemma}
%\begin{proof}
%Deriving \eqref{decompa} is a straightforward calculus after \eqref{decomp}.
%
%
%By the asymptotic formula for the binomial we get that, as $N,n\to\infty$,
%\[
%\frac{(N-1)!}{n!(N-1-n)!}\sim
%\left(\frac{\frac n{N-1}}{1-\frac n{N-1}}\right)^{-n}
%\left(1-\frac n{N-1}\right)^{(N-1)}.
%\]
%Moreover since $0<n\hbar<1$ we get since $\widetilde a$ is fast decreasing at infinity,
%\[
%\int_0^1\left|\widetilde a\left(\frac{\tau-n\hbar}\hbar\right)\right|^2\left(\frac\tau{1-\tau}\right)^n
%{(1-\tau)^{N-1}}\frac{d\tau}\hbar\sim
%\int_{-\infty}^{+\infty}\left|\widetilde a\left(\frac{\tau-n\hbar}\hbar\right)\right|^2\left(\frac\tau{1-\tau}\right)^n
%{(1-\tau)^{N-1}}\frac{d\tau}\hbar
%\]
%and $\frac1\hbar\left|\widetilde a\left(\frac{\tau-n\hbar}\hbar\right)\right|^2\to||a||_{L^2(\R)}\delta(\tau-n\hbar)$ as $\hbar=\frac1{N-1}\to 0$. Therefore we get \eqref{casymp}.
%
%\end{proof}
%\begin{definition}\label{defpsi}
Let us define
\be\label{defpsi}
\psiN_n:=\sqrt{\cons^N_n}\vp^N_n.
\ee
%\end{definition}

By \eqref{decompa} we have
\be\label{vlaa}
\int_\C|\psi^a_z\rangle\langle\psi^a_z|d\mu_N(z)=\sum_{n=0}^{N-1} |\psiN_n\rangle\langle\psiN_n|.
\ee
This leads to the following equality:
\be\label{vlaa}
\int_\C|\psi^a_z\rangle_a\langle\psi^a_z|d\mu_N(z)
={\bf 1}_{\mathcal H_N^a},
\ee
where $H_N^a$ is the same space of polynomials as $H_N$ but now endowed with the renormalized scalar product $\langle\cdot,\cdot\rangle_a$  fixed by
\be\label{oufoufouf}
\langle\psi^N_m,\psi^N_n\rangle_a=\delta_{m,n},
\ee
and
\be\label{diraca}
%\left(
|\psi^a_z\rangle_a\langle\psi^a_z|
%\right)
\psi:=\langle \psi_z^a,\psi\rangle_a\psi_z^a,\ \psi\in \mathcal H_N^a.
\ee
%\subsection{The Hilbert structure}\label{h2}
%We consider on $\mathcal H_N$ t
The Hilbert scalar product $\langle\cdot,\cdot\rangle_a$ on $\mathcal H^a_N$ is obtained out of  \eqref{oufoufouf} by bi-linearity. Since any polynomial $f$ satisfies 
$$
f=\sum\limits_0^{N-1}\langle\vp^N_n,f\rangle\vp^N_b=\frac1{C^N_n}
\langle\psi^N_n,f\rangle\psi^N_b, 
$$
we get
\[
\langle f,g\rangle_a:=\sum_{n=0}^{N-1}\frac1{(C^N_n)^2}\langle f,\psiN_n\rangle\langle\psiN_n, g\rangle=
=\sum_{n=0}^{N-1}\frac1{C^N_n}\langle f,\vp^N_n\rangle\langle\\vp^N_n, g\rangle.
\]
%and let us denote by $\mathcal H^a_N$ the corresponding Hilbert space.

Note that $\langle,\rangle_a$ is not given by an integral kernel. But if we ``change" of representation and define $F(z):=\langle\psi^a_z,f\rangle,\ G(z):=\langle\psi^a_z,g\rangle$ then, by \eqref{vlaa} we have
\[
\langle f,g\rangle_a=\langle F,G\rangle=\int_\C\bar{F(z)}G(z)d\mu_N(z).
\]
Let us remark 
%finally 
that
\[
F(z)=\int_\R a(t)e^{i\frac{\tau(z)t}\hbar}f(e^{it}z)\frac{dt}{\sqrt{2\pi}},\ \ \ 
%\]
%and
%\[
%f=\int_\C F(z)|\psi^a_z\rangle dz,\]
% namely 
% \[
  f(z')=\int_\C F(z)\psi^a_{z}(z') dz.
\]
Finally, let us define 
\be\label{cnumber}
\cof(z):=\sum_{k=-(N-1)}^{N-1}\sqrt{\frac{C^N_{\frac{\tau(z)}\hbar-(-i\partial_x)-k}}{C^N_{\frac{\tau(z)}\hbar-(-i\partial_x)}}}e^{ik\theta(z)},
\ee
where $C^N_n$ is the binomial coefficient,  
and, for a family of operators $z\to\sigma(z)$,  the convolution
\be\label{convol}
\sigma\star\cof(z)
:=\int_{S^1}
\sigma(ze^{-i\theta})
\cof
(
|z|e^{i\theta}
)
d\theta.
\ee

%\subsection{a-Toeplitz quantization}\label{at}
\begin{definition}\label{atop}[$a$-Toeplitz operator]\ 

Let  $a\in\mathcal S(\R)$.
To a (trigonometric) family $z\mapsto \Sigma(z)$ of (bounded) operators on $L^2(\R)$ we associate the operator $\mbox{Op}^T_a(\Sigma)$ on $\mathcal H_N^a$ defined by
\[
\mbox{Op}^T_a(\Sigma)=\int_{S^2}|\psi_z^{\Sigma\star\cof(z)a}\rangle_a\langle\psi^a_z|d\mu_N(z)
\]
where $\psi_z^a:=\int_\R a(t)e^{i\frac{\tau(z)t}\hbar}\rho_{e^{it}z}\frac{dt}{\sqrt{2\pi}}
$
with $\tau(z)=\frac{|z|^2}{1+|z|^2}$ and $d\mu_N(z)=\frac{i}{2\pi}\frac{dzd\bar z}{(1+|z|^2)^{N+1}}$.
\end{definition}
\subsection{Symbolic calculus}\label{symbcal}

%We define the symbol of $\mathcal N_{k;\gamma_k}$ at the point $z$ as the operator
%\be\label{whyes}
%\widetilde\sigma_{{k;\gamma_k}}(z):=
%\frac{e^{ikx}}{\bar z^k}
%(\gamma_k(\hbar\cdot)\mu_k)\left(\frac{\tau(z)}\hbar-i\partial_x\right)
%=\frac{e^{ikx}}{\bar z^k}\gamma_k(\tau(z)-i\hbar\partial_x)\mu_k\left(\frac{\tau(z)}\hbar-i\partial_x\right)
%\ee
%acting on $L^2(\R)$.
\begin{definition}\label{defsym}
Let $\gamma(\tau,\theta)=\sum\limits_{k=-K}^{K}\gamma_k(\tau)e^{ik\theta}$ be a trigonometric function on the sphere with each $\gamma_k\in C^\infty(]0,1[)\cap L^\infty([0,1])$.

Let 
\be\label{defngamma}
N_\gamma=\sum\limits_{-(N-1)}^{N-1}N_{k;\gamma_k}\mbox{ where 
}(N_{k;\gamma_k})_{ij}=\delta_{j,i+k}
\gamma_k((k-\frac{(-1)^k-1}2)\hbar)
\ee
 and $\mathcal N_\gamma$ the operator whose matrix on the basis $\{\psi^N_n\}$ is $N_\gamma$.

%\begin{definition}\label{defsym}
We call symbol of $\mathcal N_\gamma$ at the point $z\in S^2$ the operator
\be\label{whynot}
\sigma[\mathcal N_{
%\mathcal N_
\gamma}](z)=\sum\limits_{k=-(N-1)}^{N-1}\widetilde\sigma_{{k;\gamma_k}}(z)
\ee
where $\widetilde\sigma_{{k;\gamma_k}}$ is given by 
%\eqref{whyes}.
\be\label{whyes}
\widetilde\sigma_{{k;\gamma_k}}(z):=
\frac{e^{ikx}}{\bar z^k}
(\gamma_k(\hbar\cdot)\mu_k)\left(\frac{\tau(z)}\hbar-i\partial_x\right)
=\frac{e^{ikx}}{\bar z^k}\gamma_k(\tau(z)-i\hbar\partial_x)\mu_k\left(\frac{\tau(z)}\hbar-i\partial_x\right)
\ee
acting on $L^2(\R)$.
\end{definition}
The $a$-Toeplitz quantization procedure provides an exact (noncommutative) symbolic calculus as follows.

\begin{theorem}\label{onemain}
Let $\gamma,\gamma'$ and $\mathcal N_\gamma, \mathcal N_{\gamma'}$ as in Definition \ref{defsym}. Then
\begin{eqnarray}
\mathcal N_\gamma&=&\mbox{Op}^T_a(\sigma[{\mathcal N_\gamma}])
%\]
%\[
%\mbox{Op}^T_a(\sigma_{\mathcal N_{\gamma_1}})
%\mbox{Op}^T_a(\sigma_{\mathcal N_{\gamma_2}})
%=
\nonumber\\
\mathcal N_{\gamma'}&=&\mbox{Op}^T_a(\sigma[{\mathcal N_{\gamma'}}])
%\]
%\[
%\mbox{Op}^T_a(\sigma_{\mathcal N_{\gamma_1}})
%\mbox{Op}^T_a(\sigma_{\mathcal N_{\gamma_2}})
%=
\nonumber\\
\mathcal N_{\gamma}\mathcal N_{\gamma'}
&=&
\mbox{Op}^T_a(\sigma[{\mathcal N_{\gamma}}]\sigma[{\mathcal N_{\gamma'}}])
%+O(\hbar)
\nonumber
%\\
%\mathcal N_{\gamma_1}\dots\mathcal N_{\gamma_M}
%&=&
%\mbox{Op}^T_a(\sigma_{\mathcal N_{\gamma_1}}\dots\sigma_{\mathcal N_{\gamma_M}})
%%+O(\hbar)
%\nonumber
\end{eqnarray}
\end{theorem}
\subsection{Main result}\label{mrtqft}
It is easy to see that, for the remaining non generic values of $t$, $T_\gamma$ is not a standard Toeplitz operator. Nevertheless, it happens that it is an $a$-Toeplitz one.

\begin{theorem}\label{main}[\cite{tp2}]

Let again $\Sigma$ be either the once punctured torus or the $4$-times punctured sphere. 
For \textbf{all values of $t$}, the sequence of matrices $(T_r^{\gamma})$ are the matrices in the basis $\{\psi^N_n\}|_{n=0,\dots,N-1}$ of a family of $a$-Toeplitz operators on $\mathcal H_N^a$ with 
%(principal) 
symbol $\sigma^T_{T_r^\gamma}:\ L^2(\R,dx)\to L^2(\R,dx)$ satisfying, away of the two poles,
%^\gamma=\sigma_0^\gamma+\frac 1 N\sigma_1^\gamma+O(N^{-2})$ with

\be\label{princ}
\sigma^T_{T_r^\gamma}(z)=f_\gamma(e^{-ix}z)
%\mbox{Op}_z^{APS}(-\ tr\rho(\gamma))
%\sigma_{Qf_\gamma}
+O(\sqrt\hbar)
\ee
where $f_\gamma$ is the trace function defined by \eqref{fgamma}.

 \end{theorem}
Note that, far from the singularities (poles), the algebra of principal symbols is commutative and one can associate canonically to the symbol $\sigma^T_{T_r^\gamma}(z)=f_\gamma(e^{-ix}z)$ the trace function $f_\gamma(z)$ by evaluation at $x=0$ of the potential $f_\gamma(e^{-ix}z)$.
 \subsection{Classical limit and underlying ``phase-space"}\label{classical}

We can rewrite the general structure of the symbol of an $a$-Toeplitz operator $T$
has the form (near the south pole where $\tau\sim |z|\sim 0$)
\[
\sigma(z)=S(1-i\frac\hbar{\tau(z)}\partial_x,x+\theta,\tau(z)-i\hbar\partial_x,\hbar)
\]
where the function $S$ is $2\pi$ periodic in the second variable and the quantization present in the two first variables is the one of differential calculus.

The function $S$ satisfies
\[
S(1+\xi,x+\theta,\tau(z)+\xi,\hbar)\to
S(1,x+\theta,\tau(z),\hbar)=\gamma(\tau(z),\theta+x)\ \mbox{ as }\xi\to 0,
\]
where $\gamma(\tau,\theta,\hbar)$ is the so-called naive symbol of $T$.

As $\hbar\to 0$, $z\neq 0$, 
\[
\sigma(z)\to
\gamma(\tau,\theta+x)
%S(1-i\frac\hbar{\tau(z)}\partial_x,x+\theta,\tau(z)-i\hbar\partial_x,0)
\]
but the limit $\hbar,z\to 0$ is multivalued. Indeed as 
\[
\left\{\begin{array}{rcl}
\hbar&\to& 0\\
z&\to& 0\\
\frac\hbar{\tau(z)}&=&\hbar_0
\end{array}
\right.
\]
we have
\[
\sigma(z)\to
S(1-i\hbar_0\partial_x,e^{i(x+\theta)},0,0).
\]
And the ``classical" noncommutative multiplication for the function $S$ is given by:
\begin{eqnarray}
S\# S'(1-\hbar_0\xi,\theta+x,\tau,0)
&=&
S(1-\hbar_0\xi,\theta+x+i\partial_{\xi'},\tau,0)
S'(1-\hbar_0\xi',\theta+x,\tau,0)|_{\xi'=\xi}\nonumber\\
&:=&
S(1-\hbar_0\xi,\theta+x+i\overset{\rightarrow}{\partial_\xi},\tau,0)
S'(1-\hbar_0\xi,\theta+x,\tau,0)\nonumber
%\overset{\rightarrow}{\partial}
\end{eqnarray}
This  defines  the classical phase-space, as a noncommutative algebra of functions i.e. a noncommutative blow up of the singularity.
\vskip 1cm
In other words, an algebra of functions with a  multiplication law commutative except at certain points of singularity where it becomes noncommutative.

%\red{AFAIRE FAIRE UNE THOEORIE PLUS GLOBA;E, COMMUTATIVE LOIN DES POLES ET NONCOMMUTATIVE PRES DES POLES
%}
\vskip 1cm
%%%%%%%%%%%%%%%%%%%%%% pouralex8 %%%%%%%%%%%%%%%%%%%%%%%%%%%%%%%
%- the cases not covered by MP
%
%- a-Toplitz operators
%
%- the formula
%
%- symbolic calculus

%- conclusion on geometry
 \section{Long time semiclassical evolution}\label{longsemevo}
 \newcommand{\st}{\Xi}
 In this section we  consider the long time  semiclassical evolution through the linear Schr\"odinger equation, or more precisely to the associated 
von Neumann equation
\be\label{von2}
i\hbar\frac d {dt}O^t=[O^t,H],
\ee
where $H$ is a Schrödinger operator $H=-\hbar^2\Delta+V$ with smooth confining  $V$  ($V(x)\to+\infty$ as
$\vert x\vert\to\infty$)- or a more general   semiclassical pseudodifferential operator of principal symbol $h$,
elliptic and selfadjoint on the Hilbert space $L^2(\mathcal M)$, where $\mathcal M$ is a manifold of dimension $n+1$.

It is well known \cite{BGP,BR} that, for times smaller than $C\log{\frac 1 \hbar}$, $C$ small enough, 
$O^t$ is still a Weyl (semiclassical) pseudodifferential operator and that its principal symbol is the push-forward 
of the initial one by the Hamiltonian flow associated to the principal symbol $h$ of $H$.
It is easy to get convinced  
%\cite{P1} 
that this is already not true for large values of $C$ (greater than $\frac 2 3$ 
times the natural Liapunov exponent of the flow).

Through this paper we will suppose that the Hamiltonian flow generated by $h$ is Anosov, 
and moreover that there exists a smooth action 
of $\mathbb R^{2(n+1)}$ on $T^*\mathcal M$, 
$(\mu,\nu;s,p)\in\mathbb R^{2(n+1)}\to
\st^\nu\circ\Lambda^\mu\circ T^{s,p}$, satisfying

\be\label{horf}
\st^\nu\circ\Lambda^\mu\circ T^{s,p}\circ\Phi^t=\Phi^t\circ\st^{e^{-\lambda t}\nu}\circ\Lambda^{e^{\lambda t}\mu}\circ T^{s+tp,p},\ \lambda>0.
\ee
(\ref{horf}) is obviously mimicked on the case of the geodesic flow on surface of constant curvature ($\Lambda^\mu$ (resp. $\st^\nu$) 
is the (resp. anti-)horocyclic flow,
 $T^{s,0}$ is the geodesic one and $T^{0,p}$ corresponds to a shift of energy),
  %\cite{hor}, 
  but we will not suppose that we are in this case and we'll 
use only (\ref{horf}). Moreover we will restrict this note to the bidimensional situation $n=1$ (the extension to any $n$,
keeping (\ref{horf}), is straightforward) and take $\mathcal M=\mathbb R^{n+1}$. The proofs are local and therefore are easily
adaptable to the non flat situation using the results of \cite{pu}. Finally we could extend some of our results to the case of variable Liapunov exponents.

We will suppose that $O^{t=0}$ is a semiclassical pseudodifferential operator with smooth symbol supported in $h^{-1}(I)$ 
 for some interval $I\subset\mathbb R$ such that $h^{-1}(I)$ is compact.

 We first define, associated to $a\in\mathcal D(\mathbb R^{n+1}),\ \Vert a\Vert_{L^2}=1$, and the family of so-called (Gaussian) coherent states 
 
 \be\label{defcoh}
 \vp_{(p,q)}(x):=(\pi\hbar)^{-\frac{n+1}4}e^{-\frac{(x-q)^2}{2\hbar}}e^{i\frac{px}\hbar}, \ (p,q)=z\in\mathbb R^{2(n+1)},
 \ee 
 the family of Lagrangian states:
 \be\label{etat1}
 %\psi_z^a:=\int e^{ \int\limits_z^{\Lambda^\mu\circ T^{s,0}(z)}\frac{i\eta}\hbar}a(\mu,s)\vp_{\Lambda^\mu\circ T^{s,0}(z)}\frac{d\mu ds}{\hbar^{\frac{n+1}2}},\ \eta:=\xi.dx\
\psi_z^a:=\int \exp{\big(\frac{i}\hbar\int\limits_z^{\Lambda^\mu\circ T^{s,0}(z)}\eta\big)}a(\mu,s)\vp_{\Lambda^\mu\circ T^{s,0}(z)}\frac{d\mu ds}{\hbar^{\frac{n+1}2}},\ \eta:=\xi.dx\
%\psi_z^a:=\int \exp{\left(\frac{i}\hbar\int\limits_z^{\Lambda^\mu\circ T^{s,0}(z)}\eta\right)}a(\mu,s)\vp_{\Lambda^\mu\circ T^{s,0}(z)}\frac{d\mu ds}{\hbar^{\frac{n+1}2}},\ \eta:=\xi.dx\
(\mbox{symplectic potential on}\ T^*\mathbb R^n).
\ee
It is easy to see that, microlocally in the interior of $I$ and for a support of $a$ small enough, the operator defined by 
$\int_{h^{-1}(I)}\vert\psi_z^a\rangle\langle\psi_z^a\vert \frac {d^{n+1}z}{\hbar^{n+1}}$ 
%(here we denote by $\vert\psi\rangle\langle\psi\vert$ the orthogonal projector on the vector $\psi$) 
(here we denote by $\vert\psi\rangle\langle\psi\vert$ the operator:$\vp\to\langle\vp,\psi\rangle\psi$)
is equal to the identity modulo $\hbar^\infty$.

% \be\label{iden}
% \Vert \int_{h^{-1}(I)}\vert\rangle\psi_z^a\langle\psi_z^a\vert \frac{d^{n+1}z}{\hbar^{n+1}}-\mathbb I_{h^{-1}(I')}\Vert_{\mathcal B(L^2)}=O(\hbar^\infty)
% \ee
% for $I'$ strictly included in $I$. 
The key idea of this paper will be to write any pseudodifferential operator in the form
\be\label{pseud}
O=\int_{h^{-1}(I)}\vert\psi_z^{O_za}\rangle\langle\psi_z^a\vert \frac{d^{n+1}z}{\hbar^{n+1}}
\ee
for a suitable family of bounded pseudodifferential operators $O_z$.

The interest of such a formulation will be the fact that it is preserved by the evolution through (\ref{von2}). 
More precisely we prove in theorem \ref{1} that, for $n=1$ and
 any $0\leq t\leq C\hbar^{-2+\epsilon},\ \epsilon>0,$ there exists a bounded operator $O_z^t$ such that the solution $O(t)$ 
 of (\ref{von2}), $O(t=0)$ being
 microlocalized on $h^{-1}(I)$, satisfies

\be\label{oper}
\Vert \int_{h^{-1}(I)}\vert\psi_z^{O_z^ta}\rangle\langle\psi_z^a\vert \frac{d^{n+1}z}{\hbar^{n+1}}
-O(t)\Vert_{\mathcal B(L^2)}=O(\hbar^\infty)
\ee
(valid also for $O=$Identity$|_{h^{-1}(I)}$). This suggests to consider $O_z^t$ as the symbol of $O(t)$ at the point $z$.

In fact we will identify the symbol of $O(t)$ as a noncommutative object related to the space of leaves of the unstable 
foliation of the dynamics generated by the principal symbol $h$ of $H$.
Let us give the 
% a natural 
motivation behind this identification.
\subsection{Propagation}\label{mainresultnonc}
\begin{theorem}\label{1} (\cite{tp})\ 

%\redAFAIRE {RECTIFIER LA BORN SUR $\hbar^{-2+\epsilon}$ SUR T 
%}

Let us take $n=1$. There exists  bounded smooth and explicitly computable functions on $\mathbb R^{4}$, $\mathcal O^t_z\sim \mathcal O+\sum_{j=1}^\infty \mathcal O^t_j\hbar ^j$, 
such that, uniformly for 
$0\leq t\leq\hbar^{-2+\epsilon}$,
\[
\Vert e^{-i\frac{tH}\hbar}Oe^{+i\frac{tH}\hbar}-\int_{h^{-1}(I)}\vert\psi_z^{\widetilde {O_z}^ta}\rangle\langle\psi_z^a\vert
\frac{dz}{\hbar^2}\Vert_{\mathcal B(\mathcal H)}=O(\hbar^\infty),
\mbox{ $\widetilde {O_z}^t$ has total semiclassical Weyl symbol}
\]
\be\label{infty}
 \widetilde {\mathcal O^t_z}(\xi;x):=\mathcal O^t(\st^{e^{\lambda t}\nu}\circ\Lambda^{e^{-\lambda t}\mu}
\circ T^{s+tp,p}\circ\Phi^t(z)),\ x=(\mu,s),\ \xi=(\nu,p).
 \ee
\end{theorem}
\begin{remark}\label{r}
As a Corollary of the proof of Theorem \ref{1} it is easy to prove that a similar result is still valid when we replace 
the functions $a$ by a 
$\hbar$-dependent ones of the form $a_\hbar(\cdot)=\hbar^{-n\epsilon'/2}a(\hbar^{-\epsilon'}\cdot)$ for $\epsilon'\geq 0$ small enough. 
%(see (\ref{eq2})).
%More precisely, taking $\epsilon'<\epsilon$ (\ref{infty}) is still valid, and (\ref{main}) is valid by replacing $\epsilon$ by $\epsilon-\epsilon'$.
\end{remark}
\subsection{Noncommutative geometry interpretation}\label{noncom}
We first prove the following Lemma.
\vskip 1cm
\begin{lemma}\label{deflemm}
Let us define
\be\label{symbole}
\sigma_{\mathcal O^t}(z,z'):=F(z_1,x,s;x'.s')
\ee where $z=\Lambda^xT^{s,0}(z_1),\ z'=\Lambda^{x'}T^{s',0}(z_1)$ and $F(z_1,.,.)$ is the integral symbol of an operator of Weyl symbol given by (\ref{infty}). Then
$\sigma_{\mathcal O^t}(z,z')$ doesn't depend on $z_1$.
\end{lemma}
The Lemma is easily proven by the translation invariance properties of the Weyl quantization procedure.

We want to identify $\sigma_{\mathcal O^t}$ as an element of the
crossed product $\mathcal A$ of the algebra $\mathcal C_I$ of continuous functions on $h^{-1}(I)$ by the group $\mathbb R^{n+1}$ under the action (\ref{horf}).
A function $\sigma(z,z'),\  z'\in\Lambda_z$ cam be seen as a continuous function from $G$ to $\mathcal A$ by 
\be
f(\mu,t)(z)=\sigma(z,z'),\ \ \ z'=\Lambda^\mu\circ T^t(z).
\ee
Moreover we get an action of $G$ on $\mathcal C_I$ by, $\forall g\in G$,
\be
\alpha_{(\mu,t)} h(z)=h(\Lambda^\mu\circ T^t(z)).
\ee
The algebra structure on $\mathcal C_I\rtimes_\alpha G$ is given by the $\star$-product $(f_1\star f_2)(g)=\int f_1(g_1)\alpha_{g_1}(f_2(g_1^{-1}g))dg_1$.

An easy computation, using Theorem \ref{1} and the symbolic property of Weyl quantization shows easily that, at leading order and for all $0\leq t_1,t_2\leq\hbar^{-2+\epsilon}$,
\be\label{symb1}
\sigma_{\mathcal O^{t_1} O^{t_2}}\sim \sigma_{\mathcal O^{t_1}}\star\sigma_{\mathcal O^{t_2}}.
\ee
%\begin{definition}\label{def}
Moreover the norm $\vert\vert\vert.\vert\vert\vert$ on $\mathcal A$ is equal to the supremum over $z$ of the operator norm on $L^2(\mathbb R^{n+1})$ of the operator 
of integral kernel $\sigma(z,z')$ (more precisely $\mathcal A\rtimes_\alpha G$ is the completion of the algebra of compactly supported kernels
$\sigma(z,z')$ with respect to the norm $\vert\vert\vert.\vert\vert\vert$.

We can also give a corresponding interpretation of the vectors $\psi_z^a$. 
Let us define $\alpha\in\mathcal A$ by, for 
$z'=\Lambda^\mu\circ T^{s,0}(z)$ , $\alpha(z,z'):=a(\mu,s)$. 
Then   $\psi_z^a=\psi^\alpha:=\int_{\Lambda_z} e^{\frac i\hbar\int_z^{z'}\eta}\alpha(z',z)\vp_{z'}dz'$.

\newcommand{\eto}{T}
We associate to any element $\gamma$ of  $\mathcal A$ an operator $\eto(\gamma)$ on $L^2(\mathbb R^{n+1})$ defined by
\be\label{def1}
\eto(\gamma):=\int_{h^{-1}(I)} \vert\psi^{\gamma\star\alpha}\rangle\langle\psi^\alpha\vert\frac{dz}{\hbar^{n+1}}.
\ee
%\end{definition}
In particular a bounded pseudodifferential operator is such an operator (with $\gamma\sim\sum_{j=0}^\infty\gamma_j\hbar^j$). 
Moreover, by definition of the norm $\vert\vert\vert.\vert\vert\vert$, $\eto(\gamma)$ is a bounded operator for all 
$\gamma\in \mathcal A\rtimes_\alpha G$ and it is easy to see, using arguments of the proof of theorem \ref{1}, 
that $\eto(\gamma)$ is bounded \textit{uniformly with} $\hbar\in [0,1]$ for $\gamma$ compactly supported. 
Noting that (\ref{def1}) is a way of writing (\ref{pseud}) we get:
\vskip 1cm
\begin{theorem}\label{2}
For $0\leq t\leq\hbar^{-2+\epsilon}$ there exist $\Gamma^t$ of ``symbol" $\gamma^t\sim\sum_{j=0}^\infty\gamma^t_j\hbar^j\in\mathcal A$ such that
\be
\eto(\gamma)^t:=e^{-i\frac{tH}\hbar}\eto(\gamma) e^{+i\frac{tH}\hbar}=\eto(\gamma^t)+O(\hbar^\infty)\ \ \ \ \ \mbox{ with }
\ \ \ \ \gamma^t_0=(\Phi^t)^{\otimes 2}\#\gamma_0+O(\hbar^\epsilon).
\ee
Moreover the leading order symbol of $\eto(\gamma)^{t_1}\eto(\gamma)^{t_2}$ is $\gamma^{t_1}\star\gamma^{t_2}$.
\end{theorem}
%\underline{Sketch of the proof:} 
%\begin{itemize}
%\item the fact that Theorem \ref{1} is valid also for operators defined by (\ref{def1}) is contained in the proof of Theorem \ref{1} itself.
%\item the fact that the symbol of $\eto(\gamma^t)$ is in the completion by the norm $\vert\vert\vert.\vert\vert\vert$ is obtained by the Calderon-Vaillancourt theorem, since
%$\mathcal O(\st^{e^{\lambda t}\xi}\circ\Lambda^{e^{-\lambda t}x}
%\circ T^{s+tp+t,p}(z))$ is bounded and smooth, therefore defines a bounded (non semiclassical) pseudodifferential operator.
%\item the product formula of principal symbols is nothing but (\ref{symb1}).
%
%
%
%\end{itemize}
Let us remark also that an extension on the lines of Remark \ref{r} is also valid in this framework.
\vskip 1cm

\subsection{Another groupoid approach}\label{anogroup}
The following easy computation will be the basis of the starting point of Section \ref{homo}.
%\red{AFAIRE 
%\begin{proposition}\label{propanother}
%nbv
%\end{proposition}
%}
%\red{AFAIRE  
%DIRE PEUT-ETRE ICI QUE L'ON REPRESENTE LÍDENTITE EN MQ  SOUS DIVERSES REPRESENTATIONS ASSOCIEES A  DIVERSES GEOMETRIES SOUS-JACENTES , EXISTENTIELLES CAR DEPENDANTES DE LA DYNAMIQUE  POURTANT INVISIBLE SOUS LÍDENTITE
%}

Let us go back to the expression \eqref{def1}
\begin{eqnarray}
\eto(\gamma)&:=&\int_{h^{-1}(I)} \vert\psi^{\gamma\star\alpha}\rangle\langle\psi^\alpha\vert\frac{dz}{\hbar^{n+1}}\nonumber\\
&=&%\psi_z^a=
%\psi^\alpha:=
\int_{h^{-1}(I)}dz\int_{\Lambda_z}dz'\int_{\Lambda_z}dz" e^{\frac i\hbar\int_{z" }^{z'}\eta}\bar\alpha(z",z)\gamma\star\alpha(z',z)|\vp_{z'}\rangle\langle\vp_{z"}|\nonumber\\
&=&
\int_{h^{-1}(I)}dz"\int_{\Lambda_{z"}}dz'\int_{\Lambda_{z"}}dz e^{\frac i\hbar\int_{z" }^{z'}\eta}\bar\alpha(z",z)\gamma\star\alpha(z',z)|\vp_{z'}\rangle\langle\vp_{z"}|\nonumber\\
&=&
\int_{h^{-1}(I)}dz\int_{\Lambda_{z}}dz'
\int_{\Lambda_{z}}dz" 
e^{\frac i\hbar\int_{z}^{z'}\eta}
\bar\alpha(z,z")\gamma\star\alpha(z',z")
|\vp_{z'}\rangle\langle\vp_{z}|\nonumber\\
&=&
\int_{h^{-1}(I)}dz\int_{\Lambda_{z}}dz'e^{\frac i\hbar\int_{z}^{z'}\eta}
\beta(z',z)
|\vp_{z'}\rangle\langle\vp_{z}|\label{betsymb}
\end{eqnarray}
with
$$
\beta(z',z):=\int_{\Lambda_{z}}dz"\bar\alpha(z,z")\gamma\star\alpha(z',z")
$$
Note that, when $a(z,z")$ is the integral kernel of a unitary operator on $L^2(\Lambda_z)$, then $\beta=\gamma$. This will be the starting point of our discussion in Section \ref{homo}.
In particular, when $\gamma(z,z"0=\delta(z,z")$ then we recover.
the usual decomposition of the identity with diagonal diade.

So we see that the transport property given by  Theorem \ref{2} that, denoting
$A(\beta):=\int_{h^{-1}(I)}dz\int_{\Lambda_{z}}dz'e^{\frac i\hbar\int_{z}^{z'}\eta}
\beta(z',z)
|\vp_{z'}\rangle\langle\vp_{z}|$ we have
\be\label{propbeta}
e^{-i\frac{tH}\hbar}A(\beta)e^{i\frac{tH}\hbar}\sim A((\Phi^t)^{\otimes 2}\#\beta).
\ee
\subsection{Extended semiclassical measures}\label{semi}

In the same way that one associates to a vector $\psi$ (or density matrix) 
the quantity $\langle\psi,\vp_z\rangle\vert^2$ considered as a measure  by the formula
$\langle\psi,O\psi\rangle=\int \mathcal O_T(z,\bar z) \vert\langle\psi,\vp_z\rangle\vert^2dz$, 
where $\mathcal O_T$ is the Toeplitz symbol of $O$, on can associate to $\psi$ (or a density matrix) a sort 
of ``off-diagonal" version by the quantity $R_\psi(z,\bar{z'}):=\langle \vp_z,\psi\rangle\langle\psi,\vp_{z'}\rangle$ for $z'\in\Lambda_z$. 

$R_\psi$ can be considered as an element of the dual of a (dense) subalgebra of $\mathcal A$ and will have better properties of semiclassical 
propagation. For sake of shortness we express the result in the case of eigenvectors of the Hamiltonian $H$, leaving the 
straightforward derivation
 for $R_{e^i\frac{tH}\hbar}\psi$ in the same topology.
\vskip 1cm
\begin{theorem}\label{3}
Let us define for $I$ compact interval of $\mathbb R$, $\mathcal D_I\sim\mathcal D(h^{-1}(I)\times \mathbb R^n)$ the subalgebra of 
smooth compactly supported elements of $\mathcal A$. 
Let $\psi$ be an eigenfunction of $H$. 

Then, restricted to $z\in h^{-1}(I)$, $R_\psi(z,\bar{z'})$ considered as a function 
on $h^{-1}(I)\times \mathbb R^n$, belongs to $\mathcal D_I'$. Moreover, in the weak-* topology and, $\forall \epsilon >0$, 
uniformly for $0\leq t\leq \hbar ^{-2+\epsilon}$
\be\label{pf}
\Phi^t\# R(z,\bar{z'})
%\langle \vp_z,\psi\rangle\langle\psi,\vp_{z'}\rangle=\langle \vp_z,\psi\rangle\langle\psi,\vp_{z'}\rangle
=R(z,\bar{z'})+O(\hbar^\epsilon).
\ee

\end{theorem}
\subsection{Noncommutative phase space associated to %AFAIRE
time arrow}\label{phasparrow}
%\red{AFAIRE
%PRESENTER LES FEULLETAGES STABELS ET INSTABELS COMME AOSSICES AUX TEMPS VERS PLUS OU MOINS L' INFINI ET INTERODCUIRE HOMOCLINES
%}

%AFAIRE
Let us conclude this section by summarizing what has been done. What we have constructed is (an algebra of functions on) a  phase space able to handle the classical limit of the quantum propagation at diverging times as the $C^*$-algebra of the invariant foliations associated to the classical dynamics, i.e. the noncommutative space of  leaves of such invariant foliations.

But this construction is different depending on whether we consider positive or negative diverging times: the stable manifold foliation is needed in the first case though the unstable one is suitable for the second.

The idea comes then naturally, if one wishes to unify the two arrows of time in order to get objects truly invariant by the evolution for
% on 
time from $-\infty$ to $+\infty$, to consider an equivalent construction associated the the intersection of the two stable and unstable foliations. The set of intersections of   two given invariant manifolds is the set of homoclinic orbits of a trajectory. Such 
%a
 sets will become the " leaves" of a new foliation of the classical phase space, this time invariant by time evolution from $-\infty$ to $+\infty$.

This is the construction we will present in the next session, amazingly inspired by perturbation methods for the resonant harmonic oscillator.
\vskip 1cm

%%%%%%%%%%%%%%%%%%%%%% pouralex8 %%%%%%%%%%%%%%%%%%%%%%%%%%%%%%%

%- the formula
%
%-noncummutative interpretaion
%
%- sybolic calculus and transport formula
%
%- (not) 'semiclassical meausres"
%
%- towrads homo
%\section{Indeterminism versus unpredictability}\label{imvrsin}\ 
%
%\red{indetermisnime quantique choquant par arapport a Laplace mais pas a Poincar'e
%}
%\subsection{Critics of the determinism reason of classical mechanics}\label{crit}
%\subsection{Measurement in quantum mechanics}\label{meas}
%\subsection{Some space for merging the two}\label{merg}
 \section{The   noncommutative phase-space underlying the quotient by the quantum flow}\label{homo}
 Partially less rigorous mathematically than the preceding  ones, the present one proposes a construction of the classical limit of quantum systems whose  underlying  classical dynamics is chaotic. The key point of this work in progress consists in a deep change of point of view with respect to all other, to our knowledge, attempt of handling spectral asymptotics: instead of looking at eigenvalues $(e_j)_{J=1,\dots}$ of the quantum Hamiltonian $H$, the quantities we consider are  rather consider the set of frequencies $\left(\tfrac{e_i-e_j}\hbar\right)_{ij=1,\dots}$, namely the spectrum of the Heisenberg-von Neumann derivation $\tfrac1{i\hbar}[H,\cdot]$.
 \subsection
 %{(Again) the harmonic case}
{The non-resonant harmonic oscillator spectrum and the non commutative torus} 
 \label{harm}\  
 
 %%%%%%%%%%%%%%%%%%%%%% pouralex8 %%%%%%%%%%%%%%%%%%%%%%%%%%%%%%%
% - construction of (quasi)modes using coherent sates
% 
% - the construction of eigenfrequencies operators
% 
% - "symbolic"interpretation
 
 \vskip 1cm
 
 From the formula for the generating function of the Hermite polynomials $\{h^p_n\}_{n=0,\dots}$, namely
%, for all $z\in\bC$,
$$
\sum_{n=0}^\infty
h^p_n(x)\frac{z^n}{n!}=e^{2xt-t^2},\ \|h^p_n\|_{L^2(\bR,e^{-\frac{x^2}2}dx)}=\pi^\frac14 \sqrt{2^nn!}
$$
valid for all $z\in\bC$, we get easily that the normalized eigenfunctions of the harmonic oscillator $h_j,j=0,\dots$ satisfy
\be\label{defgj}
\sum_{j=0}^\infty
h_j(x)
\frac{(\frac z{\sqrt{
%2
\hbar}})^{j}e^{-\frac{|z|^2}{2\hbar}}}
{\sqrt {
%\sqrt\pi
 j!}}
 =
(\pi\hbar)^{-\frac14}
e^{
-\frac{z^2-2\sqrt 2 zx+x^2}{2\hbar}}e^{-\frac{|z|^2}{2\hbar}}:=g^z(x).
\ee
Note that $\|g^z\|_{L^2(\bR,dx)}=1$.

we define
$$
\ec{z}:=
%(\pi\hbar)^{-\frac14}
\tfrac{e^{-i\frac{qp}{2\hbar}}}{(\pi\hbar)^{-\frac14}}e^{-\frac{(x-q)^2}2+i\frac{px}\hbar}, \mbox{ with }z=q+ip.
$$

We first remark that
$$
g^{\sqrt2z}=\ec{z}
%\mbox{ with }z=\tfrac1{\sqrt2}(q+ip),
$$
so that \eqref{defgj} reads now
$$
\ec{z}=
e^{-\frac{|z|^2}{\hbar}}
\sum_{j=0}^\infty
%h_j
\frac{\left(\sqrt{\tfrac 2
\hbar}z\right)^{j}}
{\sqrt {
%\sqrt\pi
 j!}}.
 h_j
 $$
 
 \vskip 1cm
 We get that, for any choice of $z\in\bC, z\neq 0$,
$$
(2\pi)^{-\frac34}
\sqrt{\frac{|z|}{\sqrt\hbar}}
\int_0^{2\pi}
g^{e^{it}z}
e^{-ijt}
%\frac {|z|^\frac12}{\sqrt{2\hbar}}
dt
=
%\sqrt{2\pi} 
\sqrt{\frac{\big(\frac{|z|^2}{\hbar}\big)^{j}e^{-\frac{|z|^2}{\hbar}}\sqrt{2\pi\frac {|z|^2}{\hbar}}}{j!}}
%\sqrt{\frac z{\bar z}}
h_j
$$
so that $$
h_j
=
\sqrt{\frac{j!}{\big(\frac{|z|^2}{\hbar}\big)^{j}e^{-\frac{|z|^2}{\hbar}}\sqrt{2\pi\frac {|z|^2}{\hbar}}}}
(2\pi)^{-\frac34}
\sqrt{\frac{|z|}{\sqrt\hbar}}
\int_0^{2\pi}
g^{e^{it}z}
e^{-ijt}
%\frac {|z|^\frac12}{\sqrt{2\hbar}}
dt
$$
and in particular, taking now ${|z|^2}=j\hbar$,
$$
h_j
=
C_j
%\sqrt {j\frac{\bar z}z}
\frac{j^{\frac14}}{(2\pi)^\frac34}
\int_0^{2\pi}
g^{e^{i\theta}\sqrt{j\hbar}
}
e^{-ij\theta}
%\frac{dt}{\sqrt{2\pi}}
d\theta
$$
with
$$
C_j=
\sqrt{\frac{j!}
{
%\big(\frac{|z|^2}{2\hbar}\big)^{\frac j2}e^{-\frac{|z|^2}{2\hbar}}\sqrt{2\pi\frac {|z|}{\sqrt{2\hbar}}}
\sqrt{2\pi j}j^je^{_-j}
}}<1\mbox{ and }\sim 1\mbox{ as }
%\hbar,j^{-1}\to0
j\to\infty.
$$
and 
$$
h_{j'}
=
\sqrt\frac{j'!}{j!}
C_j
%\sqrt {j\frac{\bar z}z}
\frac{j^{\frac14}}{(2\pi)^\frac34}
\int_0^{2\pi}
g^{e^{i\theta}\sqrt{j\hbar}
}
e^{-ij'\theta}
%\frac{dt}{\sqrt{2\pi}}
d\theta.
$$
It is of course exceptional, as the harmonic oscillator is exceptional for quantum mechanics, that such an exact formula exists, giving an explicit expression for the spectral decomposition of the Hamiltonian.

It is well known that the following decomposition of the identity holds true
\be\label{decidn}
I=\int \adir{\ec{z}}{\ec{z}}\tfrac{dzd\bar z}{2\pi\hbar}
%=I
\ee
But there is also the following decomposition
\be\label{decoortho}
I=\sum_{j=0}^\infty
\adir{h_j}{h_j}
\ee
which can be rewritten as 

\be\label{decoec2}
I=\sum_{j=0}^\infty \int_{\bT^2}e^{ij(\theta'-\theta)}
\adir{\ec{\sqrt{j\hbar}e^{i\theta'}}}{\ec{\sqrt{j\hbar}e^{i\theta}}}d\theta' d\theta
\ee
or, more geometrically,
\be\label{decoec2geo}
I=\sum_{j=0}^\infty \int_{|z'|^2=|z|^2=j\hbar}e^{i\int_z^{z'}\bar zdz}
\adir{\ec{z'}}{\ec{z}}dz' dz
\ee
The advantage of this last formulation is that, denoting $e^{i\theta_z}:=\sqrt{\frac z{\bar z}}$, one has, for any $n\in\bZ$,
\be\label{creanil}
A_n:=\sum_{j=0}^\infty j\int_{|z'|^2=|z|^2=j\hbar}e^{i\int_z^{z'}\bar zdz}
e^{in\theta_{z'}}
\adir{\ec{z'}}{\ec{z}}dz' dz=\sum_{j=0}^\infty
\adir{h_{j+n}}{h_j}
\ee

Note that this is also equal to
$$
\sum_{j=0}^\infty j\int_0^{2\pi}d\theta_z\lim_{T\to\infty}\frac1{2T}\int_{-T}^{T}d\theta_{z'}e^{ij(\theta_{z'}-\theta_z)}e^{in\theta_{z'}}\adir{\ec{z'}}{\ec{z}}|_{|z'|^2=|z|^2=j\hbar}
$$
$$
\sim\int dz\lim_{T\to\infty}\frac1{2T}\int_{-T}^{T}dte^{i\frac{|z|}\hbar t}e^{in\theta_{e^{it}z}}\adir{\ec{e^{it}z}}{\ec{z}}
%|_{|z'|^2=|z|^2=j\hbar}
$$
$$
=\int dz\lim_{T\to\infty}\frac1{2T}\int_{-T}^{T}dte^{i\int_z^{\Phi^t(z)}\bar z dz}e^{in\theta_{\Phi^t(z)}}\adir{\ec{\Phi^{t}(z)}}{\ec{z}}
%|_{|z'|^2=|z|^2=j\hbar}
$$
where we see that only the symplectic form, the flow and the groupoid of the flow $\Gamma_{\Phi^t}$ (defined in a few lines below) appears.
%We will instead consider now the following oorator
%\begin{eqnarray}\label{decid}
%A&:=&
%\int e^{im\th^eta}\adir{\ec{z}}{\ec{z}}\tfrac{dzd\bar z}{2\pi\hbar}\nonumber\\
%&=&
%\int_{\bR^+}|z|d|z|\int ....\nonumber\\
%&\sim&
%\sum_{j}\adir{h_j}{h_{j+m}}\nonumber
%%\int_\bC \int_\bR e^{i(\tfrac{|z|}\hbar(t+\theta_z)+nt}\adir{\ec{e^{it}z}}{\ec{z}}dt\tfrac{dzd\bar z}{2\pi\hbar}\nonumber\\
%%&=&
%%\sum_{l.m}....\int^{\bR^+}\int_\bR dt\int_0^{2\pi}d\theta
%%e^{i((\frac{|z|}\hbar(t+\theta)+(n-l)t-l\theta}e^{im\theta}... |z|d|z|/2\pi\hbar\nonumber\\
%%&=&
%%\sum_{l.m}....\int^{\bR^+} dt\int_0^{2\pi}d\theta\delta(n-l-\tfrac{|z|}\hbar)
%\end{eqnarray}
%so that
%$$
%e^{-itH_0/\hbar}Ae^{itH_0/\hbar}=e^{itm}A.
%$$

%Let us consider a mixed qunatization 

We see two points:

- $e^{-itH/\hbar}A^+_ne^{itH/\hbar}=e^{int}A^+_N,\ \forall t\ \Leftrightarrow\ \frac i\hbar[H,A^+_n]=n$, i.e. 

\hfill $A^+_N$ is an eigenvector of the derivation $\frac i\hbar[H,\cdot]$ of eigenvalue $n$

- $\rho_n(z',z):=e^{i\int_z^{z'}\bar zdz}
e^{in\theta_{z'}}$ is a  function on the groupoid of the flow $\Gamma_{\Phi^t}:=\{(z',z)/\exists t,\ \Phi^t(z)=z'\}$ which is an

\hfill eigenvector of the action of the flow, i.s.   $\rho^t(z',z):=(\Phi^t)^{\otimes 2}\#\rho(z',z)$,  $\rho^t_n=e^{int}\rho_n$.
\vskip 1cm
In two dimension we have the following formula: consider 
$$
H=\omega_1H_0\otimes I+\omega_2I\otimes H_0,\  h_{j_1,j_2}=h_{J_1}\otimes h_{j_2},\ 
\ec{(z_1.z_2)}
=\ec{z_1}
\otimes \ec{z_2}.
$$
and
%\begin{eqnarray}
%B&:=&
%\int_\bR\int_{\bC^2}
%e^{(\omega_1n_1+\omega_2n_2)t}\adir{\ec{e^{i\omega_1t}z_1.e^{i\omega_2t}z_2}}{ec{z_1.z_2}}dtdz_1dz_2
%\end{eqnarray}
we then express the same formula as before in the form, for $N\in\bZ^2$, $Z=dz_1dz_2$, $\Theta_Z=(\theta_{z_1},\theta_{z_2})$,
\begin{eqnarray}
A_{N}&:=&
\int dZ\lim_{T\to\infty}\frac1{2T}\int_{-T}^{T}dte^{i\int_Z^{\Phi^t(Z)}\bar Z \cdot dZ}e^{iN\cdot\Theta_{\Phi^t(Z)}}\adir{\ec{\Phi^{t}(Z)}}{\ec{Z}}\label{erg}
%|_{|z'|^2=|z|^2=j\hbar}
\end{eqnarray}
Obviously

- $e^{-itH/\hbar}A^+_Ne^{itH/\hbar}=e^{iN\cdot\omega t}A^+_n,\ \forall t$, $\omega=(\omega_1,\omega_2)$, i.e.

\hfill $A^+_n$ is an eigenvector of the derivation $\frac i\hbar[H,\cdot]$ of eigenvalue $N\cdot \omega$, 
%$\omega=(\omega_1,\omega_2)$

- $\rho_N(Z',Z):=e^{i\int_Z^{Z'}\bar Z\cdot dZ}
e^{iN\cdot\Theta_{Z'}}$ is function on the groupoid of the flow $\Gamma_{\Phi^t}$ which is
%$\{(z',z)/\exists t,\ \Phi^t(z)=z'\}$

\hfill an eigenvector of the action of the flow, i.e.     $\rho^t_N=e^{iN\dot\omega t}\rho_N$.
\vskip 0.3cm
%\red{AFAIRE
Let us remark that, when $Z$ belongs to a Bohr-Sommerfeld quantized  torus $\bT^2\times\bT^2$,
%. Therefore, 
by the ergodic theorem
\be\label{ergtor}
\lim_{T\to\infty}\frac1{2T}\int_{-T}^{T}dte^{i\int_Z^{\Phi^t(Z)}\bar Z \cdot dZ}e^{iN\cdot\Theta_{\Phi^t(Z)}}\adir{\ec{\Phi^{t}(Z)}}{\ec{Z}}
=
\int_{\bT^2}
e^{i\int_Z^{Z'}\bar Z \cdot dZ}e^{iN\cdot\Theta_{Z'}}\adir{\ec{Z'}}{\ec{Z}}dZ'
\ee
 and, in particular, doesn't depend on $Z$. 

We will come back to this remark in Section \ref{creann} below.

%}
\vskip 0.3cm
%\hfill $A^+_n$ is eigenvector of the derivation $\frac i\hbar[H,\cdot]$ of eigenvalue $N\cdot \omega$, $\omega=(\omega_1,\omega_2)$

Finally, we remark that
$$
\Theta_{\Phi^{2\pi/\omega_2}Z}=(\theta_{z_1}+2\pi\tfrac{\omega_1}{\omega_2},\theta_{z_2}).
$$
so that
$$
\int_{-2\pi\tfrac{\omega_1}{\omega_2}K}^{2\pi\tfrac{\omega_1}{\omega_2}K}f(\Phi^t(Z))dt
=\sum_{k=-K}^K
\int_0^{2\pi}d\theta \left(P_{2\pi\tfrac{\omega_1}{\omega_2}}f\right)(\theta_1+\theta,\theta_2+\theta)
$$
where 
$$
P_{2\pi\tfrac{\omega_1}{\omega_2}}f(\cdot,\cdot)=f(\cdot+\tfrac{\omega_1}{\omega_2},\cdot).
$$

Moreover
$$
e^{-i2\pi\frac{\omega_1}{\omega_2}H}A_{n+1,0}e^{i2\pi\frac{\omega_1}{\omega_2}H}
=
e^{i2\pi\frac{\omega_1}{\omega_2}n_1}A_{(n_1,0}.
$$
one recovers easily the
$$
\{\lambda_i-\lambda_j\}=\{\omega_1n_1+\omega_2n_2,\ (n_1,n_2)\in\bZ^2\}.
$$
Let us remind that the numbers $\lambda_i-\lambda_j$ are the frequencies of the quantum evolution.

Looking at eigenfunctions of $V_{2\pi\tfrac{\omega_1}{\omega_2}}$ on torus, we get that they are of the form $\tfrac1{\sqrt{2\pi}}e^{in\theta},\ n\in\bZ$ with eigenvalues $e^{in{2\pi\tfrac{\omega_1}{\omega_2}}}$.

Obviously, $P_{2\pi\tfrac{\omega_1}{\omega_2}}$ is the Poincaré mapping of the linear ergodic flow $\Phi^t$ on the torus $\bT^2$ associated to the Poincaré section $\theta_2=0$.

Recalling that $A_N$ was constructed out of the eigenfunctions $e^{in_1\theta_1}$, one just proved the following result.

\begin{theorem}\label{main1}
The set of frequencies of the two dimensional harmonic oscillator are determined by the spectrum of the Poincaré mapping of the classical flow.
\end{theorem}

%\vskip 5cm

 \subsection{The space of frequencies and the noncommutative torus}\label{freq}
 \ 
 %\vskip 1cm
 
 %%%%%%%%%%%%%%%%%%%%%% pouralex8 %%%%%%%%%%%%%%%%%%%%%%%%%%%%%%%
% - creation and annihilation oprators
% 
% - frequencies as dynamical "a ctions : matrix elements with symbols on the same lagrangian versus between lagrangians.
% 
% - the formula
% 
% - homointerpretation
% 
% - what'sleft at the classical limit

% \vskip 3cm
 The operator $P_{\tfrac{\omega_1}{\omega_2}}$ which appeared in the last section is well known in noncommutative geometry: it is one of the two generator of the so-called noncommutative torus, namely the algebra generated by $U,V$ satisfying
 \be\label{noncoto}
 VU=e^{i\tfrac{\omega_1}{\omega_2}}UV.
 \ee
 Indeed, taking $U=\times e^{i\theta}$ on $L^2(S_1)$, one sees easily that taking  $V:=P_{\tfrac{\omega_1}{\omega_2}}$,  the couple $(U,V)$ satisfies \eqref{noncoto}.
 
 One can now reformulate Theorem \ref{main1} in a more synthetic way.
 \begin{theorem}\label{main2}
 The (algebra of the) quotient  by the flow of the harmonic oscillator determines the set of frequencies.
 \end{theorem}
 \subsection{Extensions to chaotic systems}\label{chao}\ 
 
% \red{AFAIRE MENTIONER QUE LE CALCUL PRECEDENT DONNE A LA FOIS LE SPECGTRE DE VP ET CELUI DES FREQUENCES, A CAUSE DE LÁPECT LINÉAIRRE DU SPECTRE
% 
% MAISAU LIEU DE CHERCHER DES QUASIMODES ASSOCIES A UNE VALEUR PROPRE, QUI SERAIT DELOCALIZES A CAUSE DE L'ERGODICITE, ON CHERCHE DES QUASIMODES DES FEREQUNES,  QUI DOIVENT ETRE DONC DES SUPEROPOSITIONS LOCALIZEES DE QUANSIMODES DES VALEURS PROPRES DELOCALISES.
% 
%  \textit{c'EST LA THEORIE QUI DIT CE QU'IL FAUT MESURER. ET C'EST BEN LES FREQUENCES QUE L'ON VOIT}
% 
% TRES BEAU.
%
%\red{AFAIRE REMARQUER ENFIN QUE C' EST DANS LA STRUCTURE INFINITESIMALE  DES VARIETES STABLES ET INSTABLES (GERMES) QUE L' ON RETROUVE A DENSITE DES HOMOCLINES. }

From the beginning of this section, eigenvalues and differences of eigenvalues appear the same way, due to the fact that the spectrum of the harmonic oscillator is linear. This is not the case any more for chaotic systems, and the construction which follows will emphasize the role of frequencies with respect to the one of eigenvalues. Let us note that this is far from being non-physical. On the contrary. First because the Hamiltonian is always defined u to a constant unessential to the dynamics. Second, frequencies are the quantities than one observes. Finally constructing quasi-modes associated to a given approximate eigenvalue is somehow ambiguous when the underlying classical flow is ergodic as this property forbids a clear localization of eigenfunctions in phase space. except the expected de-localization on the whole energy shell predicted by quantum ergodicity. Quasimodes associated to given frequencies are superposition of eigenfunctions, leading to the freedom of relocalization by superposition of unlocalized functions.  This is what we are going to see now.
 \vskip 1cm
 The advantage of the statement of Theorem \ref{main2} is that it relies the quantum frequencies of the system to the very only flow of the classical underlying system. But one has to say and one notices that this link is subtle: it relies the invariants of the quantum flow, namely the frequencies - a very natural quantum object due to the discrete \textit{quantum} structure of the quantum spectra - which describes perfectly the quotient of the Hilbert space by the quantum dynamics, to what replaces this concept of frequency - which doesn't exist in classical mechanics outside the very peculiar situation of integrable systems -  the quotient of the space by the flow, endowed, one has to say,  with its very noncommutative algebra of  structure.
 
 Classical frequencies don't exist generally in classical setting, but invariants of the classical dynamics do. It is therefore natural to try to get a concept of classical frequency through the algebra of the foliation by the classical dynamics, and try to get a theorem similar to Theorem \ref{main2} holding true   for  quantum situations whose underlying classical dynamics is chaotic.
 
 What are the invariants of classical dynamical systems? We have seen in section \ref{longsemevo}  how stable and unstable manifolds enter the game of long time semiclassical evolution. More precisely, we have seen that unstable manifolds are suitable for long positive time evolution, and stable ones to negative times. This leads to the idea of considering intersections of stable and unstable manifold for a good candidate for handling forward and backward evolution times. 
 
 These intersections are very natural objects in the dynamical systems theory: they are precisely the trajectories themselves. But  one may consider all the intersections of two given invariant manifolds, that is the set of homoclinic orbits associated to a given trajectory.
 \subsubsection{Homoclinic foliations versus invariant tori fibration}\label{homfol}\ 
 
  Homoclinic orbits of a given trajectory $\gamma$ are trajectories themselves $\{\Phi^t(z_{hom}),t\in\bR\}$ which asymptotically tend to $\gamma$ when $t\to\pm\infty$: $$
 d(\Phi^t(z_{hom}),\gamma)\to 0\mbox{ when }t\to\pm\infty.
 $$
 
  An invariant non resonant torus $\omega$ of an integrable classical system contains trajectories dense in $\Omega$. In this spirit, all these trajectories get closer and closer to each other at some moments when $t\to\pm\infty$.
  
It is therefore tempting to think at the set of homoclinic trajectories of a given one - i.e. the set of all intersections  of two stable and unstable manifolds - on one side, and invariant non resonant tori on the other side, as the same object with two incarnations, in chaotic and integrable situations.

   There are both similitude and big differences in the geometries of the two  situations, but let us claim already that the density feature of the trajectories on a non-resonant torus find an echo in the density of homoclinic curves around a given trajectory at $T\to\pm\infty$. This is particularly visible in the case of a closed periodic orbit around which homoclinic orbits accumulate. We'll get back to this later.

   On the other side, the geometries are different by essentially two facts: first,  invariant tori provide a foliation to the phase space which is actually a fibration by Lagrangian manifolds. In the homoclinic case, one cannot even talk about foliation strictu senso, since the set of homoclinic orbits presents an intrinsic discrete character, they are isolated trajectories: all of them have the homoclinic flavour but a ``leaf" is the infinite (countable) union of one dimensional manifolds (trajectories). Nevertheless, there is a big temptation %temptation  
   to consider them as leaf, and to consider a noncommutative algebra associated to this very singular object.
   
   In fact the very basic primary axiomatic structure of noncommutative geometry will pass over this extra non  connectedness property of the ``leaves" : indeed the concept of groupoid, in its very algebraic structure, doesn't care about connectedness.
    
 \vskip 1cm

   \subsubsection{The construction of the noncommutative structure}\label{noncostruc}\ 
   
 \newcommand{\homoc}{\mathcal H}
 \newcommand{\homs}{\homoc(\Lambda_s,\Lambda_u)}
 \newcommand{\homsi}[1]{\homoc(\Lambda_s^{#1},\Lambda_u^{#1})}
Let us consider the two stable and unstable foliations $\{\Lambda_s\}$ and \{$\Lambda_u\}$
of the classical flow. To any pair $(\Lambda_s^1,\Lambda_u^1)$ we associate the set $\homoc(\Lambda_s^1,\Lambda_u^1)$ of intersections of $\Lambda_s^1$ and $\Lambda_u^1$: it consists of a (countable) union of trajectories, all homoclinic to each others. once again this is not strictly speaking a foliation (more precisely these intersections are not leaves), but one can associated to it the following groupoid 
\be\label{defgro}
\Gamma_{\homs}:=\{\gamma,s(\gamma),r(\gamma)\in\homoc(\Lambda_s^1,\Lambda_u^1)\mbox{ for some pair }(\Lambda_s^1,\Lambda_u^1)\}.
\ee
The set of functions on $\Gamma_{\homs}$ can be seen as the ``crossed product" of the algebra of functions on the phase space by the groupoid $\homs$.

%\red{AFAIREWe define the homoclinic algebra associated to the flow as follows.
%
%}
We define the homoclinic algebra associated to the flow as usual, that is as the completion of the algebra of regular functions on the groupoid $\Gamma_{\homs}$. We'll come back to this below.
   
   \subsubsection{Bohr-Sommerfeld conditions I}\label{BSone}\ 
   
   According to Section \ref{longsemevo},  the forward (resp. backward) propagation on time  is done using the action of a symbol on the stable manifold (resp. unstable). Therefore one expect that the action on time both for ward and backward should involve the same way the two propagations. This implies that, between two points $z,z'\in\homsi{1}$, and any two paths $z_u\in\Lambda_s,z_u\in\Lambda_u$
   $$
   e^{i\frac{\int_z^{z'}\bar {z_s}dz_s}\hbar}=e^{-i\frac{\int_z^{z'}\bar {z_u}dz_u}\hbar}
      $$
      i.e.
      \be\label{condbs1}
      \oint\bar zdz=2\pi k\hbar,\ k\in\bZ
      \ee
      
  This formula provides a first set of Bohr-Sommerfeld conditions that must satisfy the homoclinic torus.  
  
  These conditions are invariant of the path along which one computes the integral, due to the Lagrangian property of invariant manifolds, precisely.
  
  It is a set of conditions  because if we are in dimension greater than $2$, that is with a   phase space of dimension $2d$, $d>2$, the invariant manifolds, being Lagrangian, have dimension $d$ and there are $d-1$ different independent cycles that connect two homoclinic trajectories.

  Therefore, there is one ``missing" Bohr-Sommerfeld condition, and it will be provided by the quantization of the Poincar\'e    mapping that will be defined in the next section.

 \vskip 1cm

 \subsubsection{The 
 %homoclinic algebra and 
  ``Poincaré" section}\label{homalg}\ 
 
 In this section, we present a construction of what should be the Poincar\'e section in the case of chaotic systems. The main difference with the integrable case will be that this section will be in a certain sense infinitesimal around the trajectory. In other words, there will be a kind of germ of Poincar\'e section. But the traces of the homoclinic orbits in this infinitesimal neighbourhood will be dense, as the trajectories on the integrable torus are. Therefore, a concept of quotient by the Poincar\'e map will survive in the chaotic situation in a noncommutative form very close to  the integrable  one.
% \red{AFAIRELA AUSSI PENSER A LA SECTION DE POINCARE COMME INFINITESIMALE AUTOUR DE LA TRAJECTIRE.... POUR POUVOIR CONSIDERER TOUTES LES TRAJECTOIRES HOMOCLINES.
% OU ALORS APSSER A UN QUOTIENT ?
% }
 \vskip 1cm
 We have seen in Section \ref{freq}, that the spectrum of frequencies of the non-resonant harmonic oscillator is determined  by the reduced algebra of the quotient  by the flow associated to a Poincaré section. This remark, added to the Bohr-Sommerfeld considerations of the preceding section leads to the following considerations.
 
 For the rest of this section, we  place ourselves in the two-dimensional situation $d=2$. 
 
 \noindent Let $z\in\gamma$ and let $\Xi,\Lambda$ be  the flows (supposed to exist)  defined in Section \ref{freq} such that $\Lambda^\mu(z)\in\Lambda_s(\gamma),\Xi^\nu(z)\in\Lambda_u(\gamma)$. If for some $\mu_\pm\in\bR^\pm$, $\Lambda^{\mu_\pm}(z)$ belongs to an homoclinic curve of $\gamma$ then, for some $\nu_\pm,t_\pm$,
 $$
 \Lambda^{\mu_\pm}(z)=\Xi^{\nu_\pm}(\Phi^{t_\pm}(z)).
 $$
Let us suppose for the moment that $t_\pm=0$, one can easily show that this is the case when $\gamma$ is periodic of period $T$ (in this case $t_\pm$ can be taken as $t_\pm=kT,\ k\in\Z$, i.e. $\Phi^{t_\pm}(z)=z$).

\noindent Suppose also that there exists $T$ such that
$$
\Phi^{-T}(\Lambda^{\mu_\pm}(z))=\Lambda^{\mu'_\pm}(z)=\Xi^{\nu'_\pm}(z).
$$
Again this hypothesis is satisfied when $\gamma$ is periodic of period $T$ with 
$$
\mu'_\pm=e^{-T}\mu_\pm,\ \nu'_\pm=e^{T}\nu_\pm.
$$
Moreover, a simple computation shows that this pseudo-period is also a pseudo-period for any other homoclinic   intersection. Therefore, $T$ is associated to the full homoclinic ``leave".
Note that we will have also
$$
\Phi^{T}(\Xi^{\nu_\pm}(z)=\Xi^{e^{-T}\nu_\pm}(z).
%,\ \nu"_\pm=e^{-T}\\nu_\pm.
$$

We will define the Poincaré section of $\gamma$ as the set of points of the form $\Lambda^{e^{-nT}\mu^i_\pm}(z)=\Xi^{e^{-T}\nu^i_\pm}(z),\ n\in\bZ$ for any homoclinic  point $z^i=\Lambda^{\mu^i_\pm}(z)=\Xi^{\nu^i_\pm}(z),\ n\in\bZ$  of $z$.

By analogy with Section \ref{harm} we will denote $e^{-T}:=\omega T$ and call $P_{\omega T}$ the operator of translation by $\omega T$ on coordinates $\mu_\pm,\nu_\pm$.
\begin{definition}\label{defsect}\ 

Let $z\in\homsi{1}$, the Poincaré section of $\homsi{1}$ at $z$ is
$$
\cP_z(\homsi{1}):=\left\{\Lambda^{e^{-nT}\mu^h_\pm}(z^h) =\Xi^{e^{-T}\nu^h_\pm}(z^h),\ n\in\bZ,
%\mbox{ for all }
\ z^h= \Lambda^{\mu^h_\pm}(z) =\Xi^{\nu^h_\pm}(z)\in\homsi{1} \right\}.
$$
and the Poincar\'e mapping at $z$ is
$$
P_z:=P_{\omega T}|_{\cP_z}.
$$
Of course
$$
\cP_z(\homsi{1})\subset\overline{\cP}_x:=
\left\{\Lambda^{t}(z^h),\ \Xi^{e-t}(z^h),\ t\in\bR^+,
%\mbox{ for all }
\ z^h= \Lambda^{\mu^h_\pm}(z) =\Xi^{\nu^h_\pm}(z)\in\homsi{1} \right\}.
%:=\left\{\Lambda^{e^{-nT}\mu^h_\pm}(z) =\Xi^{e^{-T}\nu^h_\pm}(z),\ n\in\bZ,\mbox{ for all }z^h= \Lambda^{\mu^h_\pm}(z) =\Xi^{\nu^h_\pm}(z)\in\homsi{1} \right\}.
$$

We define also
$$
\cP_\gamma(\homsi{1}):=\underset{z\in\gamma}{\cup}\cP_z(\homsi{1}).
$$
Note that
$$
\underset{\gamma\in \cH(\homsi{1})}\cup\cP_\gamma(\homsi{1})=\homsi{1}.
$$
\end{definition}
Taking now the coordinates $\mu,\nu\in\bR$ for the four parts of the stable and unstable manifold between  $z$ and $\Lambda^{\mu_\pm}(z)$ we get obviously that the eigenvectors of $P_{\omega T}$ have the form $\gamma_\lambda(\mu_\pm)=\mu_\pm^{i\lambda}$:
$$
P_{\omega T}\gamma_\lambda=e^{i\omega T\lambda}\gamma_\lambda.
$$

On the unstable manifold the ``continuation" of $\gamma_\lambda$ will be $\varphi_\lambda(\nu_\pm)=e^{-i\lambda\nu_\pm}$.

%\red{AFAIRE ...... speaking, as defined that way, $\cR_z$ is not a manifold: in particular it contains a self-crossing, and it is not obvious what should be a (good) system of coordinates on $\cP_z$.................}
Note that, defined that way, $\cP_z$ is not a manifold: in particular it contains a self-crossing, and it is not obvious what should be a (good) system of coordinates on $\cP_z$.
 Therefore we see that there are four systems of coordinates on $\cP_z$ around the point $z$ which hence appears as  singular in the sense that it is not clear what are ``regular" functions on $\cP_z$. This is a (the) main difference with the case of the Poincaré section $S^1$ in the case of the (regular) torus for the harmonic oscillator as seen in Sections \ref{harm} and \ref{freq}.
 
 Nothing in the classical paradigm provides a regularization of this singularity. But (one again) quantum mechanics suppress this singularity, just by putting a new structure on the phase space: a polarization. Without entering the formalism of geometrical quantization initiated by Souriau \cite{sour}, let us say that a choice of a polarization means that coordinates are somehow linked.  The more standard polarization is the so-called Schrödinger one: it consists in choosing   the pair of the configuration coordinate $x$ and the one on the cotangent fiber $p$,  and to link them by the Fourier transform.  In our present situation, another (but close) choice will be relevant: the pair of (germs of) the stable and unstable manifolds in phase space, which, being Lagrangian and transverse, are symplectomorphically equivalent to the $(x,p)$ axes.
 
 \newcommand{\poinc}{Poincaré }
 
 In the case of the \poinc section for the torus, a naturals system of coordinates is generated by the exponential functions $e^{inx},\ x\in [0,2\pi)$ which are also one of the two generators of the torus algebra. We will see in the next section the equivalent of this in the homoclinic setting
 \vskip 1cm

 \subsubsection{Bohr-Sommerfeld rules II}\label{BStwo}\ \ \ 
  In this section, we will give a groupoid construction of the homoclinic torus, and construct on it the action of the dynamical flow.This will show what a ``function of the noncommutative torus" is, leading to Bohr-Sommerfeld conditions and end up at the construction of a ``frequency mixed quasimode", mixed in the sense of being a density matrix.

% \red{AFAIRE1. contsruire le tore homo sous forme de groupoide
% 
% 2. montrer l' action du flot
% 
% 3. dire ce qu' est une fonction sur ce "tore"
% 
% 4. aboutir a BS
% 
% 5. vers incarnation sur un quasimode}
 \paragraph{Strategy}
 Let us suppose first that the invariant manifolds are (locally near $z$) exactly the $q$ and $p$ spaces of the phase space. We'll get back to the general situation at the end of this section.
 
 More precisely, we fix
 \begin{eqnarray}
 z&=&(0,0)\nonumber\\
 \Lambda^\pm_z:=\{\Lambda^u(z),\pm u\geq 0,|u|\leq \epsilon\}&\subset&\{(q_1,0,0,0,.....),\pm q_1\geq 0\}\nonumber\\
 \Xi^\pm_z:=\{\Xi^s(z),\pm s\geq 0,|s|\leq \epsilon\}&\subset&\{(0,p_1,0,0,.....),\pm p_1\geq 0\}\nonumber
 \end{eqnarray}
  Remembering that $\Lambda^+_z$ and $\Xi^+_z$ are connected in $\cP_z$, and the same for $\Lambda_-$ and $\Xi^-_-$, we want to define a ``regular" function on $\cP_z$ as a function which is the symbol of Lagrangian distributions associated to $\Lambda^\pm_z,\Xi^\pm_z$.
  
  What does this mean? 
  
  One passes from  $\Xi^+_z\cup\Xi^-_z$ to $\Lambda^+_z\cup\Lambda^-_z$ by a rotation in phase space of angle $\tfrac\pi4$: $(q,p)\mapsto(p,-q)$. 
  This corresponds to a Fourier transform.
  
  Therefore, a good function on $\overline{\cP}_z$ should satisfy the following property: the restriction of the function on $\Lambda^+_z\cup\Lambda^-_z$ must be the Fourier transform of the restriction of the function on $\Xi^+_z\cup\Xi^-_z$.
As we have seen, the \poinc mapping  acts by dilation by $e^{\pm T}$ of the variables $\mu,\nu$ so that eigenvector should be  complex powers of $\mu,\nu$, namely of the form $\mu^{i\lambda},\nu^{-i\lambda}$ (after extraction of the Jacobian term $\mu^{\frac12},\nu^{\frac12}$ as we shall see below).

Such Fourier transform can be found in \cite[p. 360]{gelf} (or recomputed after some easy complex dilations) and are given by the following formulas.

 \paragraph{The continuity condition}\
  
 For $\lambda\notin\pm\bN$,
 \be\label{gelfplus}
 \int_{-\infty}^{+\infty}
 x^\lambda_+
 e^{i\frac{x\sigma}\hbar}\tfrac{dx}{\sqrt{2\pi\hbar}}
 =
 i\Gamma(\lambda+1)\left[e^{i\lambda\frac\pi2}\sigma_+^{-\lambda-1}
 -e^{-i\lambda\frac\pi2}\sigma_-^{-\lambda-1}\right]\frac{\hbar^{\lambda+1}}{\sqrt{2\pi\hbar}}
 \ee
 \be\label{gelfmoins}
 \int_{-\infty}^{+\infty}
 x^\lambda_-
 e^{i\frac{x\sigma}\hbar}\tfrac{dx}{\sqrt{2\pi\hbar}}
 =
 i\Gamma(\lambda+1)\left[e^{i\lambda\frac\pi2}\sigma_-^{-\lambda-1}
 -e^{-i\lambda\frac\pi2}\sigma_+^{-\lambda-1}\right]\frac{\hbar^{\lambda+1}}{\sqrt{2\pi\hbar}}
 \ee
 We have four branches $\mu_\pm,\nu_\pm$ which have to be connected so that they form a system of coordinates for the ``torus" intersection of the two stable and unstable manifolds. This gives a set a values $\lambda_n,n=0,1,\dots$ by the following reasoning.
 
 We will have four functions, with $\lambda=i\freq$,
 $$
 c^s_+\mu^{i\freq-\frac12}_+,\ c^s_-\mu^{i\freq-\frac12}_-,\ c^u_+\nu^{-i\freq-\frac12}_+\mbox{ and }c^u_-\nu^{-i\freq-\frac12}_-.
 $$
  
  But by Definition \ref{defsect} of the \poinc section, the coefficient $c^u_+$ is related to $c^s_+$ and the same for $c^u_-$ and $c^s_-$:
  $
 \exists f_\freq^\pm\mbox{ such that }$
 \begin{eqnarray}
 c^u_+&=&f_\freq^+c^s_+\nonumber\\
 c^u_-&=&f_\freq^-c^s_-\nonumber
 \end{eqnarray}
 The ``continuity" equation will then read:
 $$
 \int_{-\infty}^{+\infty}
 (c^s_+\mu^{i\freq-\frac12}_++c^s_-\mu^{i\freq-\frac12}_-)e^{i\frac{\mu\nu}\hbar}\tfrac{d\mu}{\sqrt{2\pi\hbar}}
 =
% \int_{-\infty}^{+\infty}
 %(
 c^u_+\nu^{-i\freq-\frac12}_++c^u_-\nu^{-i\freq-\frac12}_-
 %)e^{i\mu\nu}d\nu
 =
 f_\freq^+c^s_+\nu^{-\lambda-\frac12}_++f_\freq^-c^s_-\nu^{-i\freq-\frac12}_-
 $$
 
 using \eqref{gelfplus} and \eqref{gelfmoins}, we get
 \begin{eqnarray}
 ic^s_+\Gamma(i\freq+\tfrac12)e^{-\freq\frac\pi2}e^{-i\frac\pi4}
 -
 ic^s_-\Gamma(i\freq+\tfrac12)e^{\freq\frac\pi2}e^{i\frac\pi4}
 =
 \sqrt{2\pi}\hbar^{-i\freq}f^+_\freq c^s_+\nonumber\\
 -ic^s_+\Gamma(i\freq+\tfrac12)e^{\freq\frac\pi2}e^{i\frac\pi4}
 +
 ic^s_-\Gamma(i\freq+\tfrac12)e^{-\freq\frac\pi2}e^{-i\frac\pi4}
 =
 \sqrt{2\pi}\hbar^{-i\freq}f_\freq^-c^s_-\nonumber
 \end{eqnarray}
 i.e.
 \be\label{eqlambda}
 U_\freq\begin{pmatrix}c^s_+\\c^s_-\end{pmatrix}=0
 \ee
 with
 \begin{eqnarray}\label{deful}
 U_\freq &=&
 \begin{pmatrix}
 i\Gamma(i\freq+\tfrac12)e^{-\freq\frac\pi2}e^{-i\frac\pi4}-\sqrt{2\pi}\hbar^{-i\freq}f^+_\freq &\ \ \ \ 
 &-i\Gamma(i\freq+\tfrac12)e^{\freq\frac\pi2}e^{i\frac\pi4}\\
 &&\\
 -i\Gamma(i\freq+\tfrac12)e^{\freq\frac\pi2}e^{i\frac\pi4} &\ \ \ \ 
 &i\Gamma(i\freq+\tfrac12)e^{-\freq\frac\pi2}e^{-i\frac\pi4}-\sqrt{2\pi}\hbar^{-i\freq}f^-_\freq
 \end{pmatrix}\nonumber\\
 &&
 \nonumber\\
 &=&
 \begin{pmatrix}
 \Gamma(i\freq+\tfrac12)e^{-\freq\frac\pi2}e^{i\frac\pi4}-\sqrt{2\pi}\hbar^{-i\freq}f^+_\freq &\ \ \ \ 
 &\Gamma(i\freq+\tfrac12)e^{\freq\frac\pi2}e^{-i\frac\pi4}\\
 &&\\
 \Gamma(i\freq+\tfrac12)e^{\freq\frac\pi2}e^{-i\frac\pi4} &\ \ \ \ 
 &\Gamma(i\freq+\tfrac12)e^{-\freq\frac\pi2}e^{i\frac\pi4}-\sqrt{2\pi}\hbar^{-i\freq}f^-_\freq
 \end{pmatrix}
 \end{eqnarray}
\paragraph{The transmission coeficcients}\ 
 
 Therefore the condition for \eqref{eqlambda} to have a non trivial solution is that
 $$
 \det{U_{\freq_n}}=0.
 $$
 
 We got to the point where we have to determine geometrically the two numbers $f^\pm_\freq$.
 
 As in Definition \ref{defsect} we define
 $$
 z^h= \Lambda^{\mu^h_\pm}(z) =\Xi^{\nu^h_\pm}(z).
 $$
 We impose that (remember that $\eta$ is a symplectic potential: $d\eta=\frac{dz\wedge d\bar z}{2i}$)
 $$
 e^{\frac i\hbar\int_z^{z^h}\eta^s}\ c^s_\pm(\mu^h_\pm)^{i\freq
 %-\frac12
 }=e^{\frac i\hbar\int_z^{z^h}\eta^u}c^u_\pm(\nu^h_\pm)^{-i\freq
 %-\frac12
 }
 \mbox{ i.e. }c^u_\pm=e^{\frac i\hbar\int_z^{z^h}\eta^u+\int_{zh}^z\eta^s}(\mu^h_\pm\nu^h_\pm)^{i\freq} c^u_\pm
 :=e^{\mp \frac i\hbar\oint\eta^h_\pm}(\mu^h_\pm\nu^h_\pm)^{i\freq} c^u_\pm
 $$
 where $\eta^s,\eta^u$ are the symplectic potentials taken over the stable, unstable manifold.
 
 Therefore 
 \be\label{deff}
 f^\pm_\freq=e^{\mp \frac i\hbar\oint\eta^h_\pm}(\mu^h_\pm\nu^h_\pm)^{i\freq}.
 \ee

Note tat the computation doesn't depend on the point $z^h$ Indeed, if
$$
z_1= \Lambda^{\alpha^h_\pm}(z) =\Xi^{\beta^h_\pm}(z).
$$
 then, for some $n\in\bZ$, 
 $$
 z_1^h=\Phi^{nT}z^h=\Lambda^{e^{t}\mu^h_\pm}...
 $$
 and so is ${\oint\eta^h}$.
\paragraph{Bohr-Sommerfeld conditions}\ 
 
 Therefore the condition \eqref{eqlambda} becomes
 $$
 \det
% \begin{pmatrix}
% i\Gamma({i\freq_n}+\tfrac12)e^{-{\freq_n}\frac\pi2}e^{-i\frac\pi2}-
% %f^+_\freq_n
%\sqrt{2\pi}\hbar^{-i\freq}
%e^{-\frac i\hbar\oint\eta^h_+}(\mu^h_+\nu^h_+)^{i\freq_n} 
%  &\ \ \ \ 
% &-i\Gamma({i\freq_n}+\tfrac12)e^{{\freq_n}\frac\pi2}e^{i\frac\pi2}\\
% &&\\
% -i\Gamma({i\freq_n}+\tfrac12)e^{{\freq_n}\frac\pi2}e^{i\frac\pi2} &\ \ \ \ 
% &i\Gamma({i\freq_n}+\tfrac12)e^{-{\freq_n}\frac\pi2}e^{-i\frac\pi2}-
% %f^-_\freq_n
% \sqrt{2\pi}\hbar^{-i\freq}
% e^{\frac i\hbar\oint\eta^h_-}(\mu^h_-\nu^h_-)^{i\freq_n}
% \end{pmatrix}
%\begin{pmatrix}
% \Gamma(i\freq+\tfrac12)e^{-\freq\frac\pi2}-\sqrt{2\pi}\hbar^{-i\freq}f^+_\freq &\ \ \ \ 
% &\Gamma(i\freq+\tfrac12)e^{\freq\frac\pi2}\\
% &&\\
% \Gamma(i\freq+\tfrac12)e^{\freq\frac\pi2} &\ \ \ \ 
% &\Gamma(i\freq+\tfrac12)e^{-\freq\frac\pi2}-\sqrt{2\pi}\hbar^{-i\freq}f^-_\freq
% \end{pmatrix}
\begin{pmatrix}
 \Gamma(i\freq+\tfrac12)e^{-\freq\frac\pi2}e^{i\frac\pi4}-\sqrt{2\pi}\hbar^{-i\freq}f^+_\freq &\ \ \ \ 
 &\Gamma(i\freq+\tfrac12)e^{\freq\frac\pi2}e^{-i\frac\pi4}\\
 &&\\
 \Gamma(i\freq+\tfrac12)e^{\freq\frac\pi2}e^{-i\frac\pi4} &\ \ \ \ 
 &\Gamma(i\freq+\tfrac12)e^{-\freq\frac\pi2}e^{i\frac\pi4}-\sqrt{2\pi}\hbar^{-i\freq}f^-_\freq
 \end{pmatrix}
 =0
 $$
 i.e.
% $$
% -2\Gamma({i\freq_n}+\tfrac12)^2\sinh(\freq_n)-i\Gamma({i\freq_n}+\tfrac12)
% \sqrt{2\pi}\hbar^{-i\freq}
% \big(e^{\frac i\hbar\oint\eta^h_-}(\mu^h_-\nu^h_-)^{i\freq_n}
% +
% e^{-\frac i\hbar\oint\eta^h_+}(\mu^h_+\nu^h_+)^{i\freq_n} 
%  \big)
%  +
%  2\pi\hbar^{-2i\freq}e^{\frac i\hbar\big(\oint\eta^h_--\oint\eta^h_+\big)}
%  (\mu^h_+\mu^h_-\nu^h_+\nu^h_-)^{i\freq_n}.
% $$
\begin{eqnarray}
0&=&
 2i\Gamma({i\freq_n}+\tfrac12)^2\cosh(\pi\freq_n)\nonumber\\
 &-&\Gamma({i\freq_n}+\tfrac12)e^{-\freq_n\frac\pi2}e^{i\frac\pi4}
 \sqrt{2\pi}\hbar^{-i\freq}\big(e^{\frac i\hbar\oint\eta^h_-}(\mu^h_-\nu^h_-)^{i\freq_n}
 +
 e^{-\frac i\hbar\oint\eta^h_+}(\mu^h_+\nu^h_+)^{i\freq_n} 
  \big)
 \nonumber\\
   &+&
  2\pi\hbar^{-2i\freq}e^{\frac i\hbar\big(\oint\eta^h_--\oint\eta^h_+\big)}(\mu^h_+\mu^h_-\nu^h_+\nu^h_-)^{i\freq_n}
  %=0
  \label{cbsII}
     \end{eqnarray}
 or
 $$
 i\Gamma({i\freq_n}+\tfrac12)^2\frac{\cosh(\pi\freq_n)}\pi
 -\Gamma({i\freq_n}+\tfrac12)e^{-\freq_n\frac\pi2}e^{i\frac\pi4}
 \frac{\hbar^{-i\freq}}{\sqrt{2\pi}}\big(e^{\frac i\hbar\oint\eta^h_-}(\mu^h_-\nu^h_-)^{i\freq_n}
 +
 e^{-\frac i\hbar\oint\eta^h_+}(\mu^h_+\nu^h_+)^{i\freq_n} 
  \big)
 $$
 $$
  +\hbar^{-2i\freq}e^{\frac i\hbar\big(\oint\eta^h_--\oint\eta^h_+\big)}(\mu^h_+\mu^h_-\nu^h_+\nu^h_-)^{i\freq_n}=0
 $$
 or also
 $$
 i\Gamma({i\freq_n}+\tfrac12)^2\frac{\cosh(\pi\freq_n)}\pi
 -\frac{\Gamma({i\freq_n}+\tfrac12)e^{-\freq_n\frac\pi2}e^{i\frac\pi4}}{\sqrt{2\pi}}
 \left(e^{\frac i\hbar\oint\eta^h_-}\left(\frac{\mu^h_-\nu^h_-}{\hbar}\right)^{i\freq_n}+
 e^{-\frac i\hbar\oint\eta^h_+}\left(\frac{\mu^h_+\nu^h_+}\hbar\right)^{i\freq_n} 
  \right)
  $$
  $$
 +
  e^{\frac i\hbar\big(\oint\eta^h_--\oint\eta^h_+\big)}\left(\frac{\mu^h_+\nu^h_+}\hbar\frac{\mu^h_-\nu^h_-}\hbar\right)^{i\freq_n}=0.
 $$
 
 In particular, under (BSCI),
 $$
 i\Gamma({i\freq_n}+\tfrac12)^2\frac{\cosh(\pi\freq_n)}\pi-\frac{\Gamma({i\freq_n}+\tfrac12)e^{-\freq_n\frac\pi2}e^{i\frac\pi4}}{\sqrt{2\pi}}
 \left(
 %e^{\frac i\hbar\oint\eta^h_-}
 \left(\frac{\mu^h_-\nu^h_-}{\hbar}\right)^{i\freq_n}
 +
 %e^{-\frac i\hbar\oint\eta^h_+}
 \left(\frac{\mu^h_+\nu^h_+}\hbar\right)^{i\freq_n} 
  \right)
 +
  %\hbar^{-2i\freq}
  %e^{\frac i\hbar\big(\oint\eta^h_--\oint\eta^h_+\big)}
  \left(\frac{\mu^h_+\nu^h_+}\hbar\frac{\mu^h_-\nu^h_-}\hbar\right)^{i\freq_n}=0.
 $$
 and under ``symmetry" $\pm$, $\mu^h_+=\mu^h_-:=\mu^h,\nu^h_+=\nu^h_-=\nu^h$,
 
 $$
 i\Gamma({i\freq_n}+\tfrac12)^2\frac{\cosh(\pi\freq_n)}\pi-2\frac{\Gamma({i\freq_n}+\tfrac12)}{\sqrt{2\pi}}e^{-\freq_n\frac\pi2}e^{i\frac\pi4}
 %\left(
 %e^{\frac i\hbar\oint\eta^h_-}
 \left(\frac{\mu^h\nu^h}{\hbar}\right)^{i\freq_n}
% +
% %e^{-\frac i\hbar\oint\eta^h_+}
% \left(\frac{\mu^h_+\nu^h_+}\hbar\right)^{i\freq_n} 
%  \right)
 -
  %\hbar^{-2i\freq}
  %e^{\frac i\hbar\big(\oint\eta^h_--\oint\eta^h_+\big)}
  \left(\frac{\mu^h\nu^h}\hbar
  %\frac{\mu^h_-\nu^h_-}\hbar
  \right)
  ^{2i\freq_n}=0.
 $$
 which gives
 \begin{eqnarray}
 \Gamma({i\freq_n}+\tfrac12)
&=&
 \pi \left(\frac{\mu^h\nu^h}{\hbar}\right)^{i\freq_n}
 \frac{
 \frac{-2e^{-\freq_n\frac\pi2}e^{i\frac\pi4}}{\sqrt{2\pi}}
 \pm\sqrt{i\frac2\pi e^{-2\freq_n\frac\pi2}-4i\frac{\cosh(\pi\freq_n)}\pi}}{2\cosh(\pi\freq_n)}\nonumber\\
 &=&
 \sqrt{2\pi} \left(\frac{\mu^h\nu^h}{\hbar}\right)^{i\freq_n}e^{i\frac\pi4}\frac{
 %1\pm\sqrt{1-2{\cosh(\freq_n)}}
 -e^{-\freq_n\frac\pi2}\pm ie^{+\freq_n\frac\pi2}
 }{2\cosh(\pi\freq_n)}\nonumber
 \end{eqnarray}
 and
 $$
 %\left|jhg\right|
 \left|\sqrt{2\pi} \left(\frac{\mu^h\nu^h}{\hbar}\right)^{i\freq_n}\frac{
 %1\pm\sqrt{1-2{\cosh(\freq_n)}}
 -e^{-\freq_n\frac\pi2}\pm ie^{+\freq_n\frac\pi2}
 }{2\cosh(\pi\freq_n)}\right|^2
 =2\pi\frac{\cosh{(\pi\freq_n)}}{2\cosh{(\pi\freq_n)}^2}=\frac\pi{\cosh{(\pi\freq_n)}}=|\Gamma(i\freq_n+\tfrac12)|^2
 $$
 since, by the complement formula we have
 $$
 |\Gamma({i\freq_n}+\tfrac12)|=\sqrt{\Gamma({i\freq_n}+\tfrac12)
 \Gamma({-i\freq_n}+\tfrac12)}=\sqrt{\frac\pi{\sin{(\pi({i\freq_n}+\tfrac12))}}}=\sqrt{\frac\pi{\cosh{(\pi{\freq_n})}}}.
 $$
% \subsubsubsection{jhg}
Therefore we get the (BSII) condition
$$
\cos{\left(\arg{\Gamma(i\freq_n+\tfrac12)}+\freq_n\log{\tfrac{\mu^h\nu^h}\hbar}+\tfrac\pi4\right)}
=
-\frac{e^{-\freq_n\frac\pi2}}{\sqrt{2\cosh{(\pi\freq_n)}}}
=
-\frac{1}{\sqrt{1+e^{\pi\freq_n}}}.
$$

In the symmetric case without (BSI) we get
\begin{eqnarray}
&&\cos{\left(\arg{\Gamma(i\freq_n+\tfrac12)}+\freq_n\log{\tfrac{\mu^h\nu^h}\hbar}+\tfrac\pi4\right)}\nonumber\\
&=&
{\sqrt{\cosh{(\pi\freq_n)}}}\Re{\left(e^{-\freq_n\frac\pi2}\left(e^{\frac i\hbar\oint\eta_+}+e^{\frac i\hbar\oint\eta_-}\right)+\sqrt{e^{\freq_n\pi}e^{\frac   i\hbar(\oint\eta_--\oint\eta_+)}+e^{-\freq_n\pi}\left(e^{\frac i\hbar\oint\eta_+}-e^{\frac i\hbar\oint\eta_-}\right)}\right)}.\nonumber
\end{eqnarray}
and in particular if $\oint\eta_+=\oint\eta_-:=\oint\eta$,
$$
\cos{\left(\arg{\Gamma(i\freq_n+\tfrac12)}
+
\tfrac{\oint\eta}\hbar
+
\freq_n\log{\tfrac{\mu^h\nu^h}\hbar}+\tfrac\pi4\right)}
%=
%-\frac{e^{-\freq_n\frac\pi2}}{\sqrt{2\cosh{(\pi\freq_n)}}}
=
-\frac{1}{1+e^{\pi\freq_n}}.
$$

In the fully general case:
 \begin{eqnarray}
&&\cos{\left(\arg{\Gamma(i\freq_n+\tfrac12)}-\freq_n\log{\hbar}+\tfrac\pi4\right)}\nonumber\\
&=&
{\sqrt{\cosh{(\pi\freq_n)}}}\Re\left(e^{-\freq_n\frac\pi2}\left(e^{\frac i\hbar\oint\eta_+}(\mu^h_+\nu^h_+)^{i\freq_n}+e^{\frac i\hbar\oint\eta_-}(\mu^h_-\nu^h_-)^{i\freq_n}\right)\right.\nonumber\\
&&\left.+\sqrt{e^{\freq_n\pi}e^{\frac   i\hbar(\oint\eta_--\oint\eta_+)}(\mu^h_+\nu^h_+\mu^h_-\nu^h_-)^{i\freq_n}+e^{-\freq_n\pi}\left(e^{\frac i\hbar\oint\eta_+}(\mu^h_+\nu^h_+)^{i\freq_n}-e^{\frac i\hbar\oint\eta_-}(\mu^h_-\nu^h_-)^{i\freq_n}\right)}\right).\nonumber
\end{eqnarray}
or
\begin{eqnarray}
&&\cos{\left(\arg{\Gamma(i\freq_n+\tfrac12)}+\freq_n\left(\tfrac12\log(\mu^h_+\nu^h_+\mu^h_-\nu^h_-)-\log{\hbar}+\tfrac\pi4\right)\right)}\nonumber\\
&=&
{\sqrt{\cosh{(\pi\freq_n)}}}\Re\left(e^{-\freq_n\frac\pi2}\left(e^{\frac i\hbar\oint\eta_+}(\mu^h_-\nu^h_-)^{-i\freq_n}+e^{\frac i\hbar\oint\eta_-}(\mu^h_+\nu^h_+)^{-i\freq_n}\right)
\right.\nonumber\\
&&\left.
+\sqrt{e^{\freq_n\pi}e^{\frac   i\hbar(\oint\eta_--\oint\eta_+)}
%(\mu^h_+\nu^h_+\mu^h_-\nu^h_-)^{i\freq_n}
+e^{-\freq_n\pi}\left(e^{\frac i\hbar\oint\eta_+}(\mu^h_-\nu^h_-)^{-i\freq_n}-e^{\frac i\hbar\oint\eta_-}(\mu^h_+\nu^h_+)^{-i\freq_n}\right)}\right).\label{gamasym0}
\end{eqnarray}
This is quite complicated expression, but note that, as $\freq_n\to\infty$,
\be\label{gamasym}
\arg\Gamma(i\freq_n+\tfrac12)\sim \freq_n\log{\freq_n}-\freq_n+\frac1{24\freq_n}+O(\freq_n^{-2}).
\ee
%One note the similitude between the different euqations satisfied by the frequencies 
%\red{AFAIRE PRECISER UN PEU LE CAS GENERAL
%
%}
 %\vskip 3cm.
 \paragraph{The spectrum of frequencies}\ 
 
 The preceding discussion leads to the following conclusion:
 
 \be\label{bs2}
 \big\{\lambda\omega
 %\lambda
 _n+\tfrac{2\pi}Tm,(n,m)\in\bZ^2\big\}\subset\big\{\lim_{\hbar\to 0}\tfrac{E_i-E_j}\hbar,\ i,j\in\bN\big\}.
 \ee
 
 One note the similitude between the different equations satisfied by the frequencies $\freq_n$ and the Bohr-Sommerfeld conditions obtained  for the one dimensional systems near a separatrix  \cite{CO,CP, rp}. The fact that equivalent formulas are valid for eigenvalues in the double well case and differences of eigenvalues in our setting can be explained by the fact that in the separatrix case the energy of the separatrix ix fixed equal to $0$, otherwise the formula would be valid for $E_{separatrix}-E_j$. In our case, the Hamiltonian is by no means fixed anywhere, as it shouldn't be since the addition of a constant to the Hamiltonian doesn't modify the dynamics.
% and in the fully general case:
% \begin{eqnarray}
%&&\cos{\left(\arg{\Gamma(i\freq_n+\tfrac12)}-\freq_n\log{\hbar}\right)}\nonumber\\
%&=&
%{\sqrt{\cosh{(\pi\freq_n)}}}\Re{\left(e^{-\freq_n\frac\pi2}(e^{\frac i\hbar\oint\eta_+}+e^{\frac i\hbar\oint\eta_-}+\sqrt{e^{\freq_n}e^{\frac   i\hbar(\oint\eta_--\oint\eta_+)}+e^{-\freq_n}\left(e^{\frac i\hbar\oint\eta_+}-e^{\frac i\hbar\oint\eta_-}\right)}\right)}.\nonumber
%\end{eqnarray}
 
% \red{AFAIRE MENTIONNER QUE L'EQUIVALENT DU TORE DOIT ETRE INFINITESEIMAL AUTOUR DU POINT DE LA TRAJECTOIRE.}
A closer look at the quantization formulas \eqref{gamasym0}, \eqref{gamasym}  shows that what is important s the germ of "eigenfunctions" of the Poincar\'e mapping near the trajectory, On the contrary of the integrable case which involves the whole Poincar\'e section. This is a trace of the fact that the density of homoclinic orbits is visible as the accumulation of them near the trajectory (and not by an ergodic property of the Poincar\'e map as in the integrable case).
 \paragraph{Analysis of the ``spectrum"  as $\hbar\to 0$}\ 
 %\subparagraph*{jhg}
 Let us look at the symmetric case without (BSI):
 \be\label{symwithoutBSI}
\cos{\left(\arg{\Gamma(i\freq_n+\tfrac12)}
+
\tfrac{\oint\eta}\hbar
+
\freq_n\log{\tfrac{\mu^h\nu^h}\hbar}+\tfrac\pi4\right)}
=
-\frac{1}{1+e^{\pi\freq_n}}.
\ee
\begin{itemize}
\item high part of the spectrum 

Thanks to \eqref{gamasym}, \eqref{symwithoutBSI} becomes as $\freq_n\ll\tfrac1\hbar\to\infty$
\begin{eqnarray}
\cos{(\freq_n\log{\freq_n}-\freq_n+\frac1{24\freq_n}
+
\tfrac{\oint\eta}\hbar
+
\freq_n\log{\tfrac{\mu^h\nu^h}\hbar}+\tfrac\pi4+O(\freq_n^{-2}))}&=&O(\freq_n^{-\infty})\nonumber\\
&\Updownarrow&\nonumber\\
\freq_n\log{\freq_n}-\freq_n+\frac1{24\freq_n}
+
\tfrac{\oint\eta}\hbar
+
\freq_n\log{\tfrac{\mu^h\nu^h}\hbar}&=&\tfrac\pi4+n\pi+O(\freq_n^{-2})),\ n\in\bZ\nonumber\\
&\Updownarrow&\nonumber\\
\freq_n&\sim& \tfrac1{\log{\hbar}}(\tfrac{\oint\eta}\hbar-\tfrac\pi4+n\pi),\ n\in\bZ.\nonumber
\end{eqnarray}
\item low part of the spectrum 

 \eqref{symwithoutBSI} becomes, as $\freq_n=o(1)\mbox{ as }\hbar\to\infty$, since $\Gamma(\tfrac12)=,\sqrt\pi\ \Gamma'(\tfrac12)=-\sqrt\pi(\gamma+2\log{2})   $ where $\gamma$ is the Euler-Mascheroni constant,
 so that
 \begin{eqnarray}
 \frac d{d\freq}\arg{\Gamma(i\freq+\tfrac12)}|_{\freq=0}&=&
 \frac1{2i}\left( \frac d{d\freq}\frac{\Gamma(i\freq+\tfrac12)}{\Gamma(-i\freq+\tfrac12)}\right){\frac{\Gamma(-i\freq+\tfrac12)}{\Gamma(i\freq+\tfrac12)}}|_{\freq=0}\nonumber\\
 &=&
 \frac1{2i}\frac{i\Gamma'(i\freq+\tfrac12)\Gamma(-i\freq+\tfrac12)+i\Gamma'(-i\freq+\tfrac12)\Gamma(i\freq+\tfrac12)}{\Gamma(-i\freq+\tfrac12)^2}{\frac{\Gamma(-i\freq+\tfrac12)}{\Gamma(i\freq+\tfrac12)}}|_{\freq=0}\nonumber\\
 &=&
 \frac{\Gamma'(\tfrac12)}{\Gamma(\tfrac12)}=-\gamma-2\log{2},\nonumber
 \end{eqnarray}
 \begin{eqnarray}
 \cos{\left(\sqrt\pi-(\gamma+2\log{2})\freq_n
+
\tfrac{\oint\eta}\hbar
+
\freq_n\log{\tfrac{\mu^h\nu^h}\hbar}+\tfrac\pi4+O(\freq_n^2)\right)}
&=&
-\frac{1}2+\frac{\pi\freq_n}4+O(\freq_n^2)\nonumber\\
&\Updownarrow&\nonumber\\
\freq_n&=&
%\pi\pm\left(\frac{\pi}3+\frac\pi{2\sqrt3}\freq_n\right)
\frac1{\log{\hbar}}(-\frac{3\pi}4-\sqrt\pi+\tfrac{\oint\eta}\hbar+n2\pi)+O(\tfrac1{\log{\hbar}^2}),
\ n\in\bZ.
\nonumber
%{1+e^{\pi\freq_n}}
 \end{eqnarray}
 \item number of frequencies
 
 an easy computation shows that
 $$
 \#\{\freq_n,\ |\freq_n|\leq 1\}\sim\log{\hbar}.
 $$
\end{itemize}

\vskip 1cm
% As $\hbar\to 0$
 
% \red{AFAIREDENSITE DU SPETRE
% 
% DIFFERENTES HOMOCLINES?
% 
% TORE OU APS TORE
% 
% }
 Let us remark to finish this paragraph that the spectrum of frequencies in the chaotic situation is denser than in the integrable case. This is due to the term $\log\hbar$ present in the  two last formulas. This is true as soon as we consider the contribution of a single homoclinic orbit, and is a fortiori valid when all of them are considered in a row.
 On the other side, a single homoclinic orbit might give the full frequencies spectrum, as, as we saw before, what is important is the germ near an infinitesimal neighbourhood of the trajectory. In this neighbourhood,  all the homoclinic orbits are dense, like, by ergodicity, in the integrable case in the full Poincar\'e section. Therefore, an homoclinic torus might be associated to a (any in fact) single homoclinic orbit, as the regular torus in the case of the non-resonant harmonic oscillator.
 \paragraph{Homoclinic torus}\ 
 
 We saw clearly that, at $\hbar\to 0$, the spectrum of frequencies doesn't depend, at leading order, to the invariant quantity $\mu^h\nu^h$ associated to each homoclinic orbit. More, $\mu^h\nu^h$ appears as a kind of Maslov index.
 
 Therefore, the structure limits, that is the sets of allowed functions on the \poinc section is really associated to the set of intersections of the two invariant manifolds: one has to consider all the homoclinic orbits $\gamma^h_i$ such that $\frac{\oint^{\gamma^h_i}\eta-\oint^{\gamma^h_j}}\hbar+O(\tfrac1{\log{\hbar}})\in 2\pi\bZ$. 
 
 In other words, calling prequantized the tori associated to the same ``space of functions", 
 $$
 \{\mbox{prequantized tori}\}=\{\mbox{homoclinic orbits} \gamma^h_i\}\backslash \sim,\ \ \gamma^h_i\sim\gamma^h_j\Leftrightarrow \oint_{\gamma^h_i}^{\gamma^h_j}\eta^s+\oint_{\gamma^h_j}^{\gamma^h_i}\eta^=:=\oint_{\gamma^h_i}^{\gamma^h_j}\eta^h+O(\tfrac\hbar{\log{\hbar}})\in 2\pi\hbar\bZ.
 $$
 
 \textit{It is a polarization in the framework of geometric quantization \cite{sour}.}
 
 Note that this tori are the equivalent to the tori surrounding periodic linearly stable trajectories, in the construction of quasimodes ``\`a la Ralston" \cite{R}. Note that in our case also, the tori are infinitesimal around the periodic trajectory\footnote{It would be interesting to look at our construction near an unstable periodic trajectory in the limit where the period diverges, and see if appears naturally a limit of the frequencies given by $\omega_k/\log\hbar$.} .
 
% \red{AFAIRE
% 
% PARLER DE TORE INFINTESIMAL, PUSIQUE POUR DES HOMOLCLINE DIFFERENTES LA SECTIONDE POINCARE EST DIFFERENT LOIN DE GAMMA MAIS ALA MEME TRACE PRES DE GAMM
% 
% REGARDER AUSSI LA LIMITE T VERS INFINI LOGARITHMIQUEMENT EN HBAR POUR FAIRE APPARAITRE OMEGA/LOGHBAR COMME LIMITE DE FREQUENCES....
% }
 \paragraph{General geometrical setting}\ 
 In the last section, we considered the case where the stable and unstable manifold are (tangent to) the $q$- and $p$-spaces. The general case of two invariant manifolds is got by the combination of  a rotation and a dilation 
 in phase space: for example a rotation by $-\pi/4$ sends the generation of dilation $qp$ to $\frac{p^2-q^2}2$
 and a dilation by $\alpha$ to $\alpha^{-2}(p^2-\alpha^2q^2)$ whose 
 invariant manifolds are the two lines $p=\pm\alpha q$.
 
 The quantum action of any of these linear symplectic mappings are metaplectic operators $M$ and  $U_\freq$ becomes
 $$
 \begin{pmatrix}
 M&0\\0&M
 \end{pmatrix}U_\freq
 \begin{pmatrix}
 M^{-1}&0\\0&M^{-1}
 \end{pmatrix}=U_\freq.
 $$ 
 so that the matrix, and therefore the discussion of the preceding section, doesn't depend on the geometrical setting. 
 
 \paragraph{On the spacing of energy levels}\ 
 Let us come back to \eqref{bs2} which expresses (some) differences of eigenvalues divided by $\hbar$ as $\lambda\freq_n+\tfrac{ 2\pi}Tm,\ n.m\in\bZ$:
 $$
 \frac{E_i-E_j}\hbar\sim \lambda\freq_n+\tfrac{ 2\pi}Tm.
 $$
 If we order the spectrum $\{E_i\}$ by increasing order in $i$, we find that $\frac{E_{i+1}-E_i}\hbar$ should be the smallest of all the numbers $\frac{E_{j}-E_i}\hbar, \ j\in\bZ$, kipping $i$ fixed, so we are looking at small values of $\lambda\freq_n+\tfrac{ 2\pi}Tm$, which can appear only when $\freq_n$ diverges.
 
 Therefore we look at the high part of the frequencies spectrum, that leads to look at the quantity
 $$
\tfrac1{\log{\hbar}}(\tfrac{\oint\eta}\hbar-\tfrac\pi4+n\pi) +\tfrac{ 2\pi}Tm, \ \ \ n,m\to\pm\infty.
 $$
 and look at how accumulate these numbers near $-\frac\pi4$ as $n,m\to\pm\infty$.
 \paragraph{Conclusion}\ 
 
 We see that, on the contrary with the integrable case presented in Sections \ref{harm} and \ref{freq}, the mean of the spacing between two ``eigenfrequencies" is not of order $1$ any more, but of order $\tfrac1\log{\hbar}$.
 
 This is somehow in accordance with the case  of estimates given by trace formula around a critical point of the Hamiltonian, as shown in \cite{BPU}: in our case in this paper, there is no critical Hamiltonian, but the dynamics of the \poinc resemble to the one near a critical point. The trace of the fact that a (quantized) mapping instead of a flow is concerned 
 %a (quantized) mapping instead of a flow
  is visible on the fact that our prediction concerns frequencies and not eigenvalues: no value of the Hamiltonian is fixed here in advance.  
 
 Once again, our goal is not to find eigenvectors, a task which seems to escape from semiclassical research activity, but to look at superposition of them that are well localized in noncommutative ``phase space".
 
 We will say more on this in the next section but a few remarks are in order to conclude this paragraph. 
 
 Up to now we treat only the case where the stable and unstable manifolds are orthogonal,i.e. project on the canonical system of coordinates on $T^*\R^d$. the general case can be pull back to this case by a symplectic change of variables, leading in fact to an unessential change in the quantization formulas, as expected earlier.
 
 The construction provided accentuate strongly the importance of looking technically at frequencies and not at eigenvalues. The physical interest fo this paradigmatic change was already mentioned earlier. 
 
 The logarithmic factor present in all quantization formulas shows clearly that the frequency spectrum, though being fully discrete (because the eigenvalues one is) tries to mimic, as $\hbar\to 0$ a continuous one (as this is even accentuated if one consider all homoclinic orbits in a row). This remark is not without reminding the quantum dynamics for Hamiltonian with continuous spectra, leading to diffusion on the contrary of the discrete case leading to quasiperiodic bounded evolution. This is certainly a source of what one can call, in this framework, quantum chaos.
 \vskip 0.5cm
% \red{AFAIRE PARLER DE LA PLUS PETITE OMEGA DIFFERENTE DE ZERO ET STATISTIQUES D'ECARTS DE VALEURS PROPRES (FAIT PLUS BAS DANS MISCENELEOUS
% 
% FAIRE LE CAS OU ELS VARIETES INVARIANTES NE SONT PAS ORTHOGOBALES
% 
% L'IMPACT DE CA  SUR  FORMULE DE QUAMTIFICATION
%
%  
%  
%  }
 
%\vskip 3cm 
% \red{AFAIREEXPLIUQER POURQUOI DIFFERENCES DE VALEURS PROPRES.
% 
% -LIEN AVEC...
% 
% -FACTEUR LOGATYRTHMIQUE
% 
% -RETOUR SUR TOUTES LES HOMOCLINES (?)
% 
% }
 %\vskip 3cm
 \subsubsection{Creation, annihilation and all that}\label{creann}\ 
 
 In the preceding section, we have ``guessed" by a geometrical analogy with the case of the resonant harmonic oscillator what should be frequencies quantization formula. Mow we need to construct  a quaismode corresponding to this guess.
 
 We will look at such a quasimode by using an ansatz given by formula \eqref{betsymb} in Section \ref{anogroup}, namely we look at an operator of the form
 \be\label{betasymbol}
A(\beta)=\int_{h^{-1}(I)}dz\int_{\Lambda_{z}}dz'e^{\frac i\hbar\int_{z}^{z'}\eta}
\beta(z',z)
|\vp_{z'}\rangle\langle\vp_{z}|.
\ee
Looking at $A(\beta)$ ``eigenvector" of the conjugation by the flow leads, by \eqref{propbeta}, to
$$
e^{i\frac{t(\lambda_i-\lambda_j)}\hbar} A(\beta)=e^{-i\frac{tH}\hbar}A(\beta)
e^{i\frac{tH}\hbar}=A((\Phi^t)^{\otimes 2}\#\beta)
$$
that is
\be\label{eqvponbeta}
(\Phi^t)^{\otimes 2}\#\beta=e^{i\frac{t(\lambda_i-\lambda_j)}\hbar}\beta.
\ee

The eigenvectors of the mapping $(\Phi^t)^{\otimes 2}\#$ will be exponential functions in the 
%following 
system of coordinates given by the horocyclic flow. Of course, the groupoid aspect expressed above leads to some discreteness of these coordinates, but the density of homoclinic trajectories gives an equivalent to continuous coordinate at the level of germs.

%\red{AFAIREdecrire els coordonnees et les fonctions $e^{i\lambda (coordonnees)}$}
%
%\red{et la retour au groupoide de la section precedente.}

In fact, what we need is that 
$$
\int_{\overline{\cP}_{z}}dz'e^{\frac i\hbar\int_{z}^{z'}\eta}
\beta(z',z)
|\vp_{z'}\rangle
$$
is a Lagrangian distribution. 

This corresponds to have a ``continuity condition" of the form
$$
\int_{\Lambda^+_{z}\cup\Lambda^-_z}dz'e^{\frac i\hbar\int_{z}^{z'}\eta}
\beta_{\lambda_k}(z',z)
|\vp_{z'}\rangle=\int_{\Xi^+_{z}\cup\Xi^-_z}dz'e^{\frac i\hbar\int_{z}^{z'}\eta}
\beta_{\lambda_k}(z',z)
|\vp_{z'}\rangle
$$

And this obliges ${\lambda_k}$ to satisfy the Bohr-Sommerfeld condition \eqref{bs2} of the preceding section.

This construction leads to ``eigenvectors" (more precisely ``quasidensity matrices") of the quantum flow at times multiples of $T$ and one recover, as usual, a total eigenvector, i.e. for any time $t$, of the quantum flow by integrating over a period the conjugate by the flow of the preceding construction.

We get
\be\label{ergcha}
A_{\lambda_k}=\frac1{T_\gamma}\int_0^{T_\gamma}\int_{\Lambda^+_{z}\cup\Lambda^-_z}dz'
e^{\frac i\hbar\int_{z}^{z'}\eta}
\beta_{\lambda_k}(z',\Phi^t(z))
\adir{\vp_{z'}}{\vp_{\Phi^t(z)}}.
\ee

Let us remark the strong link  between formulas \eqref{ergcha} and \eqref{ergtor} and therefore \eqref{erg}: the integral in \eqref{ergcha} can bee considered as an integral over two cycles of the homoclinic ``torus". Therefore it is the strict equivalent to \eqref{ergtor} and, thanks to the first Bohr-Sommerfeld condition(s), can be seen also as an ergodic ``average over the flow".
\vskip 0.3cm
This remark will become more lighting by 
%the following 
a ``groupoid" construction leading to a more geometrical construction of the operator $A_{\lambda_k}$ in \eqref{ergcha}.

%\red{AFAIRE'ecrire ca mieux avec groupoide}
%\vskip 3cm

In the preceding paragraph, we associate to the homoclinic algebra - quotient of the phase space by the flow - the classical limit of (part of) the set of frequencies of the quantum evolution. What are quantum frequencies? They are (-i times) the eigenvalues of the Heisenberg-von Neumann derivation
$$
\frac i\hbar[H,\cdot].
$$
What are the ``eigenvectors"? Obviously off-diagonal dyadic operators $\adir{\psi_i}{\psi_j}$ where $\psi_k$ is an eigenvector of eigenvalue $E_k$.

On the other side, to $\psi_{\lambda_n}$ one can associate  operators, by the same trick that what we did in Section \ref{harm}, to  elements of the homoclinic algebra reduced to the ``Poincaré" section which are eigenvectors of the shift operator  $P_{\omega T}$ of eigenvalue $\lambda_n,\ n\in\bZ$.

These operators are somehow functions of the creation and annihilation operators since they produce a shift in the spectrum $\psi_{E_i}\to\psi_{E_i+\lambda_n}$. Even in the case of the harmonic oscillator, the shift operators are very singular, since they reads
$$
\frac{a^\pm}{\sqrt{H_0}}\mbox{ with } a^\pm=x\pm\tfrac d{dx},
$$
which is singular at the origin, classically.  In fact the formulation
\eqref{creanil} with $n=\pm1$  desingularizes this singularity, by a non-local formula, already
for the rather spectrally stupid harmonic oscillator.

\vskip 1cm

Spectral shift operators are not  ``regular" observables, by far. they are semiclassical operators associated to noncommutative symbols, belonging to the noncommutative algebra of the homoclinic ``foliation" of our system in the chaotic case, or to the noncommutative algebra of the invariant tori foliation in the integrable case.

The construction of off-diagonal operators we propose is directly inspired from the construction established  in Section \ref{anogroup}.
 
 \vskip 1cm

 \subsubsection{A new and noncommutative framework for classical dynamics}\label{newnon}\ 
 
 The construction above shows clearly that the limit $\hbar\to 0$ of the quantum dynamics involves a much richer structure than the one provided by the classical paradigm, namely a nice geometrical space host a nice flow.
 
 Of course one knows since Poincaré that this \textit{nice} flow on a \textit{nice} space produces, when long evolutions are considered, \textit{not nice} (at all!) structures, such as complex foliations that Poincaré himself described as not possible to draw. And the chaos appears precisely when one wants to merge these singular objects with the natural original space: chaos is the trace on the phase space of the homoclinic machinery.
 
 But what we see is that the true classical limit of quantum mechanics - namely the limit of the dynamics when $\hbar\to 0$ - involves these complex ``homoclinic"  structures (showing up out of the quantum dynamics by letting $\hbar\to 0$) per se, that is the limit dynamics lives on them, and not only the traces we just mentioned.
 
 \vskip 1cm
 In this sense, classical limit is as stable as quantum mechanics, at the condition of letting it live on the right structures, different from the ones of the ``underlying classical dynamics" and keeping track of some residual noncommutativity. This limit ``space" is not an absolute one as the one of Newton (and Kant!), it depends intrinsically of the dynamics itself.
 
 \vskip 1cm

 The space doesn't host the dynamics any more, it is the dynamics which hosts the space.
  
 \vskip 1cm
 
 \hfill The dynamics comes before the space as  \textit{l'existence pr\'ed\`ede l'essence} (once again).
 
 \vskip 1cm

 %%%%%%%%%%%%%%%%%%%%%%%%%%%%%%%%%%%%pouralex8%%%%%%%%%%%%%%%%%%
\subsection{Miscellaneous}\label{misc}\ 
 
 Section \ref{homo} somehow synthesizes several aspects already presented in te preceding sections. Several remarks are in order.
 
 On the quantum side, an extension of the standard quantization procedure- in fact the Töplitz one - is necessary in order to reach the  limit $\hbar\to 0$, or at least some uniformity in $\hbar\sim 0$.This true not only for situations involving singularities or chaotic behaviour, it is also the case for the non-resonant harmonic oscillator of every day's quantum mechanics.
 
 A the classical level, the extension of Toeplitz quantization lives on an extension of the regular phase space. The fibrated extension from configuration space to  phase space - adding a fiber over each position point taking care of the momentum dimension - is not enough any more: one needs to foliate the fiber bundle by adding over each point an invariant locus of the dynamics (invariant torus, stable and unstable manifolds). This addition is linked also to a (generalization of a) foliated structure on the invariant locus itself (ergodic flow on tori, homoclinic flows on  stable/unstable manifolds.
 
 These objects remain passive when the standard classical dynamics is involved, but become somehow active when the dynamics is inherited form a quantum one. And they do remain so if one consider now that the dynamics coming form quantum mechanics is the true classical one. a new classical paradigm where the underlying a priori chaotic behaviour is in a certain sense taken into account in the kinematics itself.
 
 Indeed, this microlocal phase-space$^{\otimes 2}$, microlocal near points (or trajectories) in the standard classical phase space, the same way that the phase space is itself microlocal near points of the configuration space, can be seen as constituted of \textit{objects} which can be the same way    both points or invariant manifolds.  
 
 And this is the key of the (strange) phenomenon of long time classical/quantum evolution
% 
% should b are taken account by extension from configuration to phase space) .
% 
% \item link with CRAS : a function becomes opertair valued a la Connes. More precisely, the symbol of the symbol-operator at a point ($p,q)$ is a function on 
% 
% \centerline{$stable_{(p,q)}\times unstable_{(p,q)}$.}
% \item 
% 
 when classical unpredictability and quantum indeterminism merge (see \cite{tp8} and the next section of the present article): the long time classical evolution bows-up the point into the unstable manifold - this is the unpredictability- exactly as the quantum evolution spreads a coherent space to a (WKB) Lagrangian state localized precisely on the unstable manifold - and then the indeterminism of the quantum paradigm, according to the  Born interpretation of the wave function \cite{mb}, reduces the points of the Lagrangian to an insignificant role purely probabilistic . If a point, according to Piero della Francesca \cite{pdf} is ``that part which is not", quantum paradigm 
 provide delocalized locus ``which are not". Assuming this, we believe that quantum mechanics becomes crystal clear.

 Let us finish this miscellaneous section by mentioning several directions for future extensions of the main ideas present in this homage to Alex Grosmann.
 
 We have seen the importance of Poincaré mappings - both classical and quantum - in our construction. The Poincaré mapping is the reduction to a Poincaré section of a (continuous time) flow. Likewise, dynamical systems - such has linear mappings on tori - have a suspension which let them become a Poincaré mapping themselves. We believe that our construction furnishes an exact geometrical  construction of the spectrum and eigenvectors of the quantization of the Arnold" s cat, a quantum system both chaotic and explicitly solvable.
 
We also think that noncommutative constructions close to the one in the present section of our paper could provide existence of a  time operator, existence excluded by the famous Pauli Theorem which remarks that no self-adjoint operator $T$ can satisfy the commutation relation $[H,T]=i\hbar$ when $H$ is bounded from below\footnote{Let us recall that the formal argument is very 
 simple, if $H\psi_{\lambda_k}=\lambda_k\psi_{\lambda_k}$ and $[H,T]=i\hbar$, then, for all $\lambda\in\bR$, 
 $e^{-i\frac{\lambda T}\hbar}\psi_{\lambda_k}=(\lambda_k-\lambda)\psi_{\lambda_k}$.} (or from above): 
 the fact of considering frequencies instead of eigenvalues already cancels the difficulty of semi-boundedness.
% 
%% \item return to the algebra of the foliation  quotient by time evolution: does it determine the set of periods ?
%% 
%% \item 

 Finally, a major achievement in quantum physics of the last fifty years concerns the understanding the phenomenon of statistics of level spacing: a ``trace" of chaotic (or not) behaviour of the underlying classical mechanics exists apparently in the probabilistic distribution of the spacing between eigenvalues: our approach 
links directly frequencies - i.e. differences of eigenvalues - 
to the geometry of the underlying dynamics and the properties of its invariants.

% 
% 
% 
% 
% \begin{itemize}
% \item link with the cat : the torus is the(reduced) phase space, on which there are dense periodic trajectories of the quantum flow, that is eigenspaces. Outside of this, the flow is hyperbolic. These torus is a phase-space. 
% 
% \begin{center}
% 
% %
% THE EQUIVAENT FOR CHAOTIC SYSTEMS IS THAT : 
% 
% THE SET OF HOMOCLINIC = LAGRANGIAN TORI, 
% 
% THE PHASE-SPACE (I.E. TORUS FOR THE CAT) IS (GENERATED BY) THE SET OF STAB/UNSTABLE MANIFOLDS 
% 
% \end{center}
% 
%  Everythngs infenitisimally local aronud each trajectory. One has to think a bout again to the suspension of the cat.
%  %LE/UNSTABLE MANIFOLDS.
% 
% \item this micro, local phase-space near points (or trajectories) should b are taken account by extension from configuration to phase space) .
% 
% \item link with CRAS : a function becomes opertair valued a la Connes. More precisely, the symbol of the symbol-operator at a point ($p,q)$ is a function on 
% 
% \centerline{$stable_{(p,q)}\times unstable_{(p,q)}$.}
% \item when classical impredictability and quantum indeterminism merge
% \vskip 1cm
% \noindent More subjects:
% \ 
% 
% 
% \item Pauli and the oprtaor time
% \item return to the algebra of the foliation  quotient by time evolution: does it determine the set of periods ?
% 
% \item statistics of level spacing
% \end{itemize}

 \subsection{Conclusion: the quotient of the phase space by the flow}\label{quoflow}\ 
 
 In this section we have seen that there is, in the chaotic situation, an equivalent  to the tori fibration for integrable systems: this is the (noncommutative) space of the `foliation" of phase-space by the ``homoclinic leaves", which consist in the quotient of the set of trajectories by  the equivalence relation ``being homoclinic". 
 
 A leaf is the set of all homoclinic trajectories to a given point $(p,q)$ or trajectory $\gamma$. Each of these ``Lagrangian" leaves  lives in an extended phase-space generated by the stable and unstable manifold at $(p,q)$ or $\gamma$. Each point outside this homoclinic tangle has an hyperbolic  behaviour when evolving in time.
 
 In another, but close, point of view, there is a discrete set of frequencies associated to  each homoclinic orbit in the homoclinic ``Lagrangian torus". The smallest of these frequencies are not orbit dependent, and so are that way  associated to the whole homoclinic Lagrangian itself. The frequencies become dense as $\hbar\to 0$ (on the contrary of the integrable case) revealing a kind of continuous spectrum : this is the way chaotic hyperbolic dynamics appears.
 
 \section{Indeterminism versus unpredictability (how quantum indeterminism would have chocked Laplace but not Poincaré)}\label{imvrsin}\ 

%\red{AFAIREindetermisnime quantique choquant par arapport a Laplace mais pas a Poincar'e
%}

This last section will be concerned by the indeterministic feature of quantum mechanics, a feature usually considered as characteristic of the quantum paradigm. It could seem a priori out of context in a paper mainly concerned with space   to look at the indeterministic properties  of the process of quantum measurements. Measurement is a part of the quantum dynamics and hence belong more to time considerations. Beside, probabilistic features of quantum mechanics are often considered as disappearing at the classical limit because macroscopic sizes `forces" quantum (intrinsic) probabilities laws to become if not binary at least classical (that is witness of a lack of knowledge).

On the contrary we will show that quantum indeterminism merge, at the classical limit, the concept of unpredictability by the delocalization of the wave functions precisely on the locus emphasized by  the chaotic sensitivity to initial conditions: the unstable manifold.

%Let us finally remark
Our conclusion will be, to summarize, that if the paradigm of  quantum mechanics  differs from the classical one, it is with the Laplace one \cite[p.2]{lap} and let us remind that the ``démon de Laplace" was already destroyed at the end of XIXth. century by the discovery of chaos in the works \cite{poinc} of the ``démon \poinc",

Dont act.
\subsection{Measurement in quantum mechanics}\label{meas}
Expressed in the most synthetic way, quantum dynamics ``à la Copenhagen" is a succession (a combination) of two different types of evolution for elements of a Hilbert space $\cH$:

-  Heisenberg-Schrödinger  flows  associated to  Hamiltonian self-adjoint operators $H$   driven by the unitary operators $e^{i\frac{tH}\hbar}$: namely, it reads  infinitesimally 
$$\
\ \ \ \ \ \ \ \ \ \psi(t)\to\psi(t+\delta t)=e^{i\frac{\delta tH}\hbar}\psi(t).
$$

- measurement processes associated to self-adjoint observable $O$ driven by random  mappings sending  instantaneously the state $\psi(t)$ to any eigenvector $\psi_k$ of $O$: 

$$\
\psi(t)\to\psi(t+\delta t)=\psi_k\ \mbox{ for some random $k$,}
$$
the randomness being associated to the probability law $\cP_{\psi(t)}(k)=|\langle\psi(t)|\psi_k\rangle|^2$.

 It is important to realize that nowadays (as it is the case in the Copenhagen axiomatic) measurement is  a true part of the quantum evolution - for example one can realize nowadays experimentally evolutions only driven by successive measurements, each one being ``decided" by the result on the preceding one\footnote{cf. Serge Haroche, Lectures at Collège de France}. %\red{AFAIREtrouver reference}
 
 One the other side, it is also important to admit that measurement in not part of the first type of evolution: even in the macroscopic limit of large number of involved particles  decoherence phenomena does not explain the randomness of a single measurement process \cite{har}. It is an intrinsic probabilistic event, not a consequence of luck of precise information on the system as in other domains of classical physics, such a statistical physics. It could seem therefore totally reconcilable with the paradigm of classical mechanics, but we will see that in fact that this indeterminism, at the classical limit, will in a certain sense fill a hole existing in the (comfortable but questionable) completeness of classical dynamics.
 
 Let us finish this section by saying that the loose of determinism  while passing from classical to quantum is accomplished together with the creation of tools permitting the (in particular experimental) realization of the `measurement"axiom, namely the discreteness of spectra. Indeed, the probabilistic meaning of the distribution $\cP_{\psi^{in}}(k)$ is that, if one realize the same measurement process, that is with the same  observable $O$ {\bf and the same initial vector $\psi^{in}$}, a large (infinite) number of times, although the different results $\psi_k, k=0,\dots,\infty$ will be a each time intrinsically not possible to determine,  the statistical distribution of values of k will be given by $\cP_{\psi^{in}}(k)$. Therefore the difficulty is to be sure each measurement is effectuated on the same initial vector $\psi^{in}$. We are touching here a notion of precision  which is a subject very sensitive in classical mechanics due to the continuous feature of geometry of spaces, but which is much easier to handle in quantum mechanics just because any 
 %AFAIRE (other than $)$)  
 observable furnishes a \textit{discrete} set of vectors: its own eigenvectors (see e.g. \cite{tp10} for further details). %\red{AFAIRE rafiner un peu}

\subsection{Critics of the deterministic reason of classical mechanics}\label{crit}
%\subsection{Measurement in quantum mechanics}\label{meas}
 The truly fabulous  change of paradigm that \poinc introduces at the end of the XIXth century, much before the birth of quantum mechanics, forces us to revise, at the light of quantum  considerations such as those exposed in the last paragraph of the preceding section \ref{meas}, the notion of measurement in classical mechanics (and also in computer sciences, see \cite{lp1,lp2,lp3}): can we really talk  about determinism when the computations of Jacques Laskar show that a precision of a few meters for the knowledge of the position of the earth is needed to prove that it will stay in the solar system in a few millions of years \cite{jl}? 
 
 The phenomenon of ``sensibility to initial conditions" for chaotic dynamical systems  do not only destroy the Laplacian paradigm which would suggest that systems close to integrable are still integrable, their modern formulation poses the question of space after infinite time  evolution. Of course, in principal, time remains finite in classical dynamics\footnote{``Eternity is very long, specially at the end!", P. Desproge} but the theory of chaos needs ontologically such an infinite extension of time. For example, ergodicisolty involve averages on infinite times. But chaoticity features avoid carefully to define evolution at times $t=\infty$: one cooks up some quantities (averages etc) defined at finite times the limit of which at $t\to\infty$ are considered. The reason for this careful prevention is very simple: the classical flow $\Phi^t$ at time $t=\infty$ does not exist. More precisely it escapes to the classical paradigm: this is the hole in completeness mentioned earlier. 
 
 And this escape is a question of space: flow at infinite time escapes from the  classical paradigm because it does not preserve the geometrical structure of points in space: after infinite evolution, single points would ``become"  extended objects as manifolds.
 \subsection{Pushing sensitivity to initial conditions to its extreme}\label{pussenic}
 The following diagram  expresses the phenomenon we just described. 
 
 Let us consider an hyperbolic fixed point $y$ of an Anosov (say) dynamical system. Associated to it is an unstable manifold $\Lambda_y$, set of all the points $x$ ``attracted" to $y$ after a long backwards  evolution 
 %\red{AFAIRE revoir signes, stabel unstables etc}: 
 $$
 \Lambda_y=\{x,\lim_{t\to+\infty}\Phi^{-t}(x)=y\}.
 $$
% \red{revoir signes, stabel unstables etc}
% 
 Therefore, one has the following series of facts, (naively) logically linked as follows.
 
 \vskip 1cm
\begin{eqnarray}
\forall x\in\Lambda_y ,\ \ \ \ \ \Phi^{-t}(x)&\overset{
%\mbox{ as }
t\to+\infty}\longrightarrow&y,\ \ \nonumber\\
&&\nonumber\\
&\Updownarrow&\nonumber\\
&&\nonumber\\
\big(\forall x\in\Lambda_y ,\ \ \ \ \ \Phi^{-\infty}(x)&=&y\big)\  
%\forall x\in\Lambda_y
\ \ \ \ \ \ \ \ \ \ \ \ \ \ \ \ \  \Leftrightarrow\ \ \ \Phi^{-\infty}(\Lambda_y)=\{y\}
\nonumber\\
&&\nonumber\\
&\Updownarrow&\nonumber\\
&&\nonumber\\
\big(\Phi^{+\infty}(y)&=&x,\  \forall x\in\Lambda_y\big)
\ \ \ \ \Leftrightarrow\ \ \ \Phi^{+\infty}(y)=\Lambda_y\nonumber\\
&&\nonumber\\
%&\Updownarrow&\nonumber\\
%&&\nonumber\\
%\Phi^{-\infty}(\Lambda_y)&=&\{x\}  \forall x\in\Lambda_y\nonumber\\
%&&\nonumber\\
&\Updownarrow&\nonumber\\
&&\nonumber\\
\mbox{sensitivity to initial conditions}&\sim&  
%\mbox{indeterminism}
\mbox{impredictibility}.\nonumber
\end{eqnarray}  

Usually, the idealized version of the  sensitivity to initial conditions expressed bu ``$\Phi^{+\infty}(y)=\Lambda_y$ is indeed considered as an expression of impredictibility: one cannot predict the position of the point $y$ after and infinite time evolution as we cannot answer the question of the presence or not of the earth inside the solar systems in a quasi-infinite time (a few millions of years) without knowing its position with a quasi-infinite  precision (a few meters).

But the two last line of the diagram could be perfectly changed into

\begin{eqnarray}
{\color{white}
\big(\forall x\in\Lambda_y ,\ \ \ \ \ \Phi^{-\infty}(x)}&{\color{white}=}&{\color{white}y\big)\  
%\forall x\in\Lambda_y
\ \ \ \ \ \ \ \ \ \ \ \ \ \ \ \ \  \Leftrightarrow\ \ \ \Phi^{-\infty}(\Lambda_y)=\{y\}}
\nonumber\\
\big(\Phi^{+\infty}(y)&=&x,\  \forall x\in\Lambda_y\big)
\ \ \ \ \Leftrightarrow\ \ \ \Phi^{+\infty}(y)=\Lambda_y\nonumber\\
%&&\nonumber\\
%&\Updownarrow&\nonumber\\
%&&\nonumber\\
%\Phi^{-\infty}(\Lambda_y)&=&\{x\}  \forall x\in\Lambda_y\nonumber\\
%&&\nonumber\\
&&\nonumber\\
&\Updownarrow&\nonumber\\
&&\nonumber\\
\mbox{sensitivity to initial conditions}&\sim&  
\mbox{indeterminism}.
%\mbox{impredictibility}
\nonumber
\end{eqnarray}
This is exactly what we say by ``{\it we cannot answer the question of the presence or not of the earth in the solar systems in a few millions of years}".

We will see that this vision of impredictibility/indeterminism is conform to the heritage of the quantum indeterminism at the classical limit $\hbar\to 9$.

\subsection{Some space for merging the two}\label{merg}\ 

The so-called probabilistic interpretation of the wave function $\psi$ by Born \cite{mb} tells that the square modulus $|\psi(x)|^2$ gives the probability of ``finding" an electron at the place $x$. The link with $\cP_\psi$ introduced in Section \ref{meas} is the following. the observable ``position" is the operator $Q$ of multiplication by the variable
$$
Q:\ \psi\in\cH=L^2(\bR).\to Q\psi\mbox{ with }(Q\psi)(x)=x\psi(x).
$$

The spectrum of $Q$ is $\bR$ and the (generalized) eigenvectors are $\delta_x,\ x\in\bR$, 
the Dirac mass at $x$\footnote{defined by $\int \delta_x(y)\varphi(y)dy=\varphi(x)$ for every (test) function $\varphi$.}
so that the result of the measure of the position can be any real number $x$, and the probability of ``finding" the value $x$ by measuring the state $\psi$ is
$$
|\langle\psi|\delta_x\rangle|^2=
|\int_\bR \delta_x(x')\psi(x')dx'||^2=|\psi(x)|^2 \ \ \ \ \ \mbox{ CQFD}.
$$
Let us consider now the Hamiltonian generator of the dilations (we denote $P=-i\frac d{dx}$)
$$
H=\frac{QP+PQ}2:=-ix\frac d{dx}-\frac i2.
$$
Obviously 
$$
e^{i\frac{tH}\hbar}\psi(x)=e^{\frac t2}\psi(e^tx).
$$

The underlying classical dynamics on $\bR^2=(T^*\bR)_{(q,p)}$ is generated by the Hamiltonian $h(q,p)=qp$, which leads to the flow
$$
\Phi^t(q,p)=(e^tq,e^{-t}p).
$$
so that we see immediately that
$$
\forall t\in\bR,\ \Phi^t(0,0)=(0,0)\ \ \ \mbox{ and }\ \ \ 
\lim_{t\to+\infty}\Phi^{-t}(q,0)=(0,0),\ \forall q\in\bR.
$$

Therefore the origin of the phase space $T^*\R$ is a fixed point, and the null section of $T^*\bR$ is its unstable manifold $\Lambda_{(0,0)}$, a situation exactly comparable to the one in Section \ref{crit}.

We have seen in several places of this article that a natural way of associating to a point in phase space a vector in the quantum Hilbert space was obtained thanks to the Gaussian coherent state, as defined for example by \eqref{defcoh} $\vp_{(p,q)}(x):=(\pi\hbar)^{-\frac{1}4}e^{-\frac{(x-q)^2}{2\hbar}}e^{i\frac{px}\hbar}$. Pinned up at the origin, this definition leads to
$$
\psi(x)=\vp_{(0,0)}(x)=(\pi\hbar)^{-\frac{1}4}e^{-\frac{x^2}{2\hbar}}
$$
well localized to $0$ as $\hbar\to 0$.

We get
$$
e^{-i\frac{tH}\hbar}\psi(x)=e^{-\frac t2}\psi(e^{-t}x)
=
e^{-\frac t2}(\pi\hbar)^{-\frac{1}4}e^{-e^{-2t}\frac{x^2}{2\hbar}}.
$$
so that, when $t=t_\hbar=-\frac{\log\hbar}2$,
$$
e^{-i\frac{t_\hbar H}\hbar}\psi(x)
%=e^{-\frac t2}\psi(e^{-t}x)
=
(\pi)^{-\frac{1}4}e^{-\frac{x^2}{2}}
$$
which is fully delocalized in $x$ on the real line which is, let us recall, $\Lambda_{(0,0)}$. In other words, $\{0\}\overset{\Phi^{-\infty}}\rightarrow\Lambda_{(0,0)}$ or $\Phi^{-\infty}(0)=\Lambda_{(0,0)}$ exactly as in Section \ref{pussenic}.

Quantum indeterminism will read now: any point $x\in\Lambda_{(0,0)}$ can be obtained by measuring the position on the state $e^{-i\frac{t_\hbar H}\hbar}\psi$. A statement very close to the impredictibility/indeterminism's one of the last section.
\vskip 1cm
Therefore, when $\hbar\to 0$, quantum indeterminism and classical impredictibility merge.
\vskip 1cm
%\red{AFAIRE rajouter du texte sur les proba inheritated form quantu mechancis}
Even more: the probabilities emerging in classical mechanics because of (or thanks to) impredictibility are fully inherited from quantum mechanics.
\subsection{A classical phase space incorporating the point $t=\pm\infty$}\label{incotinf}
What we see on this very simple example (which can be easily generalized to more complex dynamics) is that the fact of having abandoned in quantum mechanics the classical concepts of points and trajectories, a fact which seems to create problems to a lot of persons attached to the old paradigm(s), is in fact a richness. It allows to handle impossible situations in classical dynamics. Let us finish this section by presenting another situation where this phenomenon appears.

The famous Cauchy-Lipschitz Theorem provide a sufficient condition for existence and uniqueness of solutions to differential equations associated to vector fields: $\dot x=f(x)$. Lipschitz regularity of $f$ ensures existence and uniqueness of the solution. Otherwise, for less regular $f$s, the flow might develop ubiquity phenomena incompatible with the classical paradigm underlying the concept of space.

Quantum mechanics is more global than local, therefore boundedness issues are more stringent than regularity ones. Situation with regularity ``under Cauchy-Lipschitz" can be perfectly reasonable, and provide perfectly well defined quantum mechanics\footnote{A nice  example showing this classical/quantum dichotomy is the case of the Coulomb potential: although the classical Hamiltonian $p^2-\frac1{|r|}$ is not bounded from below, the quantum one $-\Delta-1/|r|$ is. }, leading, for example, to non ambiguous quantum dynamics. But this quantum trajectories live in the Hilbert space of quantum states, and have few to share with the classical underlying space, except when $\hbar\to 0$. And one can easily cook up situation where a coherent state, classically associated to a single point in phase space can evolve as a superposition of two, localized on two different points. But this ubiquity  disappears as soon as  a measurement is accomplished, giving rise to a single result, a random one but surely one and only one \cite{tp11,ap1,ap2,ap3}.

%\red{AFAIRE 
%DIRE QU'ON A COMLETE L' INCOMPLETUDE
%}
{\bf In fact, quantum mechanics completes the incompleteness of classical mechanics. In a certain sense it was born for this purpose.}
 \section*{Postlude: 
% Species of spaces,
 %\centerline
 {\it 
 species of spaces, towards an epistemologigal geometry of quantum/classical mechanics}
 }\label{specspac}

 \vskip 1cm
 
 In this paper, we have stated several situation of quantum mechanics corresponding to the limit where the Planck constant $\hbar$ vanishes. As mentioned at the beginning of this paper, quantum mechanism was born deeply inside classical paradigm, after a strong epistemological jump operated by Heisenberg: a jump consisting by taking $\hbar$ as a serious parameter of physics and passing from the classical world where $\hbar=0$ to the quantum one where $\hbar>0$. 
 
 \centerline{In other words, letting, somehow, $0\to\hbar$.}
 \vskip 0.3cm
 The change of paradigm between $\hbar=0$ and $\hbar>0$ is as deep as passing from the culturally classical notion of a function, which assign to 
 a variable a {\bf given} value taken into the range of the function,  to the 
 notion of matrix which acts on a ``variable", namely a vector, by spreading it 
 %on
  all over the spectrum of the matrix. Seen like this, the two paradigms seem irreconcilable. And  indeed they are, strictu senso, if we miss the fact that there exist special vectors, namely the eigenvectors of the matrix, to which the matrix assign a single value taken into its spectrum : the corresponding eigenvalue. This remark made the makers of the foundations of quantum mechanics, say Born-Heisenberg-von Neumann\footnote{We should mention at this point two alternative points of view: the one by de Brogglie (right before Heisenberg) and Schrödinger (right after Heisenberg) making quantum mechanics a wave theory. However, one has to asy that these two viewpoints, the first one severely incomplete and the second shown quickly to be equivalent to Heisnberg's one, though taking place in a maybe more classical setting, don't  by no mean reconciliate classical determinism and quantum determinism.},  able to finish the quantum axiomatic by imposing that any eigenvalue could be assigned {\it randomly} to a general vector while measuring it. The addition of a property of the probability law insuring that eigenvectors and eigenvalues are deterministically linked makes the whole construction consistent.
  
  This view of quantum measurement, milestone of the quantum paradigm, inseparable  from the rest of the axiomatic,  is sometime seen as a somehow  inverse of the Heisenberg jump, it is in a way a jump itself, as brutal as the original one. But this point of view is not correct. For example, coupled with a strong need of determinism in the classical theory, it creates the so-called paradoxes of quantum mechanics, so-called but fake as they are not paradoxes {\bf inside} the quantum paradigm. They show some true contradictions but in a thinking dynamics `` \`a cheval" between quantum and classical: an impossible world in which one desperately try to marry two inconciliable points of view.
  
  {\bf Inside}, the word is pronounced: measurement lives inside the quantum paradigm, strictly inside, although  the classical world can be only recover by inverting $0\to\hbar$ by $\hbar\to 0$\footnote{Giving a precise discussion about what $\hbar\to 0$ means is far beyond the scope of the present paper. Let us say that, under the  proposition $\hbar\to 0$ is hidden numerous situations where the different parameters of the system studied tend to extreme value leading to the famous canonical relations $[Q,P]=i\hbar$ tend to $[Q,P]=0$}. The classical limit plays the role for quantum mechanics that plays ``$)$" for the interval $(0,1]$: it is its border.
 \vskip 0.2cm

\centerline{\it CLASSICAL LIMIT = BORDER OF QUANTUM}

\vskip 0.1cm
 The standard way of interpreting this last proposition is to say that  %CLASSICAL = BORDER OF QUANTUM, b
 the classical world is the border of the quantum one. But in the present article, we have exhibit several situation where the limit $\hbar\to 0$ does NOT correspond to the classical world, this was precisely the goal of this article. In fact a more true statement is
 
 \vskip 0.2cm

\centerline{\it CLASSICAL $\subset$ BORDER OF QUANTUM}

\vskip 0.1cm
 
 But this doesn't finish the story. Indeed, on the contrary of one usually thinks, the classical paradigm is far from being exempt of singularities and incompleteness. Newton's law of inverse dependence of the force  in the square of the distance between attracted objects doesn't tell you a single thing when this distance tends to $0$. And one knows since Poincaré that, as the time $t$ of evolution  diverges, unpredictability is a serious challenge for the so culturally well installed determinism, actually linked to quantum indeterminism, as shown in the preceding section - see also \cite{tp8,tp9}.
 
  What we claim is that classical has itself a border, for example when $t\to\infty$, and we proved in \cite{tp9} that the sensitivity to initial conditions in chaotic systems let this long time limit merge with the corresponding one for the corresponding quantum evolution: sensitivity to initial conditions create, at the limit $t\to 0$, a nonlocality feature which is intrinsic to quantum paradigm  and let classical unpredictability and quantum indetrminism merge. In other words, at list terminologically, 
  
  \vskip 0.2cm

\centerline{\it QUANTUM $\subset$ BORDER OF CLASSICAL}

\vskip 0.1cm
 So that:
 %\Large
 \vskip 0.2cm

\centerline{\it QUANTUM $\subset$ BORDER OF CLASSICAL
$\subset$ 
BORDER OF 
$\dots$ 
BORDER OF QUANTUM
}

\vskip 0.1cm 
  %\LARGE
  and in other, final, `` à la Oulipo"  words:
  %\small 
  \vskip 0.5cm

\centerline{\bf \textit{QUANTUM $\subset$ {\large(THE BORDER OF)} ITS OWN BORDER} }

\vskip 0.39cm
  %\LARGE
 What a beautiful geometry, with strange species of strange spaces !

\hfill{\it es mu$\beta$ sein}
\vskip 3cm

\end{document}